\documentclass[unnumsec,webpdf,modern,medium]{template}

\onecolumn
\usepackage{fancyhdr,amsmath,amssymb, amsthm,mathtools,dsfont}
\usepackage{amsfonts}
\usepackage{bbm,bm}
\usepackage{graphicx}
\usepackage{algorithm}
\usepackage[export]{adjustbox}
\usepackage[british]{babel}

\theoremstyle{thmstyleone}
\newtheorem{theorem}{Theorem}
\newtheorem{proposition}[theorem]{Proposition}
\theoremstyle{thmstyletwo}

\newtheorem{remark}{Remark}
\theoremstyle{thmstylethree}

\newcommand{\D}{{\rm d}}
\newcommand{\iid}{\stackrel{\mbox{\scriptsize iid}}{\sim}}
\newcommand{\ind}{\stackrel{\mbox{\scriptsize ind}}{\sim}}
\newcommand{\myequal}[1]{\mathrel{\overset{\makebox[0pt]{\mbox{\normalfont\tiny\sffamily #1}}}{=}}}
\newcommand{\indic}{\mathbb{I}}

\renewcommand{\P}{\mathbb{P}} 

\newcommand{\E}{E}

\renewcommand{\P}{\mathbb{P}} 
\newcommand{\prob}[1]{\P\left(#1\right)}

\newcommand{\mean}[1]{\E\left[\,#1\right]}

\newcommand{\pityor}{\operatorname{PYP}}
\newcommand{\Mmax}{\operatorname{M}_{\operatorname{max}}}
\newcommand{\Mr}{\operatorname{M}_{r} }
\newcommand{\Mrr}{\operatorname{M}_{1} }

\newcommand{\Fcr}{\mathscr{F}}
\newcommand{\Lcr}{\mathscr{L}}
\newcommand{\Pcr}{\mathscr{P}}

\newcommand{\muprime}{\mu^\prime}

\usepackage{xcolor}

\definecolor{RevA}{HTML}{B00020} 

\newif\ifshowrevs
\showrevstrue

\begin{document}
\pagestyle{plain}
\copyrightyear{}
\pubyear{}
\appnotes{$_.$}

\firstpage{1}

\setcounter{secnumdepth}{4}

\fontsize{9.9}{13.3}\selectfont

\title[Bayesian nonparametric inference for \\ modal missing species and features]{Bayesian nonparametric inference for  \\ modal missing species and features}

\author[1]{Alessandro Colombi}
\author[2]{Mario Beraha}
\author[1]{Daniele Durante}
\author[3]{Stefano Favaro}

\authormark{Colombi, Beraha, Durante and Favaro}

\address[1]{\orgdiv{Department of Decision Sciences and Bocconi Institute for Data Science and Analytics}, \orgname{Bocconi University},  \country{IT}}
\address[2]{\orgdiv{Department of Economics, Management and Statistics},  \orgname{University of Milano-Bicocca},  \country{IT}}
\address[3]{\orgdiv{Department of Economics and Statistics}, \orgname{University of Torino and
Collegio Carlo Alberto},  \country{IT}}

\corresp[]{Correspondence: Alessandro Colombi, Department of Decision Sciences and Bocconi Institute for Data Science and Analytics, Bocconi University, IT. Email: \href{email: alessandro.colombi@unibocconi.it}{alessandro.colombi@unibocconi.it}
\vspace{-10pt}
}

\abstract{
Species and feature sampling problems arise naturally whenever each observed unit is associated with one~or~more labels from a countable alphabet, and inference focuses on the unobserved portion of the distribution over such an alphabet. Within this framework, recent contributions have shifted attention from inference on~the~total~probability mass over unseen labels to distribution-free confidence intervals for the largest unobserved label probability (i.e., the modal missing probability), thereby providing more refined information on whether~the~missing mass is concentrated on a few high-prevalence unseen labels or distributed across many negligible ones. Besides lacking~a unified framework for species and features, these approaches employ a worst-case perspective that causes substantial information loss and produces overly-conservative intervals. We address~these~limitations~through~a unified model-based framework for inference on modal missing probabilities in both species and feature settings, which leverages a flexible Bayesian nonparametric formulation to localize uncertainty around models compatible with the observed data. This leads to sharper closed-form credible intervals that effectively exploit prior information and observed data, while preserving the theoretical frequentist properties and robustness of distribution-free intervals. Simulation studies confirm these improvements, while an organized crime application illustrates how our contribution has the potential to reshape law-enforcement decision-making in investigations.
}

\maketitle

\fontsize{9.9}{17.2}\selectfont

\section{\large 1. Introduction}\label{sec:intro}

\vspace{8pt}
\subsection{\large 1.1 Background and motivation}\label{sec_1.1i}
Inferring the unseen portion of a distribution from~a finite sample $X^n=(X_1,\ldots,X_n)$~over~a~countable alphabet $\Xi=\{\xi_1,\xi_2,\ldots\}$  is a classical, yet central, statistical problem in several scientific domains. Within this framework, considerable attention has been devoted to the total probability assigned to the unobserved labels, the so-called \textit{missing mass} \citep[][]{goodturing53,GoodToulmin}. Such an overarching focus is motivated by the direct predictive interpretation~of~the missing mass as the probability that the $(n+1)$th draw from the alphabet yields a label~that~has~not been observed in the current sample $X^n$. This central predictive interpretation has stimulated not only substantial methodological advancements \citep[e.g.,][]{bunfi93,Lij07,favaro2012,Orlitsky2015,oha19,aye21,BalocchiFavaro24} and successful applications across several important  fields \citep[e.g.,][]{efron76,Moritz2008,Mao2004,Eckburg2005,Bub13,Dol23,ane17,Noorani2026}, but also recent extensions aimed at studying alternative functionals of the unseen part of the distribution, which provide a finer understanding of what remains unobserved \citep{Chandra2024,naulet24,Painsky24,ColombiFreq}. In fact,~although the missing mass remains of central interest, such a quantity summarizes the unseen portion~of~the~distribution solely through its total probability, without revealing whether such a probability is concentrated~on~a single relevant unobserved label or distributed across several individually negligible ones. Distinguishing between these two scenarios is central to large-alphabet inference and to the design of sequential discovery procedures and stopping rules, where the goal is to determine whether the unobserved~portion of the distribution  contains a high prevalence unseen label. As illustrated~by~the~motivating~criminology application in Section~\ref{sec:appl}, this perspective can reveal, for example, whether~the yet-unobserved portion of a criminal organization contains highly-populated, and therefore operationally significant, unseen operative branches or only negligible ones \citep[see, e.g.,][]{diviak2022key}. The former case would justify further investigative efforts, while the latter  suggests that additional investigations are unlikely to uncover unobserved branches of practical relevance. As such, in the specific criminology setting,~this perspective has the potential~to fundamentally reshape law-enforcement decision-making by providing a principled framework for assessing expected gains~of~investigative~efforts.

Addressing the above objectives requires shifting the focus from the total missing mass~to~the largest probability assigned to an unseen label, which we term \emph{modal missing probability}. Despite its relevance in various scientific domains, this quantity has been formalized and studied~only recently in a limited number of contributions  developing one-sided confidence~intervals in both species \citep{Painsky24} and features \citep{ColombiFreq} sampling problems~(see Section~\ref{s12}~for details on these two settings). While these contributions represent an important step beyond~the classical missing mass problem, the resulting intervals are derived under a worst-case scenario analysis over the entire class of distributions on the alphabet. This yields solutions~that~are~either data-independent or rely only on very coarse summaries of the sample, thereby~leading~to~overly-conservative intervals and a substantial loss of information. Furthermore, the modal missing probability problem is addressed separately for the species and features settings, with no attempt to provide a unified~framework. As we clarify within Section~\ref{s12}, such a framework is not only conceptually natural, but would also offer key methodological and theoretical advantages.

Motivated by the above discussion, we develop in Sections~\ref{sec_2_mod}--\ref{sec_3_conf} a unified framework for inference on the~modal missing probability in both species and features settings, replacing the worst-case scenario approach with a model-based solution that quantifies uncertainty under a structured,~yet flexible, characterization of the underlying distribution over the alphabet  $\Xi$. This is achieved by eliciting~suitable Bayesian nonparametric priors \citep[see, e.g.,][]{hjort2010bayesian} over such a distribution, which  can incorporate structural information on the underlying probabilities, while yielding posteriors that  localize uncertainty around models compatible with the observed~data.~As~a~consequence, such a model-based perspective results in more informative one-sided credible intervals for the modal missing probability that effectively exploit the observed data, while preserving the valid frequentist coverage and overall robustness of those obtained under worst-case scenario~analyses. Further details on our unified treatment of modal missing probabilities under a model-based Bayesian nonparametric perspective can be found within Sections~\ref{s12}~and \ref{s13}, respectively. Section~\ref{s14} introduces, instead, the motivating organized crime application. Besides demonstrating the practical advantages of the proposed approach over state-of-the-art alternatives, this~application~should also be regarded as a stand-alone contribution to an area in which the modal missing probability problem has largely been overlooked, despite its potential to support law enforcement agencies in improving investigative strategies \citep[][]{campana2022studying,diviak2022key}. Proofs can be found in the Supplementary Material.

\vspace{11pt}
\subsection{\large 1.2 A unified treatment of the modal missing probability in species and features settings}\label{s12}
In developing a unified model-based treatment of the modal missing probability problem for both species and features settings, we first recall that designing an appropriate model for~the~unseen portion of the population depends fundamentally on the type of information that the sample $X^n$ carries about the alphabet $\Xi$. In the \emph{species} setting, each observation corresponds to exactly one label in $\Xi$. Thus, if $X_i\in\Xi$, the quantities $p_j := \mathbb P(X_i=\xi_j)$, $j\geq 1$, define a probability distribution on $\Xi$. Such a setting covers, for instance, sampling individual animals by species, recording surnames in a population, or observing the word appearing at a given position~in~a~text.~By~contrast,~in~the \emph{features} setting, each observation may  simultaneously exhibit several labels from $\Xi$. In~this case, $X_i$ is naturally identified with a finite subset of $\Xi$. Hence, $p_j := \mathbb P(\xi_j\in X_i)$, $j\geq 1$, denote inclusion probabilities rather than the masses of a probability distribution on $\Xi$.~This~latter setting is appropriate for presence--absence data, such as taxa observed in environmental samples, document--term matrices in natural language processing, or actor--event networks in which each single event may involve several participants.

Although the structural differences between the two aforementioned settings necessitate distinct statistical models, the modal missing probability problem in both the species and features scenarios can be cast as the same inferential question within a unified framework. More specifically, let $n_j$ denote the number of times label $\xi_j$ is observed in the sample $X^n$, so that $n_j>0$ if $\xi_j$ has been observed, and $n_j=0$ otherwise. The modal missing probability is defined as
\begin{equation}\label{eqn:Mmax_def}
\Mmax\,\coloneqq\, \Mmax(X^n;p_j, j \geq 1)=\max\nolimits_{j:n_j=0} p_j,
\end{equation}
which corresponds to the largest probability assigned to a label that has not  yet appeared~in the observed sample. Definition \eqref{eqn:Mmax_def} clarifies that, although the species and feature  settings differ at the modeling level, the inferential target is shared among the two. It is therefore natural to seek statistical methods that can address the modal missing probability problem under these~two settings within a unified framework. This perspective is also useful in applications, where the same dataset may often be interpreted through either lens, depending on the scientific question of interest. For instance, recalling the motivating organized crime application~in~Section~\ref{sec:appl}, law enforcement is often interested in both previously unseen, yet highly populated, operative~branches, and also in missing criminals who actively attend summits of the organization \citep[][]{campana2022studying,diviak2022key}. These objectives can be recasted under a species and features setting, respectively, but the common focus is on quantifying whether~the unobserved part of the organization may~still contain branches and criminals of substantive importance, rather than negligible ones.

\vspace{20pt}
\subsection{\large 1.3 Model-based Bayesian nonparametric inference on the modal missing probability}\label{s13}
As anticipated in Section~\ref{sec_1.1i}, we address the modal missing probability problem under~both~species and features settings through a model-based Bayesian nonparametric perspective, which~incorporates meaningful, yet flexible, structure for providing more informative inferences relative~to those obtained by recent worst-case scenario analyses \citep{Painsky24, ColombiFreq}.

The above objective is achieved by assuming that the sample $X^n$,  $n\geq1$, is conditionally i.i.d. given an almost surely discrete random measure  $\mu=\sum_{j\geq 1}p_j\delta_{\xi_j}$ supported on $\Xi$, whose prior~is inherited from Bayesian nonparametrics. In particular, when the probabilities are expected to display geometric or power-law decay, we take $\mu$ to be a Pitman--Yor process in the species~setting \citep{PitmanYor97Annals}, and a three-parameter Beta process under the feature setting \citep{Teh09}. Conversely, when the distribution is believed to have a finite support, possibly with an unknown total number of labels, we adopt mixtures of Dirichlet distributions under the species setting \citep{deblasi2015,argiento2022annals} and mixtures of Beta distributions in the feature setting \citep{Ghilotti2026}. Besides providing popular constructions within the Bayesian nonparametric literature, these prior choices yield posteriors that achieve~a~flexible characterization of the underlying data-generating mechanism, while facilitating tractable derivation of a valid one-sided credible interval for $\Mmax$, namely  a credible upper bound such~that,~for~fixed level $\alpha\in(0,1)$, the posterior probability that $\Mmax$ exceeds this bound, given $X^n$, is at most $\alpha$. As~outlined in Section~\ref{sec_3_conf}, this is obtained by first  approximating the modal missing probability via suitable $\ell_r$-norms of the vector of unseen probabilities, referred to as $r$th-order missing~mass \citep{Chandra2024,naulet24}, and then deriving analytical expressions for the posterior means associated with these $r$th-order missing masses to ultimately obtain closed-form valid  credible intervals for $\Mmax$.

Technically, the above  $\ell_r$-norms framework is obtained by transferring analytical strategies~employed in worst-case scenario analyses \citep{Painsky24} to our model-based Bayesian nonparametrics solution of the modal missing probability problem. As clarified in Sections~\ref{sec_3_conf} and \ref{sec:comp_and_infer},~such a transfer of techniques is, however, not direct, but rather requires conceptual and methodological innovations. In fact, we do not compute  expectations under a worst-case data-generating~mechanism, but rather  average with respect to  a posterior distribution localizing uncertainty around probability structures which are compatible with the observed data. This yields conceptually and technically different credible intervals for $\Mmax$ replacing distribution-free guarantees of existing frequentist procedures with Bayesian adaptive model-based uncertainty quantification~that~results in sharper, and hence, more informative intervals. Crucially, as proved in Section~\ref{sec:freq_properties_all}~these intervals have, under mild assumptions, asymptotically valid frequentist coverage, and hence, can be interpreted also as asymptotic one-sided confidence intervals. To the best of our knowledge, this is the first Bayesian model-based approach to the construction of valid credible intervals~for~$\Mmax$.

\vspace{15pt}
\subsection{\large 1.4 Application to organized-crime data}\label{s14}
Besides providing methodological and theoretical advancements in inference on modal missing probabilities, our contribution also aims to clarify the substantial gains that result from applying this framework to law enforcement investigations of hard-to-untangle criminal organizations. As discussed in e.g.,  \citet{campana2022studying} and \citet{diviak2022key}, the covert nature inherent to organized crime implies substantial data incompleteness problems in law-enforcement investigations, which lead to a partial view of the whole illicit organization, both in terms of its operative branches and affiliates. Addressing this boundary specification problem is of paramount importance in criminology, since its solution could facilitate the design of stopping rules reconciling the trade-off between the attempt to achieve a comprehensive view of the criminal organization and the resources required to attain it. In Section~\ref{sec:appl}, we show that, although largely overlooked in these contexts, model-based inference on modal missing probabilities provides a natural, principled and effective solution to the boundary specification problem, and hence, to the design of stopping rules for assessing when the probability of discovering relevant portions of the organization has become too negligible to ultimately justify further investigative efforts.

The practical and conceptual advantages of the above perspective are illustrated in Section~\ref{sec:appl} with a focus on a Mafia-type organization (\textit{'Ndrangheta}) operating in the Milan area, Italy, whose activities were documented during a major law-enforcement~investigation~named~\textit{Operazione Infinito} \citep{Calderoni2017}. The judicial documents resulting from this operation contain,~among others, information on the \textit{'Ndrangheta} affiliates identified at different occasions during  investigations, together with the corresponding membership to the discovered branches of the organization, known as \textit{locali} (i.e., operative territorial subgroups). As~noted~in the judicial documents, despite the substantial efforts of \textit{Operazione Infinito}, an unseen portion of the illicit organization~is~still unavoidable, both in terms of relevant missing \textit{locali} and single affiliates. We address these two cases within the modal missing probability framework in the species and features settings, respectively. The former assesses whether, given current evidence, an unobserved, yet highly populated, \textit{locale} may still be discovered in future investigations,~while~the~latter quantifies  presence of previously unseen active affiliates who could appear on multiple occasions~if the investigation were to continue. As shown in Section~\ref{sec:appl}, the sharper intervals derived under the proposed model-based Bayesian nonparametric procedure provide a more informative quantification of the boundaries of the monitored organization, while enabling the design of effective stopping rules that substantially improve upon the overly conservative ones of worst-case analyses.

\vspace{25pt}
\section{\large 2. Bayesian Nonparametric Modeling of Species and Features}\label{sec_2_mod}

\vspace{9pt}
\subsection{\large 2.1 The species sampling setting}\label{sec:species}
Let $\Xi=\{\xi_1,\xi_2,\ldots\}$ denote a countable alphabet of labels, and let $X^n=(X_1,\ldots,X_n)$~be~a~sample from a discrete distribution $\mu=\sum_{j\geq 1}p_j\delta_{\xi_j}$ supported on $\Xi$, where $p_j=\prob{X_i=\xi_j}\geq 0$, and $\sum_{j\geq 1}p_j=1$. Since the sampling distribution is discrete, ties are expected to occur in the sample. We denote by \(K_n\) the number of distinct labels in $X^n$, and by $n_j$ the number of occurrences of label $\xi_j$, so that $\sum_{j\geq 1}n_j=n$. As anticipated in Section~\ref{sec:intro}, under these settings we employ~a~Bayesian model-based  perspective assuming
\begin{equation}\label{eqn:Definetti_model}
\begin{split}
(X_1,\ldots,X_n \mid \mu) &\iid \mu, \\
\mu &\sim \Pcr,
\end{split}
\end{equation}
where $\Pcr$ is a flexible prior distribution on $\mu$, chosen within a general class of Bayesian nonparametrics constructions. Formulation \eqref{eqn:Definetti_model} implies that $(X_i)_{i\geq 1}$ is an exchangeable sequence with~de Finetti measure $\Pcr$ \citep{deFinetti1937}. Notice that specifying $\Pcr$ is equivalent to eliciting a joint prior distribution for the labels $(\xi_j)_{j\geq1}$, the corresponding probabilities $(p_j)_{j\geq1}$, and the alphabet size $M$, which is typically unknown and may be unbounded. In the following we assume that the labels $\xi_1,\ldots,\xi_M$~are~independent of the probabilities, and i.i.d. distributed from a diffuse prior $P_0$, which is otherwise left unspecified. This choice is  supported by the fact that,~in~line~with~standard formulations of both species and feature sampling problems,~the~specific~identities~of~the~labels are irrelevant also in our analysis, serving merely as identifiers. Regarding the prior on $(p_j)_{j\geq1}$~and~$M$, we consider instead three main elicitations covering the different support regimes for $\mu$  induced by $\Xi$.  In the finite-alphabet case, such a support consists of $M<+\infty$ labels, with $M$ either known or unknown, while in the unbounded-alphabet setting, the support of $\mu$ is countably infinite.

Focusing first on  finite alphabet cases, we assign to $\bar{\mu}=((p_j)_{j\geq1},M)$ the prior on label~probabilities and alphabet size resulting from the  finite Dirichlet process, i.e.,~$\bar{\mu} \sim \operatorname{FDP}(q_M,\gamma)$ \citep[][]{argiento2022annals}. This is obtained by letting $M\sim q_M$,~where~$q_M$~is~a probability mass function defined on $\{1,2,3,\ldots\}$, and then assuming that $((p_1,\ldots,p_M)\mid M=m)\sim \operatorname{FD}_m(\gamma,\ldots,\gamma)$, where $\operatorname{FD}_m(\gamma,\ldots,\gamma)$ is 
 the standard Dirichlet distribution on the $(m-1)$-dimensional simplex.
If the true alphabet size $M$ is known, we set $q_M(m)=\delta_{M}(m)$, recovering the standard finite Dirichlet prior. Otherwise, $M$ is learned from the data. In this latter setting a sensible choice for  $q_M$ is the one-shifted~Poisson~distribution, yielding the prior $\bar{\mu}\sim \operatorname{FDP}(\Lambda,\gamma)$ \citep{argiento2022annals,colombi2025species}. Here,  \(\Lambda\) controls the prior expected alphabet size, with~higher values placing more mass on larger alphabets. The parameter $\gamma$ controls, instead, the heterogeneity of the probabilities $p_{j}$'s. Larger values of $\gamma$ produce distributions closer to the uniform on the alphabet, i.e., $p_j=1/M$, corresponding~to maximal diversity, for instance as measured by Simpson's diversity index \citep{Simpson49}. Conversely, as $\gamma$ decreases, the diversity diminishes and~the~distribution approaches a degenerate configuration  concentrating all mass on one label.

For unbounded alphabets, we consider as prior for  $\bar{\mu}$ the one induced by a Pitman--Yor process, with the discount  parameter $\sigma\in[0,1)$ and strength $\theta>0$  \citep[][]{PitmanYor97Annals}, namely,~$\bar{\mu}\sim \pityor(\sigma,\theta)$. Leveraging the stick-breaking construction \citep{pitman95,PitmanYor97Annals}, the resulting prior on the label probabilities is defined by the sequential representation $p_1=W_1$,  $p_j=W_j\prod_{l=1}^{j-1}(1-W_l)$,~$j\geq2$, where $W_j\ind \operatorname{Beta}(1-\sigma,\theta+j\sigma)$, $j\geq1$. Under this choice,~the discount parameter~$\sigma$~controls the tail behaviour of the distribution and, consequently, the richness~of~the unseen part of the alphabet along with the rate at which new labels continue to appear~as~$n$~grows. For $\sigma>0$, the ranked probabilities exhibit a power-law behaviour with exponent $\sigma^{-1}$ as $j\to+\infty$ \citep{PitmanYor97Annals}. Thus, larger values of $\sigma$ correspond to heavier tails, greater heterogeneity, and a higher proportion of rare labels. When $\sigma=0$, the prior reduces to the Dirichlet process and no longer exhibits power-law behaviour. The strength parameter $\theta$ regulates the overall~richness of the sample without changing the asymptotic tail behaviour. Specifically, for fixed $\sigma$, larger~values of $\theta$ increase the prior tendency to generate new species.

As stated in Theorem~\ref{thm:species_post}, specializing  the general Bayesian model in \eqref{eqn:Definetti_model} to the aforementioned priors yields tractable posterior distributions under the different support regimes. These posterior characterizations form the basis to obtain the improved credible intervals for $\Mmax$ developed in Section~\ref{sec_3_conf}.

\begin{theorem}[Posterior distributions.  \cite{argiento2022annals}, \cite{pitman1996}]\label{thm:species_post}
    Let $X^n$~be~a sample from the Bayesian model in  \eqref{eqn:Definetti_model} made of $K_n$ different labels $(\xi^\star_1,\ldots,\xi^\star_{K_n})$~with~corresponding frequencies $n_j\ge1$, $j=1,\ldots,K_n$. Then the posterior of $\mu$ is defined as
    \begin{equation}
        \label{eqn:species_posterior}
        (\mu\mid X^n) \,\myequal{d} \,
        \sum\nolimits_{j=1}^{K_n}\,p^*_j\delta_{\xi^\star_j} \,+\, 
        p^*\muprime\,,
    \end{equation}
   with the following characterization, depending on the choice of the prior for $\mu$:
   \begin{enumerate}
\item if $\Pcr=\operatorname{FDP}(q_M,\gamma) \times P_0$, then
     \begin{equation}
        \label{eqn:FDP_post}
        \begin{split}
            (&p^*_1,\ldots,p^*_{K_n},p^* \mid M^\star)\sim 
            \operatorname{FD}_{K_n+1}(n_1+\gamma,\ldots,n_{K_n}+\gamma,M^\star \gamma)\,,\\
           & M^\star \sim\ q_{M^\star,} \quad q_{M^\star}(m) \propto [(m+K_n)!/m!] \, q_M(m+K_n) \,  [\left( \gamma(K_n + m) \right)_{n}]^{-1}, \\
           & \muprime\,\sim\,
            \operatorname{FDP}(q_{M^\star},\gamma) \times P_0.
        \end{split}
    \end{equation}
  where $M^\star=M-K_n$ denotes the number of unobserved labels. Under the finite Dirichlet prior, $q_{M^\star} = \delta_{M-K_n}$, whereas  for the $\operatorname{FDP}(\Lambda,\gamma)$ prior,~the~quantity $q_M$ denotes the probability mass function of a one-shifted Poisson distribution.
   
       \vspace{7pt}

\item if $\Pcr=\pityor(\sigma,\theta) \times P_0$, then
 \begin{equation}
            \label{eqn:PYP_post}
            \begin{split}
            &( p^*_1,\ldots,p^*_{K_n},p^* ) \,\sim \,
            \operatorname{FD}_{K_n+1}
            \left(n_1-\sigma,\ldots,n_{K_n}-\sigma,\theta + K_n\sigma\right), \\
            &\muprime \,\sim \, \pityor(\sigma, \theta + K_n\sigma) \times P_0,
            \end{split}
        \end{equation}
\end{enumerate}
\end{theorem}

Section~\ref{sec:features} derives a related Bayesian model-based construction under the features setting.
\vspace{20pt}

\subsection{\large 2.2 The features sampling setting}\label{sec:features}
Let $\Xi=\{\xi_1,\xi_2,\ldots\}$ denote again a countable alphabet of labels. Unlike the species framework, in the features setting each statistical unit $i=1, \ldots, n$ can simultaneously exhibit multiple labels. This information can be naturally modeled through the Bernoulli variables $Z_{ij}$, $j \geq 1$, each indicating presence or absence of the generic feature $j$ in unit $i$. Letting $p_{j}=\prob{Z_{ij}=1}\geq 0$ for $j \geq 1$, and assuming independence among these binary indicators, such a construction~can~be~equivalently expressed via $X_i=\sum_{j\geq1} Z_{ij}\delta_{\xi_j}$, with each $X_i$, $i=1, \ldots, n$, being distributed according to a Bernoulli process $\operatorname{BeP}(\mu)$ with directing discrete measure $\mu=\sum_{j\geq 1}p_j\delta_{\xi_j}$ supported on $\Xi$  \citep[e.g.,][]{Ghilotti2026}. Notice that, in contrast to the species setting of Section \ref{sec:species},  the  parameters $(p_j)_{j\geq1}$, denote inclusion probabilities and therefore need not sum to one. Instead, such quantities are only required to satisfy $\sum_{j\geq1}p_j<\infty$, which ensures that each $X_{i}$ contains only finitely many labels almost surely. As such,  \(\mu\) is not a probability measure. \smash{Nonetheless, it is still discrete}, meaning that the same labels occur repeatedly across the sample, and hence the data can be again summarized by the $K_n$ distinct labels and the corresponding~frequencies $n_j=\sum_{i=1}^n Z_{ij}$. Thus, under these settings, the features sampling~analog~of~the~Bayesian~species~model~in~\eqref{eqn:Definetti_model}~is
\begin{equation}
\label{eqn:fea_bnp_model_def}
    \begin{split}
    (X_1,\ldots,X_n \mid \mu \,) &\sim \, \operatorname{BeP}(\mu) \\
    \mu \, &\sim \, \Fcr\,,
    \end{split}
\end{equation}
where $\Fcr$ is a suitably-selected Bayesian nonparametric prior  for $\mu$. As for $\Pcr$  in Section \ref{sec:species},~also in this case the elicitation of $\Fcr$  is equivalent to specifying a joint prior  for the labels $(\xi_j)_{j\geq1}$, the corresponding inclusion probabilities $(p_j)_{j\geq1}$, and the alphabet size $M$, which is often unknown and possibly unbounded. Following the same reasoning and arguments  as in Section \ref{sec:species},~the~labels $(\xi_j)_{j\geq1}$, are again assumed to be i.i.d.\ from a diffuse $P_0$, while the prior for $\bar{\mu}=((p_j)_{j\geq1},M)$~depends on the alphabet, which may either be countably infinite or finite, with known or unknown size. In eliciting such a prior it is also important to ensure that the information contained in any single observation is almost surely finite, i.e., $\prob{X_i(\Xi)<\infty}=1$.  

In finite alphabets, we employ for $\bar{\mu}$ the prior induced by a mixed \smash{Binomial~process~$\operatorname{MBP}(q_M,a,b)$} \citep[][]{beraha2025bayesian}, which corresponds to assuming $((p_1,\ldots,p_M)\mid M) \iid \operatorname{Beta}(a,b)$,~for~\smash{$a,b>0$}, where $M\sim q_M$, with   $q_M$ denoting a probability mass function on $\{0,1,2,\ldots\}$.~A~relevant~practical example for such a probability mass function is the one induced by the negative binomial  $\operatorname{NegBin}(m,q)$ with size  $m>0$ and failure probability $q\in(0,1)$, so that $\mathbb E[M]=mq/(1-q)$. In this case, we write $\bar{\mu}\sim \operatorname{MBP}(m,q,a,b)$. Conversely, when the alphabet size $M$  is known,~it~suffices~to set $q_M(m)=\delta_{M}(m)$, recovering the classical finite Beta prior.  

When the alphabet is unbounded, we consider instead a three-parameter \smash{Beta process $\operatorname{IBP}(\gamma,\sigma,c)$} prior for $\bar{\mu}$ \citep{Teh09}, with parameters $\gamma>0$, $\sigma\in[0,1)$, and $c>-\sigma$. The inclusion probabilities $(p_j)_{j\geq1}$ are therefore viewed as the points of a Poisson process on $(0,1)$, having Lévy intensity measure $\rho(s)\,\D s$, where $\rho(s)=\gamma s^{-\sigma-1}(1-s)^{c+\sigma-1}$. In particular, the special case $\sigma=0$ recovers the classical Beta process \citep{Hjort90}, whose induced marginal distribution~on~the binary feature-allocation matrix is the Indian buffet process \citep{GriffGharha05}.

Including the above priors within the general Bayesian formulation in~\eqref{eqn:fea_bnp_model_def} yields the posterior distributions stated in  Theorem~\ref{thm:features_post}, which provide the feature-setting analog of those derived~for~the species framework in Theorem~\ref{thm:species_post}.

\begin{theorem}[Posterior distributions. \cite{Ghilotti2026} and \cite{beraha2025bayesian}]\label{thm:features_post}
    Let $X^n$~be~a sample from the model in \eqref{eqn:fea_bnp_model_def} made of $K_n$ different labels $(\xi^\star_1,\ldots,\xi^\star_{K_n})$~with~the~corresponding~frequencies $n_j\ge1$, $j=1,\ldots,K_n$. Then the posterior of $\mu$ is defined as
    \begin{equation}
        \label{eqn:features_posterior}
        (\mu\mid X^n) \,\myequal{d} \,
        \sum\nolimits_{j=1}^{K_n}\,p^*_j\delta_{\xi^\star_j} \,+\, 
        \muprime\,,
    \end{equation}
   with the following characterization, depending on the choice of the prior for $\mu$:
    \begin{enumerate}
        \item if $\Fcr=\operatorname{MBP}(q_M,a,b) \times P_0$, then
            \begin{equation}
        \label{eqn:MBP_post}
        \begin{split}
            p^*_j \,& \, {\normalfont \ind} \, \operatorname{Beta}\left(n_j+a,\,n-n_j+b\right),\\
            M^\star \,&\sim\ q_{M^\star}\,, \quad q_{M^\star}(m) \propto [(\kappa_n)^m\,(m+K_n)!/m!\,]q_M(m+K_n),\\
            \muprime\,&\sim\,
            \operatorname{MBP}(q_{M^\star},\,a,\,b+n) \times P_0\,,
        \end{split}
    \end{equation}
    where $\kappa_n = (b)_n/(a+b)_n$, while $M^\star=M-K_n$ denotes the number of unobserved labels.~Under the  finite Beta prior, $q_{M^\star} = \delta_{M-K_n}$, while for the $\operatorname{MBP}(m,q,a,b)$ prior, $q_{M^\star}$ is the probability mass function of a $\operatorname{NegBin}\left(m+K_n,\, q\,\kappa_n\right)$.
    
        \vspace{7pt}

    \item if $\Fcr=\operatorname{IBP}(\gamma,\sigma,c) \times P_0$, then
    \begin{equation}
            \label{eqn:IBP_post}
            \begin{split}
            p^*_j \,& \, {\normalfont \ind} \, \operatorname{Beta}\left(n_j-\sigma,\,n-n_j+\sigma+c\right),\\
            \muprime \,&\sim \, \operatorname{IBP}(\gamma,\sigma, c + n) \times P_0\,.
            \end{split}
        \end{equation}
    \end{enumerate}
\end{theorem}

Section~\ref{sec_3_conf} derives valid model-based credible intervals for $\Mmax$ for both species and features settings under the posterior distributions stated in Theorems~\ref{thm:species_post} and \ref{thm:features_post}, respectively.

\vspace{20pt}
\section{\large 3. Valid Model-Based Credible Intervals for Modal Missing Probabilities}\label{sec_3_conf}

\vspace{5pt}
\subsection{\large 3.1 Problem formulation and computational challenges}\label{sec:direct_inference}
Our  goal is to construct a valid one-sided credible interval for~$\Mmax$ in both species and feature settings, under the Bayesian models developed in Sections~\ref{sec:species} and \ref{sec:features}.    Formally, for a prescribed credibility level $1-\alpha$, with $\alpha\in(0,1)$, this amounts to identifying a statistic $U(X^n)$ that satisfies
\begin{equation}\label{eqn:UB_max_general}
\P\left( \Mmax \geq U(X^n)\mid X^n\right) \leq \alpha.
\end{equation}
Condition \eqref{eqn:UB_max_general} ensures that the posterior probability of $\Mmax$ exceeding $U(X^n)$, given $X^n$, is at most $\alpha$. Consequently, $U(X^n)$ defines a credible upper bound yielding the desired valid one-sided credible interval $[0,U(X^n)]$ for $\Mmax$ at level $1-\alpha$. 

Focusing on the species settings (related arguments apply under the features framework), and recalling Theorem~\ref{thm:species_post}, the posterior distribution of $\Mmax$ coincides with the law of $p^*{\cdot} p'_{(1)}$, where $p'_{(1)}$ denotes the largest probability in $\muprime$. Therefore, in principle, one could take $U(X^n)$ to be the $(1-\alpha)$-quantile of $p^*{\cdot} p'_{(1)}$. However, the law of $p'_{(1)}$ is generally intractable. Under the PYP~prior for $\bar \mu$, integral representations of both the law and moments of $p'_{(1)}$~are available \citep{PitmanYor97Annals}, but these
expressions are not practically useful for our purposes. Under the FDP prior for $\bar \mu$, with $q_M=\delta_{M}$, $p'_{(1)}$ is the largest order statistic of a finite Dirichlet random vector, for which no closed-form distribution is available. Finally, for a general prior $q_M$,~the~law~of~$p'_{(1)}$~is~even~more complex. Thus, closed-form inference on $\Mmax$ is generally infeasible.

A natural solution to address the above issues  is to produce suitable Monte Carlo samples~from the posterior of $p^*{\cdot}p'_{(1)}$ under the models in Section~\ref{sec_2_mod}, and then leverage these samples to obtain empirical estimates for the $(1-\alpha)$ posterior quantile of $p^*{\cdot}p'_{(1)}$. In practice, however, this approach may be unsuitable for our purposes. The first issue is computational. Exact simulation of $p'_{(1)}$ is not equally convenient under the different specifications considered in Section~\ref{sec_2_mod}. For the FDP, this can be done by simulating $\mu^\prime$, whose number of support points $M^\star\sim q_{M^\star}$ is almost surely finite, and taking its largest component. Under the PYP prior, one may in principle sample $p'_{(1)}$ using the exact algorithm of \cite{Dassios2021}, but this requires solving a collection of nontrivial optimization problems and is therefore considerably less convenient in practice. More importantly, even if an i.i.d. Monte Carlo sample of size $B$ from the posterior distribution of $\Mmax$ were available, estimation of its
$(1-\alpha)$-quantile could still be numerically unstable. Standard large-sample theory for empirical quantiles implies that the corresponding Monte Carlo error is governed by the posterior density evaluated at the target quantile. In particular, if $\hat U(X^n)$~denotes~a~Monte Carlo estimate of
the true quantile \(U(X^n)\), \smash{then the relative error $E(U(X^n)-\hat U(X^n))^2/U(X^n)^2$ is} approximately proportional to \smash{$1/(B[U(X^n)f(U(X^n))]^2)$}, where $f$ is  the posterior density  of $\Mmax$. Consequently, a large sample size $B$ is needed whenever \smash{$U(X^n)f(U(X^n))$} is small. This typically occurs in regimes of greatest practical interest, namely when the posterior distribution of $\Mmax$ is highly concentrated near zero. For instance, under the FDP prior this may arise for small values of $\gamma$ or for large $M^\star$, while under the $\pityor$ prior it may occur when $\theta+K_n\sigma$ is large.

In Section~\ref{section:species_results}, we address the above challenges by taking a different perspective based~on~$r$th-order missing masses \citep[see, e.g.,][]{Chandra2024,naulet24}.

\vspace{12pt}

\subsection{\large 3.2 Valid one-sided credible intervals for $\Mmax$ in species and features settings}\label{section:species_results}

Because directly targeting the $(1-\alpha)$ posterior quantile of $\Mmax$ is analytically and computationally challenging, we seek a tractable posterior approximation leveraging $r$th-order~missing~masses \citep[][]{Chandra2024,naulet24}. Such a technical construction~has~been successfully employed in worst-case scenario analyses  \citep{Painsky24,ColombiFreq},~and, as clarified in the following, can be effectively exploited as a building-block in deriving the desired credible intervals for $\Mmax$ also under our model-based Bayesian nonparametric framework. To this end, for any integer $r\geq 1$, define the \textit{\(r\)th-order missing mass} as
\begin{equation}\label{eqn:rnorm_def}
\Mr \,\coloneqq\, \Mr(X^n;p_j, j \geq 1)=\sum\nolimits_{j\geq 1} p_j^r \indic(n_j=0),
\end{equation}
and notice that $\Mrr$ yields the classical missing mass, while, \smash{more generally, $\Mr^{1/r}$ is the $\ell_r$-norm} of the vector of unseen probabilities, with $\Mmax$ corresponding to the $\ell_\infty$-norm. As such, by  norm monotonicity, any valid credible upper bound for $\Mr^{1/r}$ also controls $\Mmax$. More specifically,
\begin{equation}\label{eqn:UB_Mr_approx}
\P\left( \Mmax \geq U(X^n)\mid X^n\right)\leq\P( \Mr^{1/r} \geq U(X^n)\mid X^n) \leq \alpha,\qquad \mbox{for any } r\geq 1 .
\end{equation}
The key advantage of working with the quantity $\Mr^{1/r}$ rather than directly with $\Mmax$~is~that~the posterior mean $\mean{\Mr\mid X^n}$ can often be
computed analytically. Since $\Mr$ is nonnegative, by the Markov's inequality we have
\begin{displaymath}
\P( \Mr^{1/r} \geq U(X^n)\mid X^n)=\P\left( \Mr \geq U(X^n)^r\mid X^n\right)\leq \mean{\Mr\mid X^n}/U(X^n)^r.
\end{displaymath}
Therefore, setting
\begin{equation} \label{eqn:species_UB_general}
U(X^n)=(\mean{\Mr\mid X^n}/\alpha)^{1/r},
\end{equation}
yields a credible upper bound for $\Mmax$ at level $1-\alpha$, and hence a corresponding valid one-sided credible interval $[0,U(X^n)]$ at level $1-\alpha$. Notice that, although \eqref{eqn:species_UB_general} holds for every $r \geq 1$, the optimal order is the one that  minimizes  $U(X^n)$ in \eqref{eqn:species_UB_general}, thus yielding the shortest, and hence~more informative, valid credible interval $[0,U(X^n)]$; see Section~\ref{sec:comp_and_infer} for details.

Specializing \eqref{eqn:species_UB_general} to the Bayesian models designed in Sections~\ref{sec:species} and \ref{sec:features}, yields the closed-form expressions for $U(X^n)$ stated in Theorems~\ref{thm:Species_UB_FDP_PD} and \ref{thm:Features_upper_bounds}, respectively.

\begin{theorem}[{Credible upper bounds in the species setting}]
    \label{thm:Species_UB_FDP_PD}
    Let $X^n$ be a sample~from~the~model~\eqref{eqn:Definetti_model}, and let $U(X^n)$ be a credible upper bound for $\Mmax$ as defined in \eqref{eqn:UB_max_general}. Then,
    \begin{enumerate}
        \item if $\Pcr = \operatorname{FDP}(q_M,\gamma) \times P_0$, we have that $U(X^n)=U^{\operatorname{FDP}}$ with
        \begin{equation}
            \label{eqn:UB_FDP}
            U^{\operatorname{FDP}} :=  \min\nolimits_{r\geq1}   U_r^{\operatorname{FDP}}(X^n)=
            \min\nolimits_{r\geq1}
          \{ \left(
            {\E_{M^\star}}[M^\star (\gamma)_r/(n + \gamma(K_n+M^\star))_r]/\alpha
            \right)^{1/r}
           \},
        \end{equation}
        where $M^\star$ is distributed according to $q_{M^\star}$ as in \eqref{eqn:FDP_post}. When the alphabet size is known (i.e.,~$\Pcr = \operatorname{FDP}(\delta_{M},\gamma) \times P_0$), \eqref{eqn:UB_FDP} reduces to $U^{\operatorname{FD}}:=\min_{r\geq1}\{\left([(M-K_n)(\gamma)_r/(n+\gamma M)_r]/\alpha\right)^{1/r}\}$.

        \vspace{7pt}

        \item if $\Pcr = \pityor(\sigma,\theta) \times P_0$, we have that  $U(X^n)=U^{\pityor}$ with
        \begin{equation}
            \label{eqn:PD_UB_species}
            U^{\pityor} \,: = \,    \min\nolimits_{r\geq1}   U_r^{\pityor}(X^n) =
            \min\nolimits_{r\geq1}
      \{ \left(
          [(\theta + K_n\sigma)\,
            (1-\sigma)_{r-1}/
            (n + \theta)_{r}]/\alpha
            \right)^{1/r}
     \}. 
        \end{equation}
    \end{enumerate}
\end{theorem}

\begin{theorem}[{Credible upper bounds in the features setting}]
    \label{thm:Features_upper_bounds}
    Let $X^n$~be~a~sample~from~the~model~\eqref{eqn:fea_bnp_model_def}, and let $U(X^n)$ be a credible upper bound for $\Mmax$ as defined in \eqref{eqn:UB_max_general}. Then,
    \begin{enumerate}
        \item if $\Fcr = \operatorname{MBP}(q_M,a,b) \times P_0$, we have that $U(X^n)=U^{\operatorname{MBP}}$ with
        \begin{equation}
        \label{eqn:UB_GeneralMBP}
            U^{\operatorname{MBP}} \, := \, \min\nolimits_{r\geq1} U_r^{\operatorname{MBP}}(X^n)=
            \min\nolimits_{r\geq1}
         \{ \left( \,
           [ E_{M^{\star}}[M^{\star}]  (a)_r/(a+b+n)_r]/\alpha
            \right)^{1/r}
\},
        \end{equation} 
        where $M^\star$ is distributed according to $q_{M^\star}$ as in \eqref{eqn:MBP_post}. Under the probability mass function of~the $\operatorname{NegBin}(m,q)$, it holds   $E_{M^{\star}}[M^{\star}] = q\kappa_n(m+K_n)/(1-q\kappa_n)$ and the resulting~upper~bound $U(X^n)$ is denoted by $U^{\operatorname{N-MBP}}$. Conversely, when the  size of the alphabet is known (i.e., $\Fcr = \operatorname{MBP}(\delta_{M},a,b) {\times} P_0)$, \eqref{eqn:UB_GeneralMBP} reduces to  $U^{\operatorname{FB}} := 
            \min\nolimits_{r\geq1}
         \{ \left( \,
           [ (M-K_n)  (a)_r/(a+b+n)_r]/\alpha
            \right)^{1/r}\}$.
    \vspace{7pt}
        \item if $\Fcr = \operatorname{IBP}(\gamma,\sigma,c) \times P_0$,  we have that $U(X^n)=U^{\operatorname{IBP}}$ with
        \begin{equation}
            \label{eqn:UB_IBP} 
            U^{\operatorname{IBP}} \, = \,  \min\nolimits_{r\geq1} U_r^{\operatorname{IBP}}(X^n)=
            \min\nolimits_{r\geq1}
        \{ \left(
            B(r-\sigma,\, c+\sigma+n) \gamma/\alpha
            \right)^{1/r}
    \}, 
        \end{equation}
        where $B(\cdot,\cdot)$ is the Beta function.
    \end{enumerate}
\end{theorem}

{
Theorems \ref{thm:Species_UB_FDP_PD}--\ref{thm:Features_upper_bounds} provide bounds whose  calculation requires \smash{solving a minimization problem over} the discrete orders $r \geq 1$. As clarified in Section~\ref{sec:comp_and_infer}, such a minimization~admits a~tractable~implementation and yields a unique minimum. Before turning to these computational~aspects~it~is~however instructive to first examine the behaviour of the quantities being minimized,~as~this~provides insight into the resulting bounds; see Section~\ref{app:UB_limits} in the Supplementary Material for details.

With this goal in mind, we start by considering $U^{\operatorname{FD}}$ within Theorem \ref{thm:Species_UB_FDP_PD} as it makes explicit~the effect of the hyperparameter $\gamma$. For any fixed $r\geq1$, the quantity being minimized tends~to~zero as $\gamma\to0$. This is intuitive. As $\gamma\to0$, Simpson's diversity index approaches zero and the distribution concentrates on a single label. If $K_n>0$, such a label has necessarily been observed, and hence $\Mmax=0$. Conversely, as $\gamma\to+\infty$, which corresponds to Simpson's diversity index approaching its maximum value, the quantity entering the minimization stabilizes, for fixed $r \geq 1$, around the horizontal asymptote
\begin{displaymath}
\lim_{\gamma\to+\infty}\left([(M-K_n)(\gamma)_r/(n+\gamma M)_r]/\alpha\right)^{1/r}={(M-K_n)^{1/r}}/({M\,\alpha^{1/r}}),
\qquad r\geq1 .
\end{displaymath}
Thus, for large $\gamma$, the bound approaches the behaviour associated with a nearly uniform distribution on the alphabet. Analogous results can be derived  for $U^{\pityor}$ in Theorem~\ref{thm:Species_UB_FDP_PD}. In particular, for any fixed $r > 1$, the quantity entering the minimization underlying $U^{\pityor}$ converges to zero~as $\sigma \to 1$, regardless of $\theta$.  Intuitively, when $\sigma$ approaches one, each label can be observed only once, akin to the case of continuous distributions, which makes the notion of $\Mmax$ ill-posed.~Conversely,  fixing $\sigma < 1$ and letting $\theta \to 0$, yields
\begin{displaymath}
\lim_{\theta\to 0}\  \left(
          [(\theta + K_n\sigma)\,
            (1-\sigma)_{r-1}/
            (n + \theta)_{r}]/\alpha
            \right)^{1/r} = \,  \left(K_n \, \sigma\,(1-\sigma)_{r-1}/(\alpha \, (n)_r) \right)^{1/r}\,,\quad r>1\,.
\end{displaymath}
Hence, if $\sigma = 0$, which corresponds to the Dirichlet process, then $U^{\pityor} \to 0$ as $\theta \to 0$, which is to be expected since the Dirichlet process degenerates to a measure supported on a single point. If instead $\sigma \in (0, 1)$ and $\theta \to 0$, then $U^{\pityor}$ remains strictly positive.
Finally, for $\theta \to \infty$, the quantity minimized in $U^{\pityor}$ vanishes polynomially, with order $\theta^{-(1-1/r)}$. In particular
\begin{displaymath}
\lim_{\theta\to+\infty}  \left(
          [(\theta + K_n\sigma)\,
            (1-\sigma)_{r-1}/
            (n + \theta)_{r}]/\alpha
            \right)^{1/r} = \,
\lim_{\theta\to+\infty}
(  (1-\sigma)_{r-1}/\alpha )^{1/r}\theta^{-(1-1/r)} \, = \,
0\,.
\end{displaymath}
for any  $r > 1$. Figure \ref{fig:UB_shape_FDP_PD} shows $U^{\operatorname{FD}}$, $U^{\operatorname{FDP}}$,  and $U^{\pityor}$ for different choices of $M$~and~$K_n$~as~a~function of~the hyperparameters $\gamma$ and $\theta$. For $\sigma=0$, $U^{\pityor}$ is not affected by $K_n$ that is a well-known property of the posterior under the Dirichlet process. In general, we observe that as $\gamma$ and $\theta$ start to increase from zero, the credible upper bounds increase as well to reach a global maximum for relatively small values of $\gamma$ and $\theta$, and then decrease. The choices of $q_M$ and $\sigma$ do not change this overall behavior, but rather impact the length of the induced credible intervals. 

\begin{figure}[b!]
\vspace{5pt}
    \centering
    \includegraphics[width=0.9\linewidth]{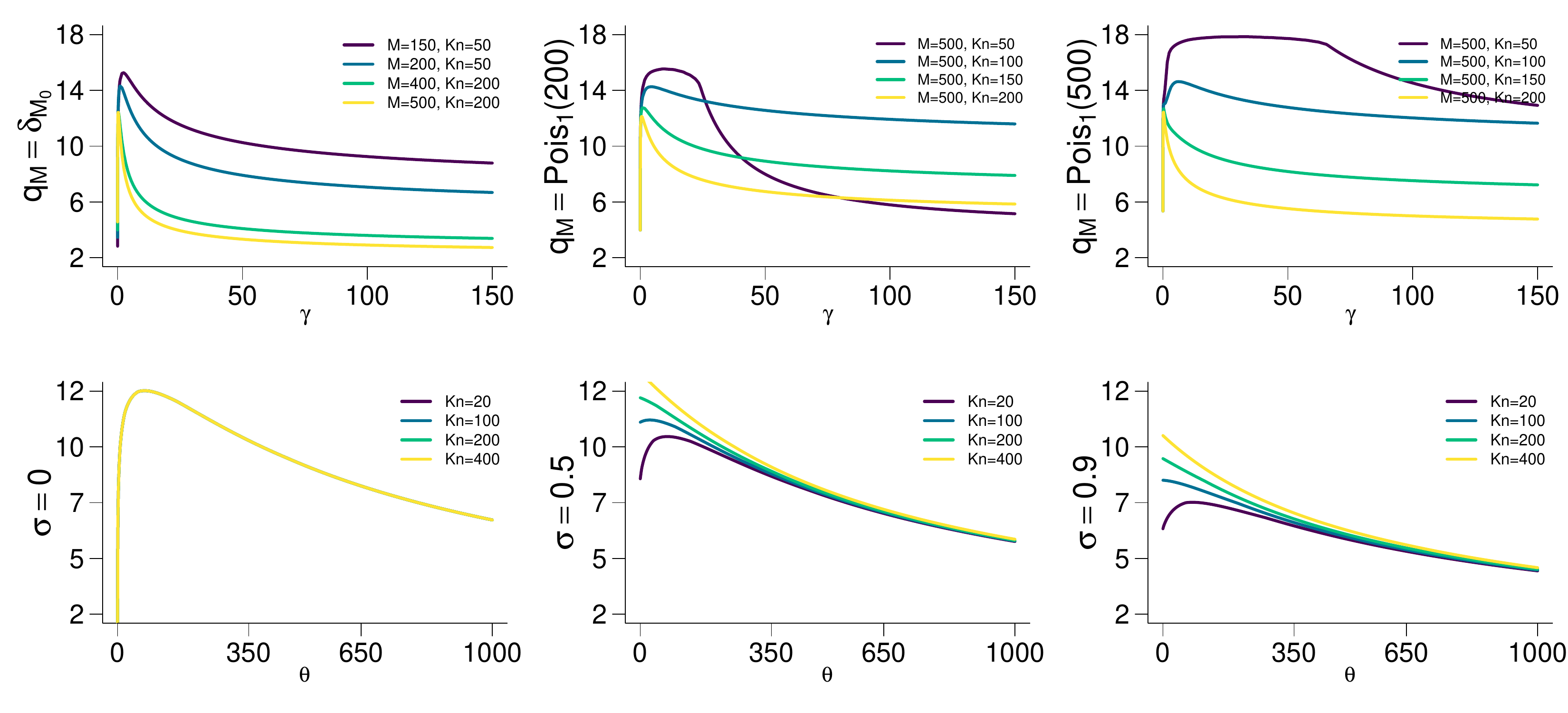}
    \vspace{-10pt}
    \caption{Bayesian credible upper bounds under the FD/$\operatorname{FDP}$ (top panels), and $\pityor$ (bottom panels) priors. 
    The vertical axis has been rescaled by $10^3$ for readability.
    }
    \label{fig:UB_shape_FDP_PD}
    \vspace{-7pt}
\end{figure}

Let us now turn our attention to the upper bounds for the features setting in Theorem~\ref{thm:Features_upper_bounds}. To this end,~we~first focus on $U^{\operatorname{MBP}}$ in \eqref{eqn:UB_GeneralMBP} viewed as a function of the hyperparameters $a,b>0$, while keeping the multiplicative factor $\mathbb E_{M^\star}[M^\star]$ and $r\geq1$ fixed. We also reparametrize \eqref{eqn:UB_GeneralMBP} in terms of expected value $\omega:=a/(a+b)\in(0,1)$ of the Beta prior, and the parameter $\eta:=a+b+1>1$, so that the variance is $\omega(1-\omega)/\eta$. Under these settings, the quantity minimized in \eqref{eqn:UB_GeneralMBP}  can be re-expressed as
$[E_{M^\star}[M^\star]\left(\omega(\eta-1)\right)_r/(\alpha\left(\eta-1+n\right)_r)]^{1/r}$, which clarifies that, for any fixed~$\eta>1$, $U^{\operatorname{MBP}}$ is strictly increasing in $\omega$ and converges to zero as $\omega\to0$. Similarly, it is strictly increasing in $\eta$, at least for values of $r$ smaller than $1+n\omega/(1-\omega)$; for larger values of $r$, the monotonicity is less transparent. Finally, considering  $U^{\operatorname{IBP}}$ in \eqref{eqn:UB_IBP}, notice that its expression does not involve the number $K_n$ of distinct labels. Hence, its dependence  on the data arises only indirectly, through the estimation of the hyperparameters. Moreover, $c$ and $n$ enter $U^{\operatorname{IBP}}$ only through a sum, and hence, the relative effects cannot be disentangled. As a consequence, the impact of $c$ becomes~negligible for large sample sizes. The credible upper bound $U^{\operatorname{IBP}}$ is strictly decreasing in $c+n$, consistently with the fact that this is the only term through which the sample size enters the expression. In addition, it is increasing in $\sigma$ whenever $\sigma>(r-c-n)/2$, a condition expected to hold in practice since $n$ is typically much larger than the selected value of $r$. Finally, $U^{\operatorname{IBP}}$ is increasing in $\gamma$,~which appears as a multiplicative factor. 
Figure \ref{fig:UB_shape_MBP} displays the credible upper bounds \eqref{eqn:UB_GeneralMBP} and \eqref{eqn:UB_IBP} as functions of the parameters $\omega$ and $c$, for selected values of the remaining hyperparameters, and confirms the analytical findings discussed above.

\begin{figure}[t!]
    \centering
    \includegraphics[width=0.9\linewidth]{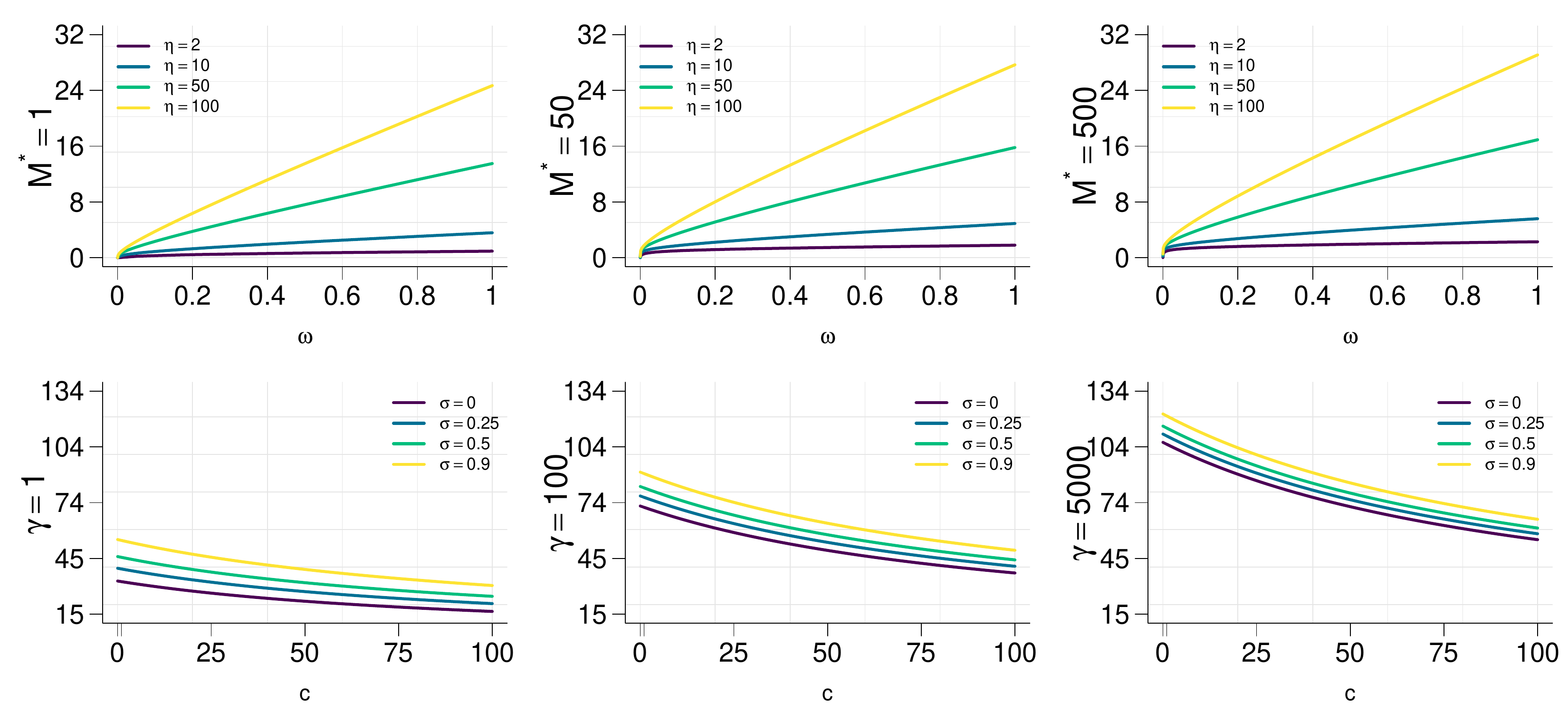}
    \vspace{-10pt}
    \caption{Bayesian credible upper bounds under the MBP (top panels) and IBP (bottom panels) priors. The vertical axis has been rescaled by $10^3$ for readability.}
    \label{fig:UB_shape_MBP}
\end{figure}

\vspace{10pt}
\subsubsection{3.2.1 Comparison with worst-case scenario confidence intervals}\label{section:c_worst}
As anticipated in Section~\ref{sec:intro}, the available upper bounds for $\Mmax$ in  \citet{Painsky24}~and~\citet{ColombiFreq} are derived by computing expectations with respect to a worst-case \smash{data-generating} distribution, thereby yielding the most conservative interval among those induced by~all~the~admissible choices for such a distribution. Conversely, recalling the result in \eqref{eqn:species_UB_general}, the  upper  bounds stated in Theorems \ref{thm:Species_UB_FDP_PD} and \ref{thm:Features_upper_bounds} are obtained by averaging with respect to a posterior distribution~that effectively combines flexible Bayesian nonparametric priors and empirical evidence from the observed sample to  localize the one-sided interval for $\Mmax$ around the data-generative mechanisms coherent with $X^n$. In fact, unlike for the upper bounds arising from worst-case scenario analyses, those we derive exploit both structural information through the prior and empirical evidence from the data.  As clarified in Section~\ref{sec:comp_and_infer}, the latter enter not only the posterior expectation in \eqref{eqn:species_UB_general}, but also estimation of the prior hyperparameters, resulting in more informative intervals than those available from worst-case analyses, which largely discard data-driven information.

To clarify the above discussion, let us consider the more challenging unbounded-alphabet~case. 
In the species setting, the worst-case frequentist confidence intervals depend only on the sample size $n$, and not on the composition of the observed sample $X^n$. By contrast, even for fixed~hyperparameters, our credible intervals already adapt to the number $K_n$ of discovered species. Similar gains can be highlighted for the  features setting by rewriting the $r$-th power of the quantity~minimized in \eqref{eqn:UB_IBP} as  $(E(\sum_{j\ge 1}p_j)/\alpha)\,(B(r-\sigma,c+\sigma+n)/B(1-\sigma,c+\sigma))$, and comparing it with~the worst-case  oracle counterpart $(S_0/\alpha)\,\left((r-1)/(n+r-1)\right)^{r-1}\left(n/(n+r-1)\right)^n\,$  in \citet[Lemma 3.2]{ColombiFreq}, which assumes knowledge of the true total mass $S_0=\sum_{j\ge 1}p_{0,j}$. As~such~$S_0$ and $E(\sum_{j\ge 1}p_j)$ effectively represent the same object. This makes it possible to isolate the effect of worst-case calibration versus the proposed model-based one by comparing the remaining factors. Such a comparison is shown in Figure \ref{fig:Freq_vs_Bayes}, for $n=500$ and $r=\log(n)$, while varying~the~IBP hyperparameters. It emerges that, regardless of the chosen hyperparameters, the factor corresponding to the IBP case is always tighter, and hence, preferable to its frequentist counterpart.

\begin{figure}[t!]
    \centering
    \includegraphics[width=0.7\linewidth]{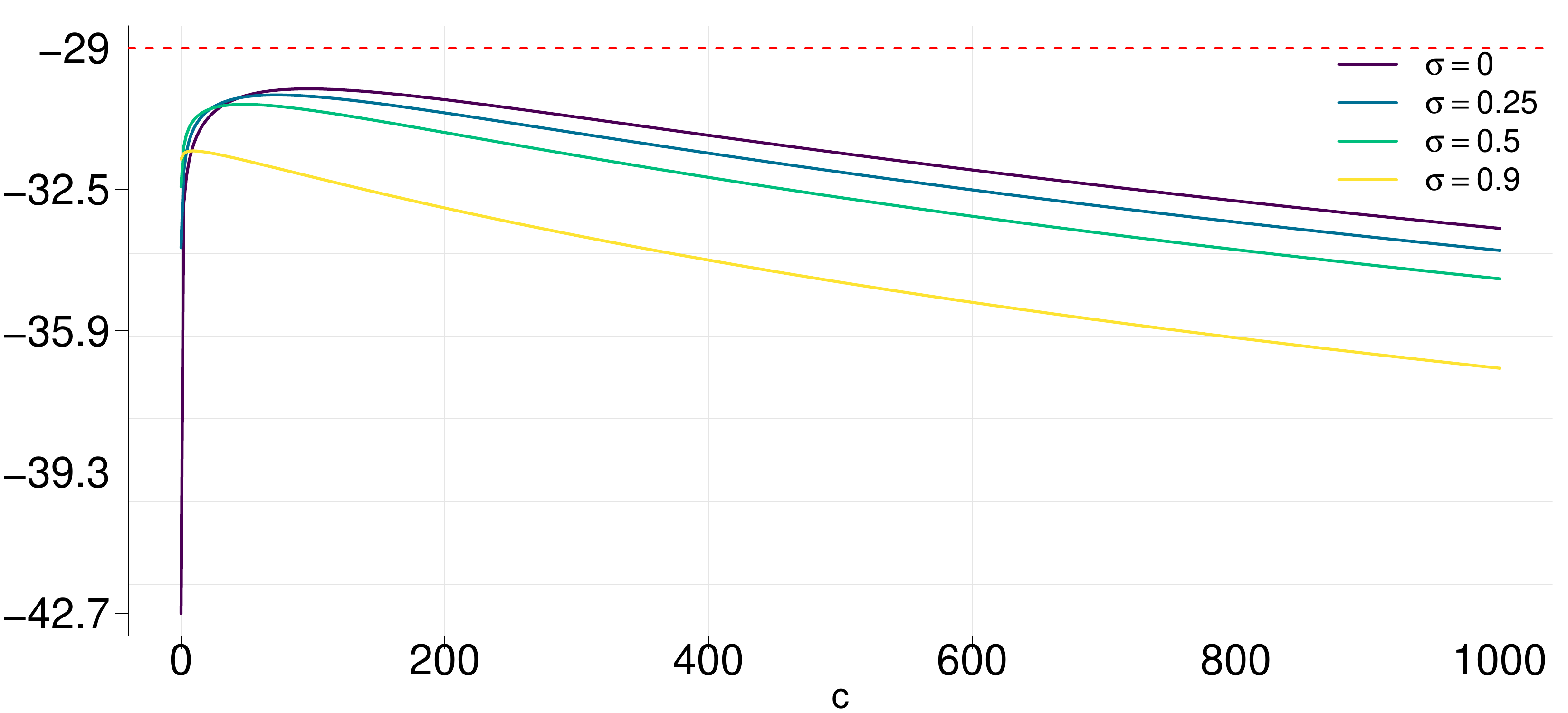}
    \caption{The dashed line shows the worst-case frequentist scaling factor. The solid ones represent~the factor resulting from our  Bayesian nonparametric solution under the IBP, as a function of the hyperparameters.}
    \label{fig:Freq_vs_Bayes}
    \vspace{-8pt}
\end{figure}

The above differences reduce as the sample size $n$ grows. In this regime, Lemma~\ref{lem:ibp-asymptotics-theorem6}~in~the~Supplementary Material proves that \eqref{eqn:UB_IBP} is minimized at $r_n = \sigma \log n - \left(\sigma+0.5\right)\log\log n+O(1)$,~and the resulting upper bound becomes  $U^{\operatorname{IBP}}=n^{-1}\left(\sigma\log n-(\sigma+\log\log n)\right)(1+O(1))\,.$
The latter matches the optimal rate established in \citet[]{ColombiFreq}, up to lower-order terms.
Similarly, in the bounded-alphabet case with known  size $M$, the optimal bound is $n^{-1}\log((M-K_n)/\alpha)$ \citep{ColombiFreq}, which is matched (asymptotically) by assuming a finite Beta prior. Indeed, simple algebra leads to $U^{\operatorname{FB}} = n^{-1}\log\left((M-K_n)/\alpha\right) + o\left(1/n\right).$}

\section{\large 4. Theory on Frequentist Coverage of Credible Intervals}\label{sec:freq_properties_all}
Expanding the discussion in Section~\ref{section:c_worst}, another key difference between our model-based perspective and existing worst-case analyses is that the latter provide one-sided confidence intervals under a frequentist framework, whereas those we derive in Theorems~\ref{thm:Species_UB_FDP_PD}--\ref{thm:Features_upper_bounds}~correspond~to Bayesian credible intervals. More specifically, current frequentist solutions \citep{Painsky24,ColombiFreq} assume the sample $X^n$ to be i.i.d. according to a true distribution $\mu_0$, and  then define~$U(X^n)$ as the upper bound satisfying, for any $\mu_0$,  $\P_{\mu_0}\!\left(\Mmax \ge U(X^n)\right) \le \alpha,$ where $\P_{\mu_0}$ is the probability under the true model.  Conversely, our intervals  in Theorems~\ref{thm:Species_UB_FDP_PD} and \ref{thm:Features_upper_bounds} are designed to satisfy the inequality in \eqref{eqn:UB_max_general}, which relies on posterior probabilities. As such, it is important to understand whether these intervals admit also a frequentist support.

Theorems \ref{thm:freq-validity-species} and \ref{thm:freq-validity-features} below, provide such a support by stating that the Bayesian credible intervals we derive in Section~\ref{section:species_results} have also valid asymptotic frequentist coverage.  That is, under suitable assumptions on $\mu_0$, the one-sided credible intervals in Theorems~\ref{thm:Species_UB_FDP_PD} and \ref{thm:Features_upper_bounds} are asymptotically valid confidence intervals. Prior to stating the results, let us recall that $\mu_0=\sum_{j\geq1}p_{0,j}\delta_{\xi_{j}}$ belongs~to~the class of regularly varying distributions if
\begin{equation}\label{eq:reg_vary}
    \bar\nu_0(x) := \# \{ j \ge 1 : p_{0,j} > x \}, \sim c_0 x^{-\sigma_0} L(1/x),
    \qquad x \downarrow 0,
\end{equation}
for constants $\sigma_0$, $c_0$ and a positive slowly-varying function $L$, and furthermore
\begin{equation}\label{eq:smoothness-rv}
    0 \le \bar\nu_0(x) - \bar\nu_0((1+t)x) \le C_I t \bar\nu_0(x),
    \qquad 0 < x < x_0,\quad 0 \le t \le 1;
\end{equation}
for a positive constant $C_I$. Condition \eqref{eq:smoothness-rv} is a one-sided local smoothness condition on $ \bar\nu_0(x) := \# \{ j \ge 1 : p_{0,j} > x \}$; such a condition is implied, for instance, by normalized regular variation of $\bar\nu_0$ \citep[Chapter 1]{bingham1989regular}. Under these settings, Theorems \ref{thm:freq-validity-species} and \ref{thm:freq-validity-features} provide frequentist support to the one-sided credible intervals in Theorems~\ref{thm:Species_UB_FDP_PD} and \ref{thm:Features_upper_bounds}. 

\begin{theorem}[{Frequentist validity in the species setting}]\label{thm:freq-validity-species}
Fix $\alpha \in (0,1)$, and let $X^n$ be a sample of size $n$ from model \eqref{eqn:Definetti_model} with fixed $\mu_0 = \sum_{j \ge 1} p_{0,j} \delta_{\xi_j}$. Let $U_n^{\operatorname{FDP}}$ and $U_n^{\pityor}$ be the credible upper bounds defined in \eqref{eqn:UB_FDP} and \eqref{eqn:PD_UB_species}, respectively, with the subscript $n$ emphasizing the dependence~on the sample size $n$. Then, as $n \to \infty$,
\begin{enumerate}
    \item if $\mu_0$ has support on a finite number of points, we have that
    \[
    \P_{\mu_0} ( \Mmax(X_n;\mu_0) \le U_n^{\sharp} ) \to 1, \quad U_n^\sharp \in \{U_n^{\operatorname{FDP}}, U_n^{\operatorname{PYP}}\}.
    \]
    
    \item if instead $\mu_0$ belongs to the class of regularly varying distributions, 
   with $0  < \sigma < \sigma_0 + 0.5$, 
       \[
        \P_{\mu_0} ( \Mmax(X_n;\mu_0) \le U_n^{\pityor} ) \to 1.
    \]
\end{enumerate}
\end{theorem}

\begin{theorem}[{Frequentist validity in the features setting}]\label{thm:freq-validity-features}
Fix $\alpha \in (0,1)$. Let $X^n$ be a sample of size $n$ from  model~\eqref{eqn:fea_bnp_model_def} with fixed $\mu_0 = \sum_{j \ge 1} p_{0,j} \delta_{\xi_j}$. Let $U_n^{\operatorname{MBP}}$ and $U_n^{\operatorname{IBP}}$ denote the credible~upper bounds defined in \eqref{eqn:UB_GeneralMBP} and \eqref{eqn:UB_IBP}, respectively, with the subscript $n$ emphasizing the dependence~on the sample size $n$. Then, as $n \to \infty$,
    \vspace{-5pt}
\begin{enumerate}
    \item if $\mu_0$ has support on a finite number of points, we have that
        \vspace{-6pt}
    \[
    \P_{\mu_0} ( \Mmax(X_n;\mu_0) \le U_n^{\sharp} ) \to 1, \quad U_n^\sharp \in \{U_n^{\operatorname{MBP}}, U_n^{\operatorname{IBP}}\}.
    \]
            \vspace{-20pt}
    \item if $\mu_0$ instead belongs to the class of regularly varying distributions, the following results hold:
    \vspace{-20pt}
    \begin{itemize}
        \item[2.1] for fixed hyperparameter values $(\sigma,\gamma,c)$, $U_n^{\operatorname{IBP}}$ yields an asymptotically valid frequentist confidence interval when $\sigma > \sigma_0$.
        \item[2.2] assume that the maximum marginal likelihood estimator of $(\sigma,c)$ is constrained to a compact parameter set $\Theta \subset \{(\sigma,c):0<\sigma<1,\ c>-\sigma\}$. More explicitly, there exist constants $\underline{\sigma},\overline{\sigma}$ and $\eta>0$ such that, for every $(\sigma,c)\in\Theta$, $0<\underline{\sigma}\leq \sigma\leq \overline{\sigma}<1$, $c+\sigma\geq \eta$. Let $(\hat\sigma_n,\hat\gamma_n,\hat c_n)$ be selected by maximum marginal likelihood over $\Theta \times \mathbb R_+$. Then,         if $\limsup_{n\to\infty}\hat\sigma_n<\sigma_0+1/2$,  $U_n^{\operatorname{IBP}}$ yields an asymptotically valid frequentist confidence interval.
    \end{itemize}
\end{enumerate}
\end{theorem}
\vspace{-5pt}
Besides providing frequentist support, Theorems \ref{thm:freq-validity-species}--\ref{thm:freq-validity-features} crucially clarify the role of the hyperparameters in regulating the large-sample properties of the intervals  in Theorems~\ref{thm:Species_UB_FDP_PD}--\ref{thm:Features_upper_bounds}.~In~the~species setting, these hyperparameters have only a limited effect on the validity of the  frequentist results established within Theorem \ref{thm:freq-validity-species}. If $\mu_0$ has finite support, then $U^{\operatorname{FDP}}$ has valid asymptotic frequentist coverage regardless of the choice of $q_M$ and $\gamma$. If, instead, $\mu_0$ has infinite support and belongs to the class of regularly varying distributions, then $U^{\pityor}$ has valid asymptotic frequentist coverage for a broad range of values of $\sigma$, and for any $\theta$. Therefore, although  the hyperparameters remain crucial for the design of informative credible intervals and for the corresponding finite-sample coverage, these terms do not affect the limiting frequentist properties of the intervals.  

According to Theorem~\ref{thm:freq-validity-features}, the above results apply also to the features setting when  $\mu_0$~has~finite support. Conversely, in infinte-support regimes, asymptotic frequentist validity under  fixed hyperparameters holds when $\sigma > \sigma_0$,  a condition that is difficult to verify in practice.~As discussed in Section \ref{sec:comp_and_infer}, this issue can be addressed by selecting the hyperparameters via maximum marginal likelihood. In this case, the desired frequentist validity follows naturally~from~the resulting data-dependent calibration.   As a matter of fact, the assumptions on the estimator  in Theorem~\ref{thm:freq-validity-features}~are mild. We do not require consistency of the full maximum marginal likelihood estimator, nor do we require $\hat\sigma_n$ to converge to the true variation index $\sigma_0$. The proof~only~uses~the~fact that the optimization over $(\sigma,c)$ is performed on a compact subset of the positive-discount parameter space, bounded away from the boundary, and that $\hat\sigma_n$ is eventually smaller than $\sigma_0+0.5$~in~probability. Recalling \citet[][]{ascolani2026asymptotic}, $\hat \sigma_n \to \sigma_0$. Hence, $0 < \sigma < \sigma_0 + 0.5$ is a mild assumption.

\section{\large 5. Computational Aspects}\label{sec:comp_and_infer}
{
The credible intervals for $\Mmax$ derived in Theorems \ref{thm:Species_UB_FDP_PD}--\ref{thm:Features_upper_bounds} are available in closed form once both the minimizing order $r^\star$ and the prior hyperparameters have been specified. These two quantities play conceptually different roles and should therefore be treated separately. Recalling Section~\ref{section:species_results}, the order only controls the length of the interval, and hence, should be selected to minimize such a length, thereby yielding the most informative one-sided interval among those that satisfy  \eqref{eqn:UB_max_general}. As clarified in Proposition~\ref{prop:uniquness_minimizer}, this minimization problem admits a well-defined solution.

\begin{figure}[b!]
    \centering
    \includegraphics[width=0.32\linewidth]{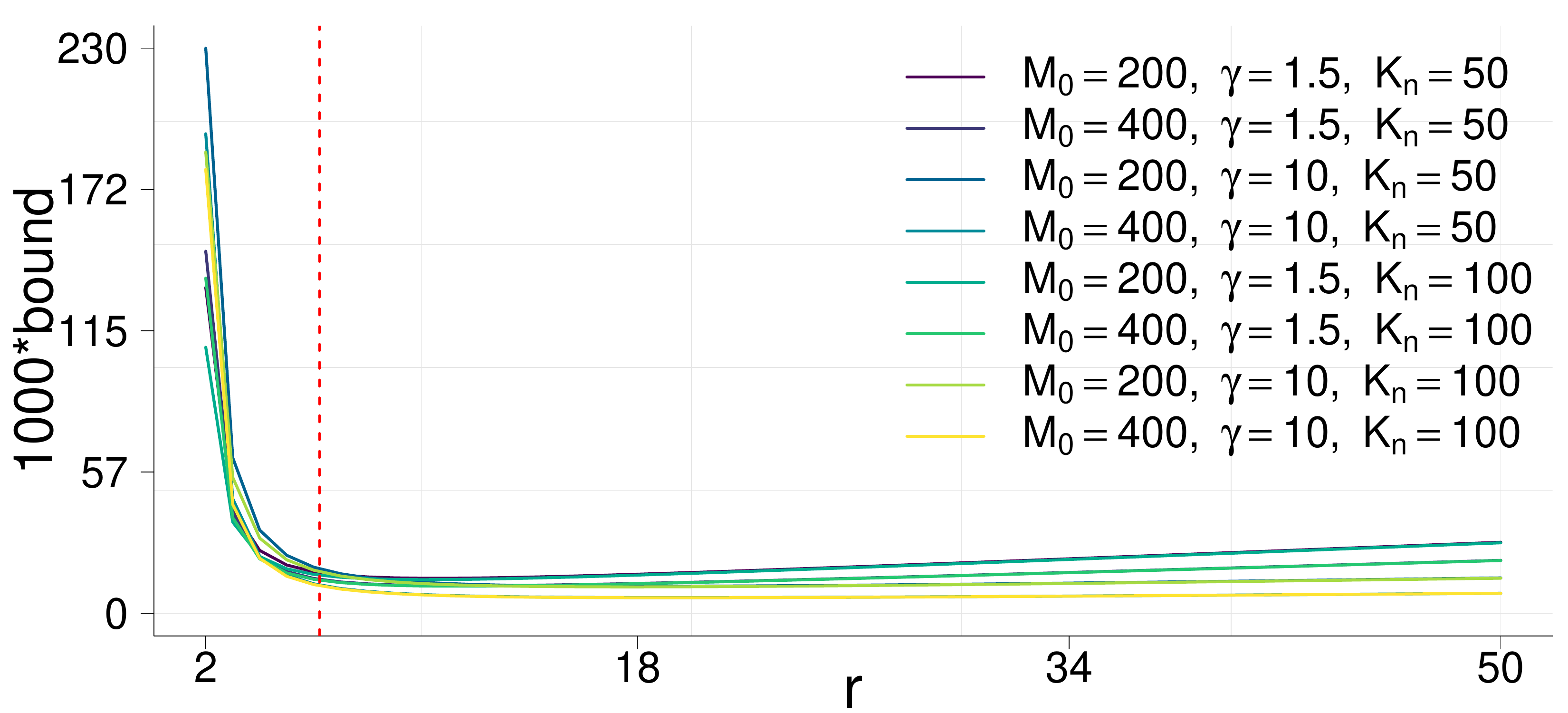}
    \hfill
    \includegraphics[width=0.32\linewidth]{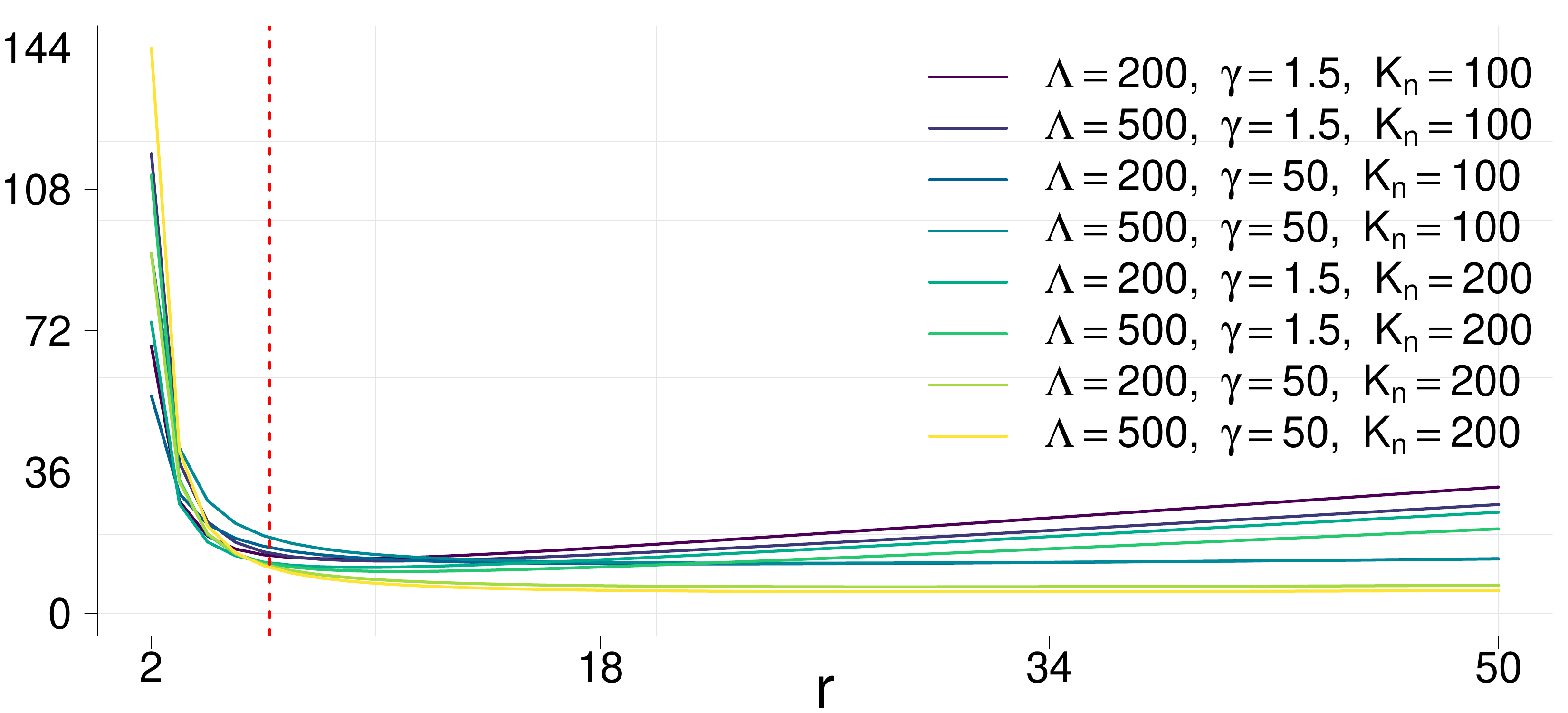}
    \hfill
    \includegraphics[width=0.32\linewidth]{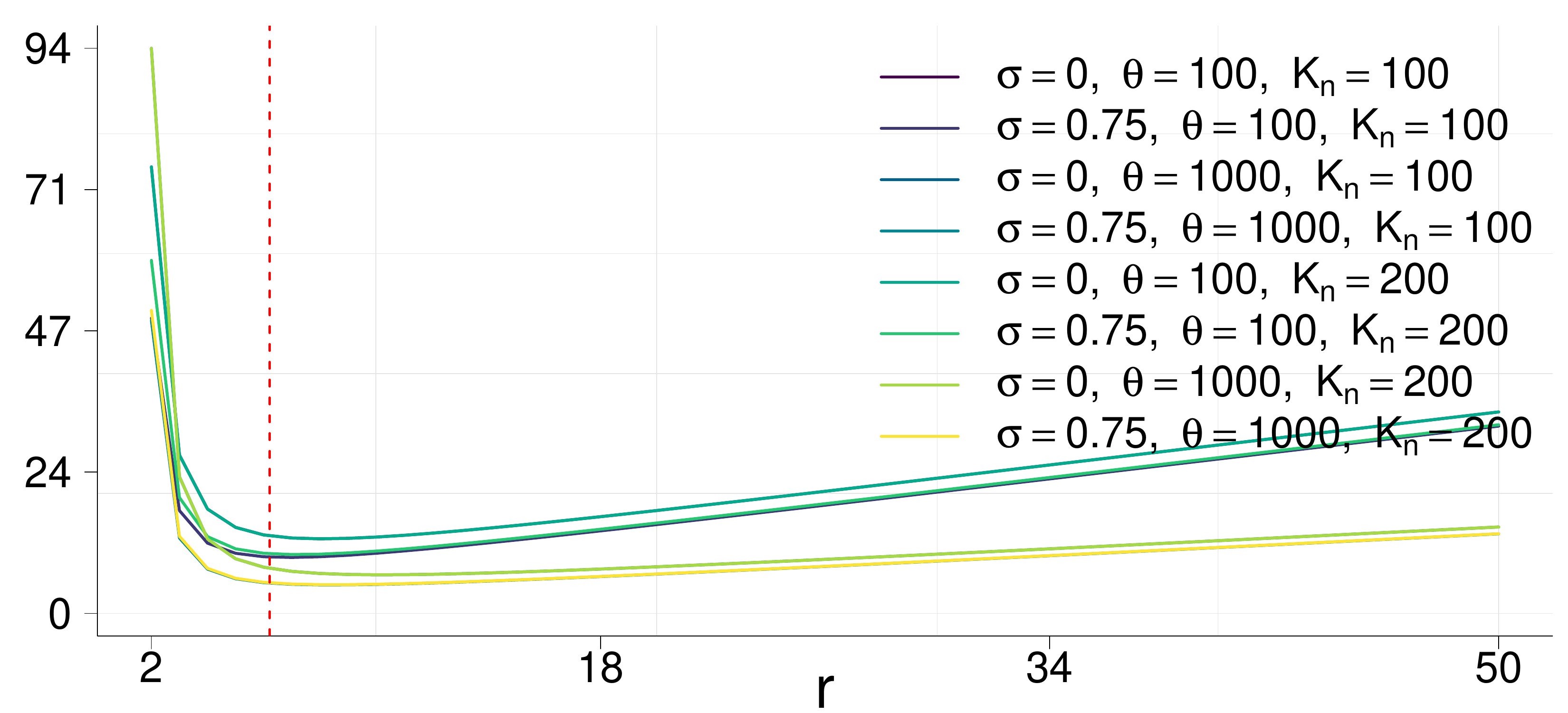}\\
    \includegraphics[width=0.32\linewidth]{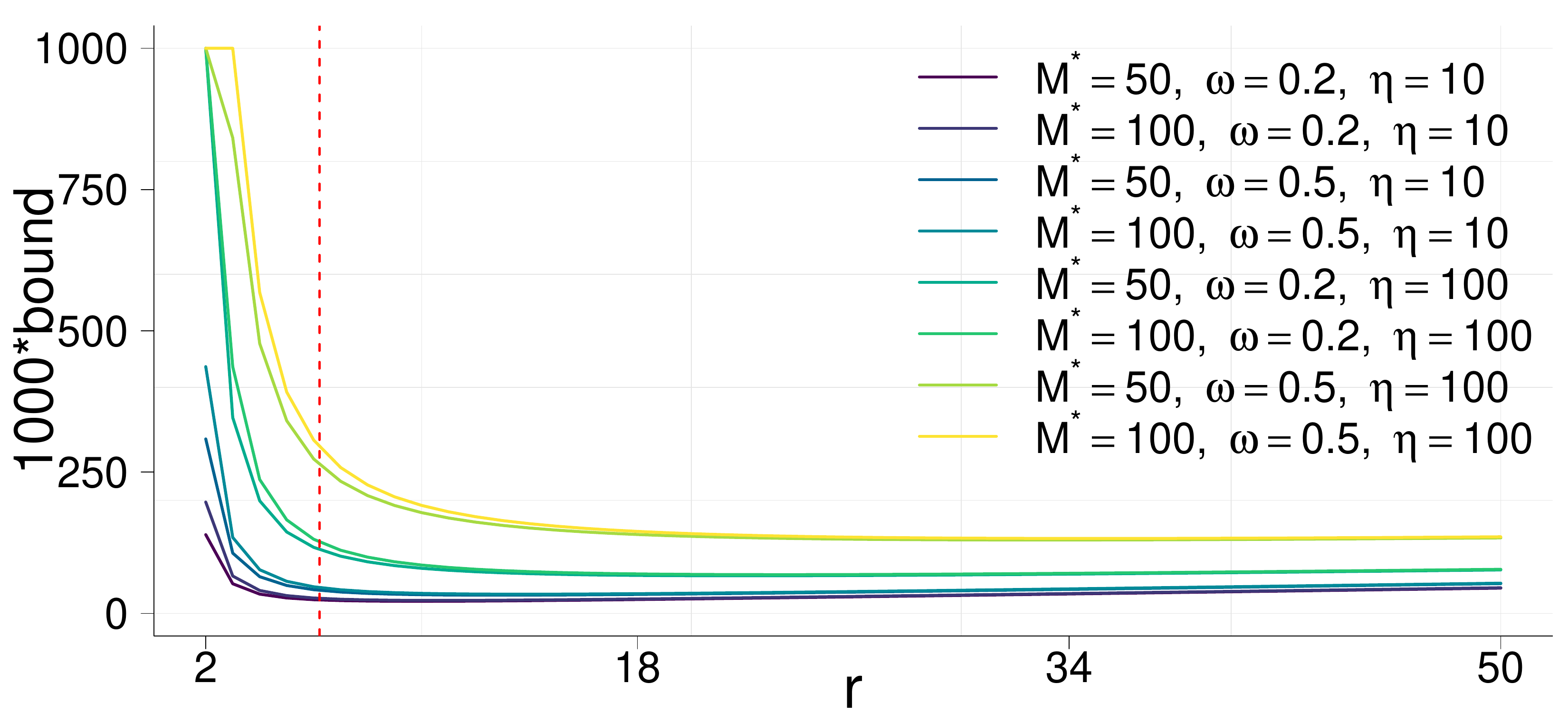}
    \hfill
    \includegraphics[width=0.32\linewidth]{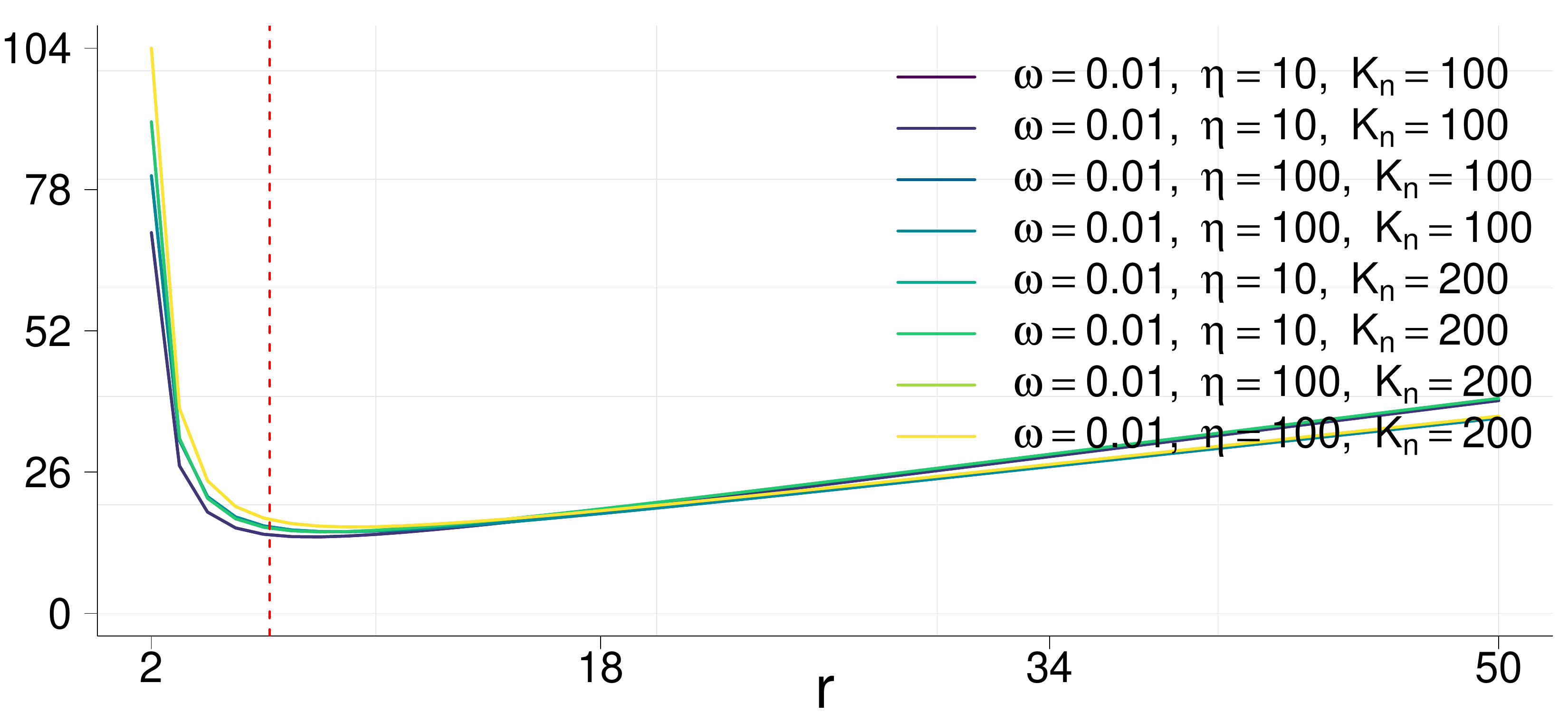}
    \hfill
    \includegraphics[width=0.32\linewidth]{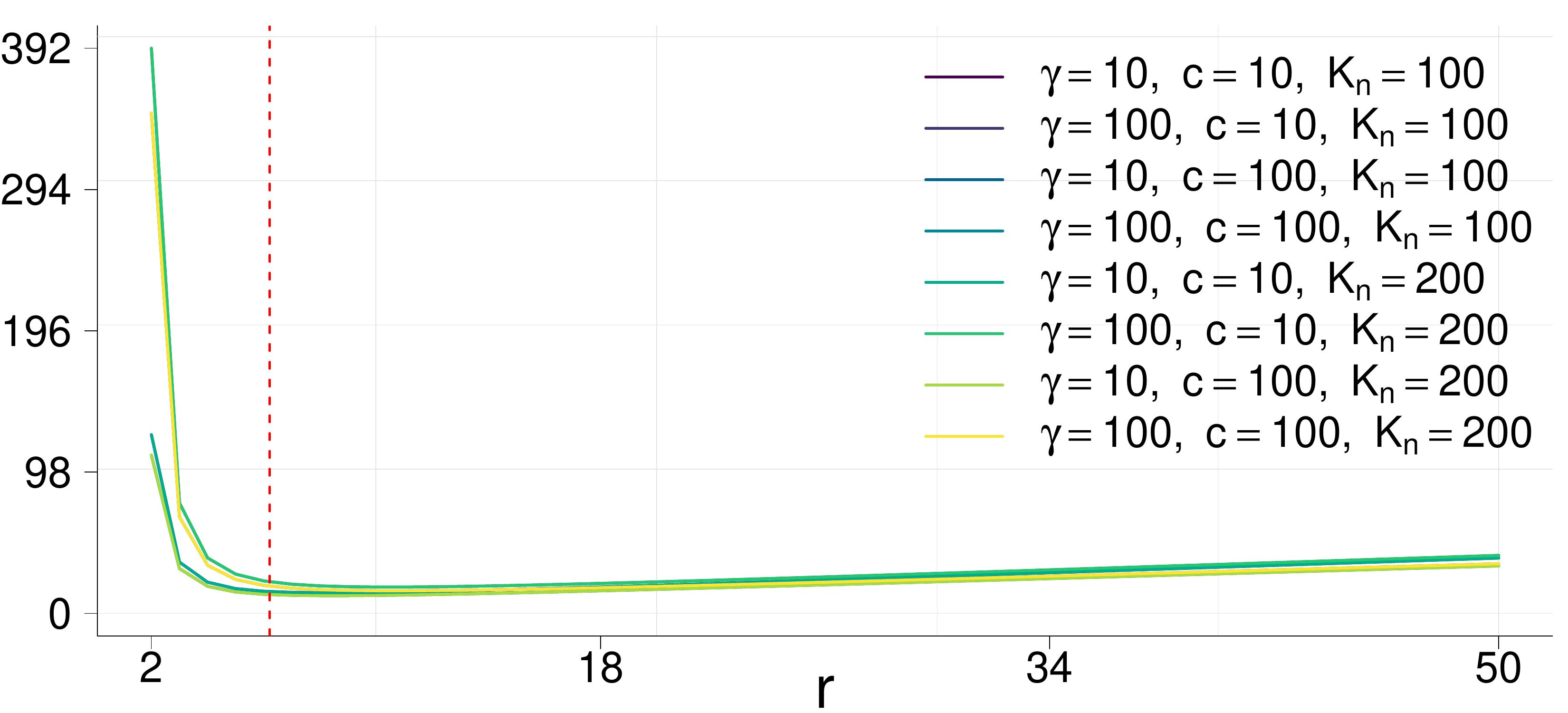}
    \vspace{-5pt}
    \caption{Credible upper bounds $U_r^{\operatorname{FD}}$ (top-left), $U_r^{\operatorname{FDP}}$ (top-middle), $U_r^{\operatorname{PYP}}$ (top-right),
    $U_r^{\operatorname{FB}}$ \smash{(bottom-left)}, $U_r^{\operatorname{N-MBP}}$ (bottom-middle) and $U_r^{\operatorname{IBP}}$ (bottom-right) displayed as a function of $r \geq 1$, for different hyperparameters. The dotted vertical line represents the starting point $\lceil \log n \rceil$ of the optimization routine. The vertical axis is rescaled by a factor of $10^3$ for readability. }
    \label{fig:U_minimiz_r}
        \vspace{-10pt}
\end{figure}

 \begin{proposition}
    \label{prop:uniquness_minimizer}
    Let $U_r(X^n)$ be any of the credible upper bounds minimized in Theorems~\ref{thm:Species_UB_FDP_PD}--\ref{thm:Features_upper_bounds}~as~a function of $r$.   Then,~the map $r\mapsto U_r(X^n)$ is convex, and $U_r(X^n)$ admits~a~unique~global~minimum in $\mathbb{R}_{\geq 1}$, which is either attained by a single integer $r^\star$ or two consecutive integers.
\end{proposition}

Notice that, although the order $r$ is restricted to the positive integers, Proposition~\ref{prop:uniquness_minimizer} is stated and proved by treating $U_r(X^n)$ as a function of a generic $r \in \mathbb{R}_{\geq 1}$. This choice is mainly motivated by technical considerations. More importantly, the possible existence of two global minima poses no difficulty within our setting, since we are only interested in the minimum value, rather than in the minimizer itself. Notably, Proposition \ref{prop:uniquness_minimizer} also holds under a FDP prior without any additional assumption on \(q_{M^\star}\) beyond being a probability mass function. Figure \ref{fig:U_minimiz_r} provides graphical support to  Proposition \ref{prop:uniquness_minimizer}  by displaying the credible upper bounds $U_r(X^n)$ as functions of $r$ for $n=500$ and different hyperparameters. In practice, we solve the problem by exploiting the unimodality of the sequence \(r \mapsto U_r(X^n)\), for $r \geq 1$,  with initial guess $\lceil \log n \rceil$, which is a first-order approximation of the minimizers of \eqref{eqn:PD_UB_species} and \eqref{eqn:UB_IBP} in the asymptotic regime $n \to \infty$; see Lemmas \ref{lem:r-minimizer} and \ref{lem:ibp-asymptotics-theorem6}.~The effectiveness of this initial value is also illustrated by the vertical dotted line in Figure~\ref{fig:U_minimiz_r}.~We~then evaluate the bound at the neighboring integers and move in the direction where $U_r(X^n)$~decreases. The search proceeds one step at a time until the first local increase is observed. By unimodality, that point is a global minimizer, up to a possible tie with the adjacent integer.

As  for the prior hyperparameters, these quantities  encode the global characteristics~of~the~data-generative mechanism, such as support size, evenness, and tail behavior. Therefore,~the~corresponding values should be elicited in order to provide an accurate characterization of the observed sample. This is essential for ensuring reliable inference for $\Mmax$~based~on~credible~intervals~obtained under a Bayesian model that is coherent with the data under analysis.~To~this~end,~we~rely on a plug-in empirical Bayes strategy that estimates hyperparameters~from~the~observed~sample. Specifically, let $\bm{\phi}$ be the vector of hyperparameters underlying the assumed Bayesian model, we set  $\bm{\phi}$ at the maximum a posteriori defined as
\begin{equation}\label{eqn:species_EB}
\hat{\bm{\phi}}=\operatorname*{argmax}\nolimits_{\bm{\phi}\in\Phi}\left\{\log \Pi(n_1,\ldots,n_{K_n}\mid \bm{\phi})+\log \Lcr(\bm{\phi};v)\right\},
\end{equation}
where $\Pi(n_1,\ldots,n_{K_n}\mid \bm{\phi})$ is the marginal distribution of the frequency spectrum induced by \eqref{eqn:Definetti_model} or \eqref{eqn:fea_bnp_model_def} under $\mu$, while  $\Lcr(\bm{\phi};v)$ is a suitably selected hyperprior for $\bm{\phi}$, viewed~here~as~a~regularizer. This choice facilitates further localization of posterior inference around distributions~that~are~both plausible under the prior and supported by the data, thereby downweighting~regions~of~the~model space that predict the observed occupancy profile $n_1,\ldots,n_{K_n}$ poorly. Such a localization helps in~further~sharpening the credible upper bound for $\Mmax$, by incorporating data information also via the estimated hyperparameters.

We highlight that the penalty term $\Lcr(\bm\phi;v)$ is especially needed when the marginal likelihood is nearly flat. This may occur, for instance, in the species setting under the FD and FDP priors when the observed sample is compatible with a nearly uniform distribution over the alphabet. In that regime, large changes in $\gamma$ may produce negligible improvements in the fit while substantially decreasing the credible upper bound. The regularization in \eqref{eqn:species_EB} stabilizes the fit and prevents~such overly optimistic choices of $\gamma$.  We did not observe analogous instability phenomena in the feature setting, and therefore no additional penalization was introduced there.
{This is also in line with our assumptions in Theorem \ref{thm:freq-validity-features}.} Further discussion on the choice of $\Lcr(\bm{\phi};v)$ is provided in Section~\ref{app:hy_elicitation} of the Supplementary Material. Finally, let us emphasize that while a fully Bayesian treatment~is, in principle, possible, it would introduce an additional layer of integration within the model, thus sacrificing the analytical simplicity of the credible intervals for $\Mmax$ in 
Theorems \ref{thm:Species_UB_FDP_PD}--\ref{thm:Features_upper_bounds}.

Before assessing the empirical performance of the methods developed in Sections~\ref{sec_2_mod} and \ref{sec_3_conf}, let us conclude by providing details on the use of such methods for addressing a fundamental problem that arises in the sequential sampling design inherent to both species and features settings. This problem consists in defining principled {\em stopping rules}, which determine when collecting further samples no longer yields relevant additional information about the underlying distribution over the alphabet beyond that already contained in the observed data. Despite its practical relevance and impact, the design of stopping rules has been theoretically studied and formalized only recently in \citet{ColombiFreq} through the use of the modal missing probability as a more  robust and informative summary than the classical missing mass. While this contribution has the merit of  providing the methodological foundations for addressing such a problem, it relies on a worst-case scenario analysis with focus on the features setting only. Leveraging our unified framework,  these developments can be broadened to encompass both the species and feature settings, while improving the effectiveness  of the resulting stopping rules through the more refined Bayesian model-based approach to inference on modal missing probabilities  we propose in Sections~\ref{sec_2_mod}--\ref{sec_3_conf}. More specifically, extending \citet{ColombiFreq} to our more general framework, yields stopping rules that prescribe terminating  collection of further data after $n^\star$ samples, with $n^*$ defined as
\begin{eqnarray}
    \label{eqn:Nstop_definition}
    n^\star := \min\left\{\,n\geq 1 \,:\, U(X^n) < \varepsilon\right\},
\end{eqnarray}
where, $U(X^n) $ is any of the upper bounds for $\Mmax$ derived in Theorems \ref{thm:Species_UB_FDP_PD}--\ref{thm:Features_upper_bounds}, while $\varepsilon$ represents an interpretable user-specified threshold  denoting the  largest modal missing probability~that~one is willing to tolerate. As such, stopping at $n^*$ implies that, with posterior probability at least $1-\alpha$, all the unseen labels have a population frequency at most $\varepsilon$. This provides a  rigorous quantitative criterion for declaring that the prevalence of any yet-unobserved label is practically negligible, with~the~notion~of~negligibility~determined by the user through the choice of $\varepsilon$.

\vspace{25pt}
\section{\large 6. Illustrative Simulation}\label{section:sim_study}
In the following experiments we assess the performance of the model-based credible intervals~for $\Mmax$ proposed in Sections~\ref{sec_2_mod}--\ref{sec_3_conf}, and compare the results with those of existing state-of-the-art solutions arising from the worst-case analyses in \citet{Painsky24} and \citet{ColombiFreq}. To this~end, we simulate $n$ i.i.d.\ realizations  from either model \eqref{eqn:Definetti_model} (species) or \eqref{eqn:fea_bnp_model_def} (features), under different specifications of the true underlying distribution $\mu_0$ at varying alphabet size $M_0$.

In the species setting, we set $n=500$ and let $\mu_0$ correspond to four benchmark data-generating mechanisms employed in e.g., \citet[][]{elena2025comprehensive,Painsky24}.~Specifically,~we~consider (i) a Zipf distribution, with  $p_{0,j}\propto j^{-2}$, (ii)  a geometric law yielding $p_{0,j}\propto 0.9^j$,~(iii)~a~negative binomial, with $p_{0,j}\propto [\Gamma(j+5)/(\Gamma(5)j!)] 0.003^5(1-0.003)^j$,~and,~finally,~(iv)~a~uniform, where $p_{0,j}=1/M_0$. Similarly, in the features setting we let  $n=2000$ and assess performance~under (i) a Zipf specification with  $p_{0,j}= (j+1)^{-2}$, (ii) a geometric one where $p_{0,j}=0.75^j$,~and~(iii)~a~uniform formulation setting $p_{0,j}=10^{-3}$. Recall that, as discussed within Section~\ref{sec:features}, in the feature~case the label probabilities need not sum to one over the alphabet. 

\begin{table}[t!]
\centering
\renewcommand{\arraystretch}{1.2}
\fontsize{7.5}{9.5}\selectfont
\setlength{\tabcolsep}{2pt}
\begin{tabular}{r @{\hspace{9pt}}
@{}c@{\hspace{5pt}}c@{\hspace{5pt}}c@{\hspace{5pt}}c@{\hspace{10pt}}
@{}c@{\hspace{5pt}}c@{\hspace{5pt}}c@{\hspace{5pt}}c@{\hspace{10pt}}}
\hline
$M_0$ & \multicolumn{4}{c}{Zipfs} & \multicolumn{4}{c}{Geometric} \\
\hline
 & PYP & FDP & FD & Freq & PYP & FDP & FD & Freq  \\
\hline
100  & {\bf 1.5 (0.99)} & 1.6 (0.99) & 1.6 (0.99) & 2.4 (1.00) & {\bf 1.5 (0.99)} & {\bf 1.5 (0.99)} & {\bf 1.5 (0.99)} & 2.0 (1.00)  \\
300  & {\bf 1.5 (0.99)} & 1.6 (1.00) & 1.6 (1.00) & 2.4 (1.00) &  {\bf 1.5 (0.99)} &  {\bf 1.5 (0.99)} &  {\bf 1.5 (0.99}) & 2.0 (1.00)  \\
500  & {\bf 1.2 (0.98)} & 1.3 (1.00) & 1.3 (1.00) & 1.9 (1.00) &  {\bf 1.3 (0.98)} & 1.4 (0.98) &  {\bf 1.3 (0.98)} & 1.8 (1.00)  \\
800  & {\bf 1.5 (0.99)} & 1.6 (1.00) & 1.6 (1.00) & 2.4 (1.00) &  {\bf 1.5 (0.99)} &  {\bf 1.5 (0.99)} &  {\bf 1.5 (0.99)} & 2.0 (1.00)  \\
1000 & {\bf 1.8 (0.99)} & 1.9 (1.00) & 2.0 (1.00) & 2.8 (1.00) &  {\bf 1.5 (0.99)} &  {\bf 1.5 (0.99)} & {\bf  1.5 (0.99)} & 2.0 (1.00)  \\
\hline
$M_0$ & \multicolumn{4}{c}{Uniform} & \multicolumn{4}{c}{Negative Binomial} \\
\hline
 & PYP & FDP & FD & Freq & PYP & FDP & FD & Freq  \\
\hline
100  &{\bf 1.2 (1.00)} & 1.3 (1.00) & 1.3 (1.00) & 1.4 (1.00) & {\bf 1.4 (1.00)} & 1.5 (1.00) & {\bf 1.4 (1.00)} & 1.7 (1.00)\\
300  &  3.5 (1.00) & {\bf 2.1 (1.00)} & {\bf 2.1 (1.00)} & 4.3 (1.00) & {\bf 1.1 (0.99)} & {\bf 1.1 (0.99}) & {\bf 1.1 (0.99)} & 1.3 (1.00) \\
500  & 5.0 (1.00) & {\bf 2.4 (1.00)} & 2.6 (1.00) & 7.2 (1.00) & 1.6 (1.00) & {\bf 1.3 (1.00)} & 1.6 (1.00) & 1.9 (1.00) \\
800  &  6.4 (1.00) & 3.0 (1.00) & {\bf 2.9 (1.00)} & 11.5 (1.00) & 2.7 (1.00) & {\bf 1.4 (1.00)} & 2.4 (1.00) & 3.7 (1.00) \\
1000 &  7.0 (1.00) & 3.4 (1.00) & {\bf 3.1 (1.00)} & 14.4 (1.00) & 3.5 (1.00) & {\bf 1.7 (1.00)} & 2.8 (1.00) & 5.3 (1.00)\\
\hline
\end{tabular}
\vspace{10pt}
\captionof{table}{Performance assessment in the species sampling setting. Comparison among the upper bounds of the $95\%$ one-sided intervals for $\Mmax$ obtained under the proposed model-based perspective (PYP, FDP, FD) and existing worst-case analyses (Freq), across different data-generating mechanisms and varying alphabet size $M_0$. Each entry reports the ratio between the derived upper bound and the oracle one obtained under the true data-generating mechanism (values in parentheses denote empirical coverage). Smaller ratios indicate sharper bounds. The bold values highlight the best-performing method among those achieving at least the nominal coverage level. Results are averaged over 500 independent replicated experiments.}
\label{tab:SS_species_table_unified}
\end{table}
\begin{table}[h!]
\centering
\renewcommand{\arraystretch}{1.2}
\fontsize{7.5}{9.5}\selectfont
\begin{tabular}{r @{\hspace{8pt}}
@{}c@{\hspace{2pt}}c@{\hspace{2pt}}c@{\hspace{2pt}}c@{\hspace{2pt}}c@{\hspace{10pt}}
@{}c@{\hspace{2pt}}c@{\hspace{2pt}}c@{\hspace{2pt}}c@{\hspace{2pt}}c@{\hspace{10pt}}}
\hline
$M_0$ & \multicolumn{5}{c}{Zipfs} & \multicolumn{5}{c}{Geometric}  \\
\hline
 & IBP & N-MBP & FB & Bdd & Ubd & IBP & N-MBP & FB & Bdd & Ubd \\
\hline
100  &  {\bf 1.4 (0.99)} & 1.5 (1.00) & {\bf 1.4 (0.99)} & 2.0 (1.00) & 2.2 (1.00) & 1.7 (0.99) & 1.9 (0.99) & {\bf 1.6 (0.99)} & 2.8 (1.00) & 3.6 (1.00) \\
1100 & {\bf 1.3 (0.99)} & 1.4 (1.00) & 1.4 (1.00) & 2.4 (1.00) & 2.0 (1.00) & {\bf 1.7 (0.99)} & 1.9 (1.00) & {\bf 1.7 (0.99)} & 3.8 (1.00) & 3.6 (1.00)  \\
3100 & {\bf 1.5 (0.99)} & 1.7 (1.00) & 1.6 (0.99) & 3.2 (1.00) & 2.4 (1.00) &{\bf  1.2 (0.99)} & 1.4 (1.00) & 1.3 (0.99) & 3.1 (1.00) & 2.7 (1.00)  \\
7600 & {\bf 1.4 (0.99)} & 1.6 (1.00) & 1.5 (1.00) & 3.2 (1.00) & 2.2 (1.00) & {\bf 1.2 (0.99)} & 1.4 (1.00) & 1.4 (1.00) & 3.4 (1.00) & 2.7 (1.00)  \\
9600 & {\bf 1.3 (0.99)} & 1.4 (1.00) & {\bf 1.3 (0.99)} & 3.0 (1.00) & 2.0 (1.00) & {\bf 1.7 (0.99)} & 1.9 (1.00) & 2.0 (1.00) & 4.6 (1.00) & 3.6 (1.00)  \\
\hline
$M_0$  & \multicolumn{5}{c}{Uniform} & \multicolumn{5}{c}{} \\
\hline
 & IBP & N-MBP & FB & Bdd & Ubd & &  &  &  &  \\
\hline
100  &  2.7 (1.00) & 1.3 (1.00) & {\bf 1.2 (0.96)} & 3.8 (1.00) & 3.5 (1.00) & &  &  &  & \\
1100 &   3.6 (1.00) & 1.7 (1.00) & {\bf 1.4 (1.00)} & 5.0 (1.00) & 4.4 (1.00) & &  &  &  & \\
3100 &  4.0 (1.00) & 2.1 (1.00) & {\bf 1.5 (1.00)} & 5.5 (1.00) & 4.8 (1.00)  & &  &  &  &\\
7600 &  4.3 (1.00) & 2.3 (1.00) & {\bf 1.6 (1.00)} & 6.0 (1.00) & 5.2 (1.00)  & &  &  &  &\\
9600 &  4.4 (1.00) & 2.5 (1.00) & {\bf 1.6 (1.00)} & 6.1 (1.00) & 5.3 (1.00)  & &  &  &  &\\
\hline
\end{tabular}
\vspace{10pt}
\captionof{table}{Performance assessment in the features sampling setting. Comparison among the upper bounds of the $95\%$ one-sided intervals for $\Mmax$ obtained under the proposed model-based perspective (IBP, N-MBP, FB) and existing worst-case analyses (Bdd, Ubd), across different data-generating mechanisms and varying alphabet size $M_0$. Each entry reports the ratio between~the derived upper bound and the oracle one obtained under the true data-generating mechanism (values in parentheses denote empirical coverage). Smaller ratios indicate sharper bounds. Bold values highlight the best-performing method among those achieving at least the nominal coverage level. Results are averaged over 500 independent replicated experiments.}
\label{tab:SS_features_table_unified}
\end{table}

Tables~\ref{tab:SS_species_table_unified}--\ref{tab:SS_features_table_unified} compare  the upper bounds of the $95\%$ one-sided intervals for $\Mmax$ obtained~under the proposed model-based perspective  and the existing worst-case analyses in \citet{Painsky24}~and \citet{ColombiFreq}, across the aforementioned data-generating mechanisms in 500 independent replicated experiments. In the species setting, we study, in particular, the credible upper  bounds $U^{\operatorname{FD}}$, $U^{\operatorname{FDP}}$ and $U^{\pityor}$ derived in Theorem \ref{thm:Species_UB_FDP_PD}, together with the worst-case confidence~bound~$U^{\operatorname{Freq}}$ of \cite{Painsky24}. In the feature case we focus, instead, on $U^{\operatorname{FB}}$, $U^{\operatorname{N-MBP}}$, and $U^{\operatorname{IBP}}$ from Theorem~\ref{thm:Features_upper_bounds} and two the frequentist counterparts $U^{\operatorname{Bdd}}$ and $U^{\operatorname{Ubd}}$ derived~in~\citet[]{ColombiFreq}~under bounded and unbounded alphabets, respectively. In Tables~\ref{tab:SS_species_table_unified}--\ref{tab:SS_features_table_unified}, performance is summarized by the ratio between each upper bound under analysis and the corresponding oracle counterpart computed under the true data-generating distribution $\mu_0$. Our  bounds in Theorems~\ref{thm:Species_UB_FDP_PD}--\ref{thm:Features_upper_bounds} are estimated as  in Section~\ref{sec:comp_and_infer} with the hyperpriors detailed in Section \ref{app:hy_elicitation} of the Supplementary Material.

To ensure a fair comparison, Tables~\ref{tab:SS_species_table_unified} and \ref{tab:SS_features_table_unified} also report the empirical coverage of the one-sided interval associated with each of the aforementioned upper bounds. This is computed~as~the~proportion of simulated datasets for which the interval contains the true value of $\Mmax$, and it should not fall below the nominal level of $95\%$.  Provided that this requirement is satisfied, the preferred one-sided interval is the one having the smallest, and hence more informative, upper bound.~Building on this criterion, the results  in Tables~\ref{tab:SS_species_table_unified}--\ref{tab:SS_features_table_unified} show that the proposed model-based procedures consistently produce  sharper upper bounds than the corresponding worst-case approaches, without compromising empirical coverage. These findings confirm that flexible model-based solutions localizing inference around data-generating mechanisms coherent with the observed data can produce more informative one-sided intervals for $\Mmax$ than worst-case procedures, while maintaining the same empirical coverage. The improved tightness results from the inclusion of both structural and data-driven information through the prior and model likelihood, whereas coverage preservation is facilitated by the Bayesian nonparametric formulation which retains sufficient flexibility to adapt to a broad class of data-generating mechanisms. This combination allows the proposed procedures to achieve the robustness of distribution-free worst-case solutions while delivering substantially more informative one-sided intervals. These results and conclusions were confirmed across multiple additional experiments under different settings of the true data-generating parameters. 

\renewcommand{\arraystretch}{1}
\vspace{20pt}
\section{\large 7. Organized Crime Application}\label{sec:appl}

\vspace{9pt}
\subsection{\large 7.1 Species setting}\label{sec:application_specie}
As anticipated in Section~\ref{s14}, the modal missing probability framework can play a central~role in law-enforcement decision-making by providing a principled assessment of the expected gains of investigative efforts. To illustrate the substantial benefits of this yet-unexplored perspective, we consider data extracted from the law-enforcement operation named {\em Operazione Infinito} \citep{Calderoni2017}, whose goal was to monitor and then disrupt the branch of the {\em 'Ndrangheta} Mafia \citep[see, e.g.,][]{sergi2016ndrangheta} operating in the area of Milan, Italy. Recalling Section~\ref{s14}, the judicial documents resulting from the investigations provide information on  $n=118$ criminals along with the corresponding unique membership to  $K_n=15$ observed territorial subgroups~of~the organization, known as {\em locali}. These coordinated modules constitute the core organizational units of {\em 'Ndrangheta}, motivating substantial investigative efforts to identify the complete set of {\em locali}, with a specific focus on highly populated, and hence structurally~relevant,~ones. 

The above endeavor can be naturally formalized within a species sampling framework, where criminals are the observation units and the alphabet is given by the set of all \textit{locali} operating~in Milan area. A central question in this context is whether the investigations of {\em Operazione Infinito} uncovered all \textit{locali}, and, if not, the extent to which the unseen portion comprises highly populated relevant ones. This question is further motivated by the fact that, in  wiretapped conversations,~a member of the organization refers to the presence of $20$ \textit{locali}. Such a discrepancy suggests the possible existence of an unseen portion of the organization. To this end, the left panel of Figure~\ref{fig:Mod1_Acc_Nj} reports the number $n_j$ of detected affiliates  for the $j$th observed \textit{locale}, with $j=1, \ldots, K_n$. The right panel displays instead the corresponding sample accumulation curve, a standard ecological summary describing how the number of distinct observed species~evolves~with~sample~size~\citep{colwell2004}. Since the order of appearance is unavailable in the recorded data,~the~curve~is~obtained by averaging over repeated random permutations of the sample, with the shaded gray~bands representing the corresponding $95\%$ bootstrap intervals. Notice how the curve saturates rapidly. After observing $60$ criminals, the mean number of discovered \textit{locali} is already close to $14$, out of the $15$ identified. This suggests that the investigations were effective in uncovering the backbone of the organization, with few, if any, \textit{locali} remaining undiscovered. What remains unclear, however, is whether any possibly unseen \textit{locale} is highly populated and hence operationally relevant.

\begin{figure}[t!] 
\vspace{-19pt}
    \centering 
    \includegraphics[width=0.475\linewidth]{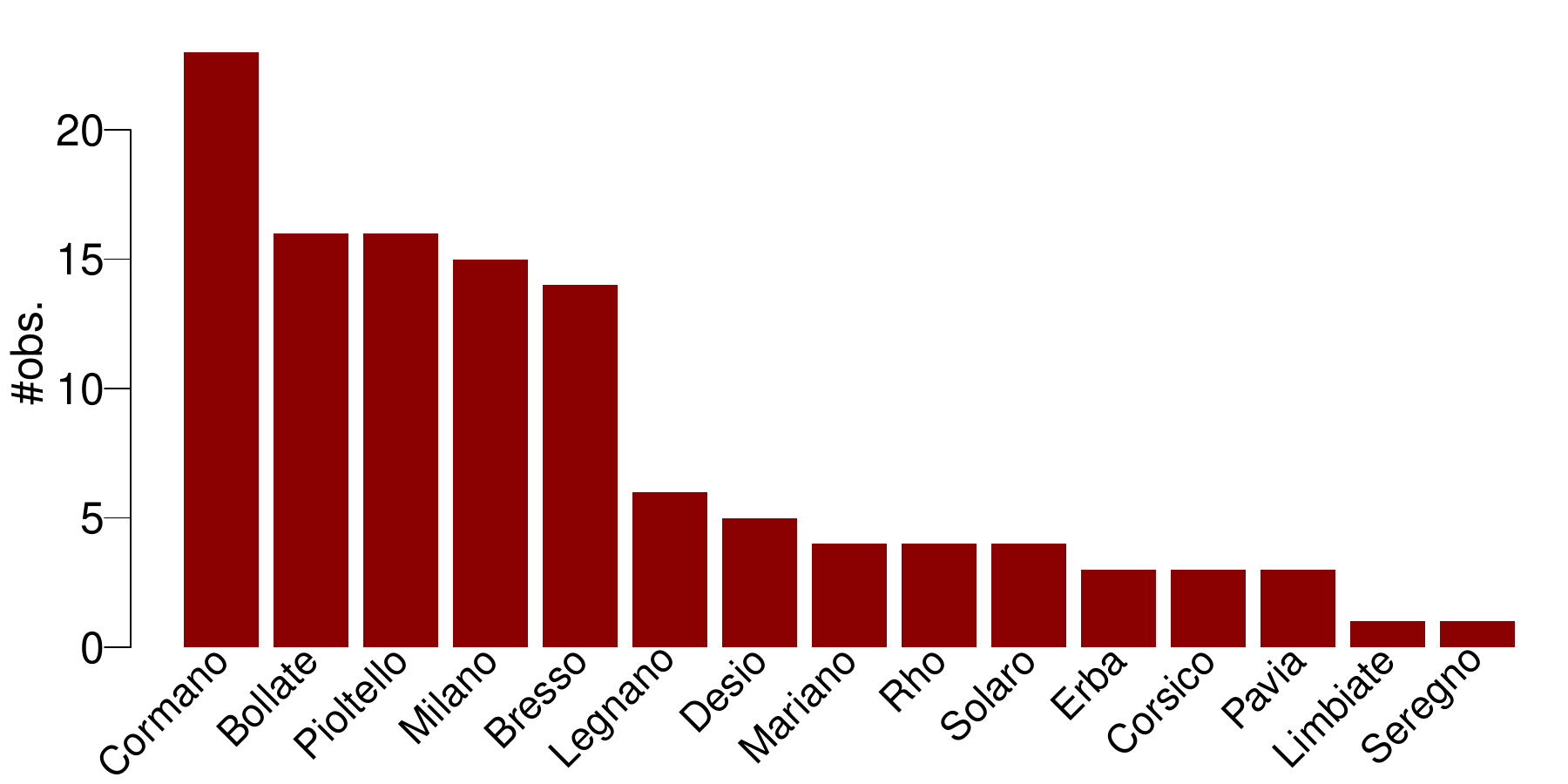} 
    \hfill 
    \includegraphics[width=0.475\linewidth]{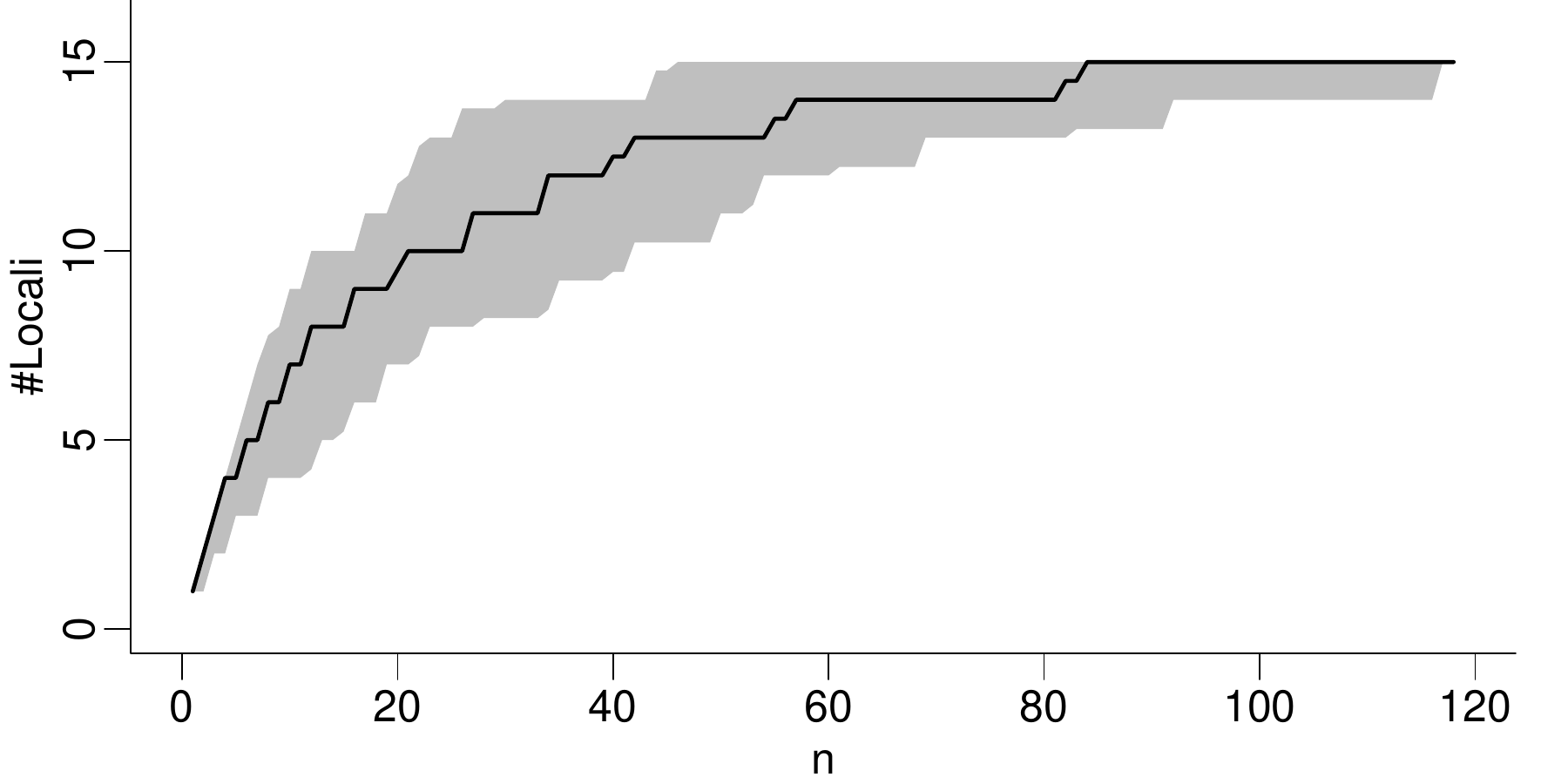} 
                \vspace{-5pt}
    \caption{Left: number of detected criminals for each observed \emph{locale}. Right: sample accumulation curve for the number of distinct observed \emph{locali}. The shaded bands represent $95\%$ bootstrap intervals.}
    \label{fig:Mod1_Acc_Nj} 
    \vspace{-10pt}
\end{figure}

We answer the above question by deriving credible upper bounds for $\Mmax$ under the proposed model-based perspective with $\pityor$, $\operatorname{FDP}$ and FD priors (see Section~\ref{sec:species} and Theorem~\ref{thm:Species_UB_FDP_PD}). The corresponding parameters are estimated via the empirical Bayes strategy described in Section~\ref{sec:comp_and_infer}, setting $M=20$ under the FD case,  in accordance with the~value mentioned in the judicial~documents. Before discussing the resulting credible upper bounds for $\Mmax$, it is important to assess whether the three Bayesian models employed provide an adequate representation of the observed data and hence yield reliable inference on $\Mmax$. To this end, under each of these models, we generate $1000$ samples from posterior of $\mu$ in Theorem \ref{thm:species_post}, and compare in Figure~\ref{fig:Crim4species_paramfit} the sampled \emph{locali} probabilities with the observed frequencies, sorting both in decreasing order. The posterior means and $95\%$ credible bands shown in Figure~\ref{fig:Crim4species_paramfit} are generally aligned with the observed frequencies, suggesting an adequate fit for all three Bayesian models, with a slightly reduced performance for  $\operatorname{FD}$. Hence, in this application, fixing $M=20$ is less suited than learning the alphabet size from the data. Consistent with this finding, under $\operatorname{FDP}$,~the~posterior~expected number of unseen~{\em locali} is $\E(M^\star \mid X)=2.66$, which yields a posterior on the total number $K_n+M^\star$ of \textit{locali} concentrated near, but below, the value of $20$ appeared in the judicial documents. 

\begin{figure}[t]
\vspace{-5pt}
    \centering
    \includegraphics[width=0.495\linewidth]{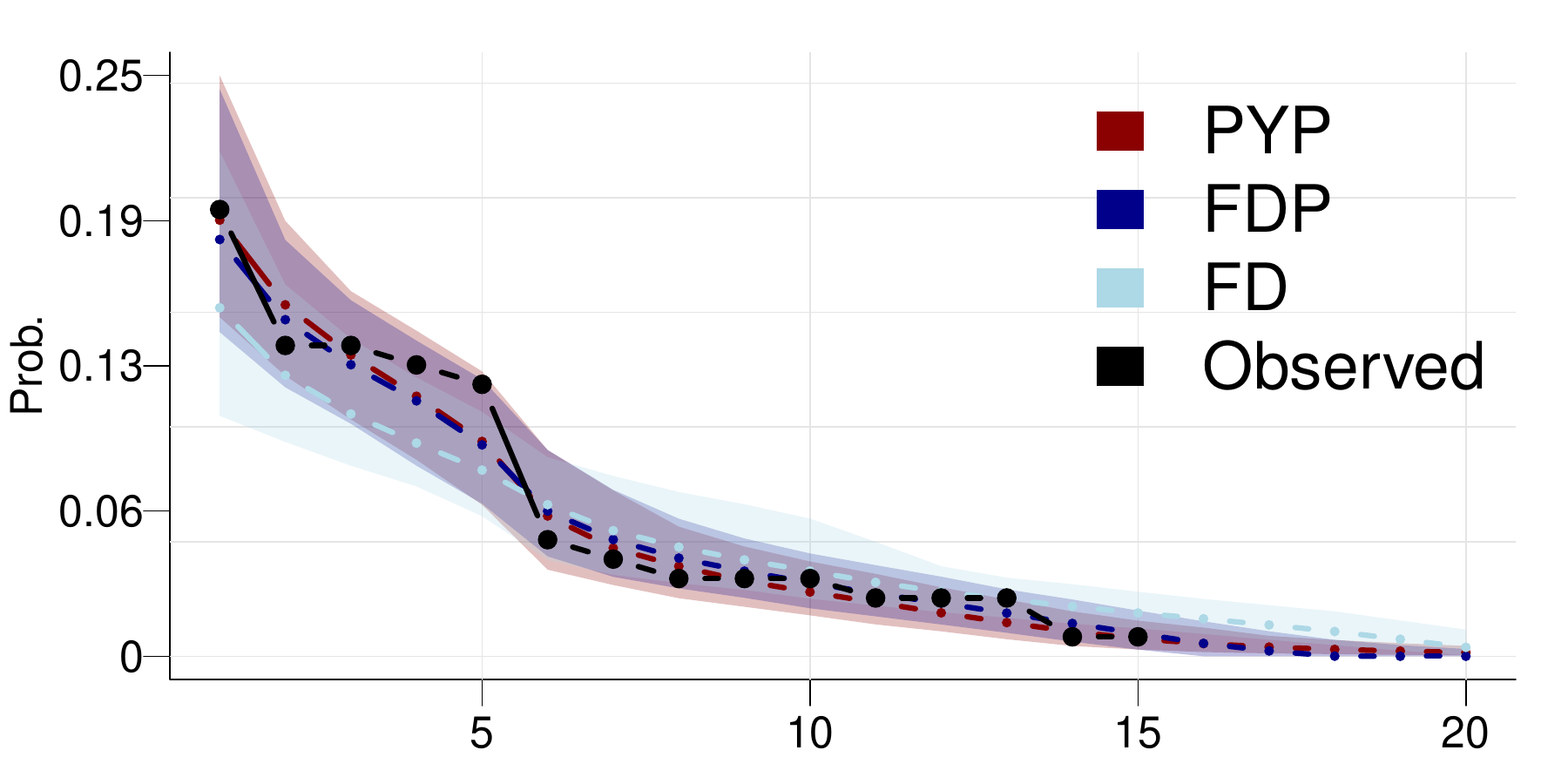}
    \hfill
    \includegraphics[width=0.475\linewidth]{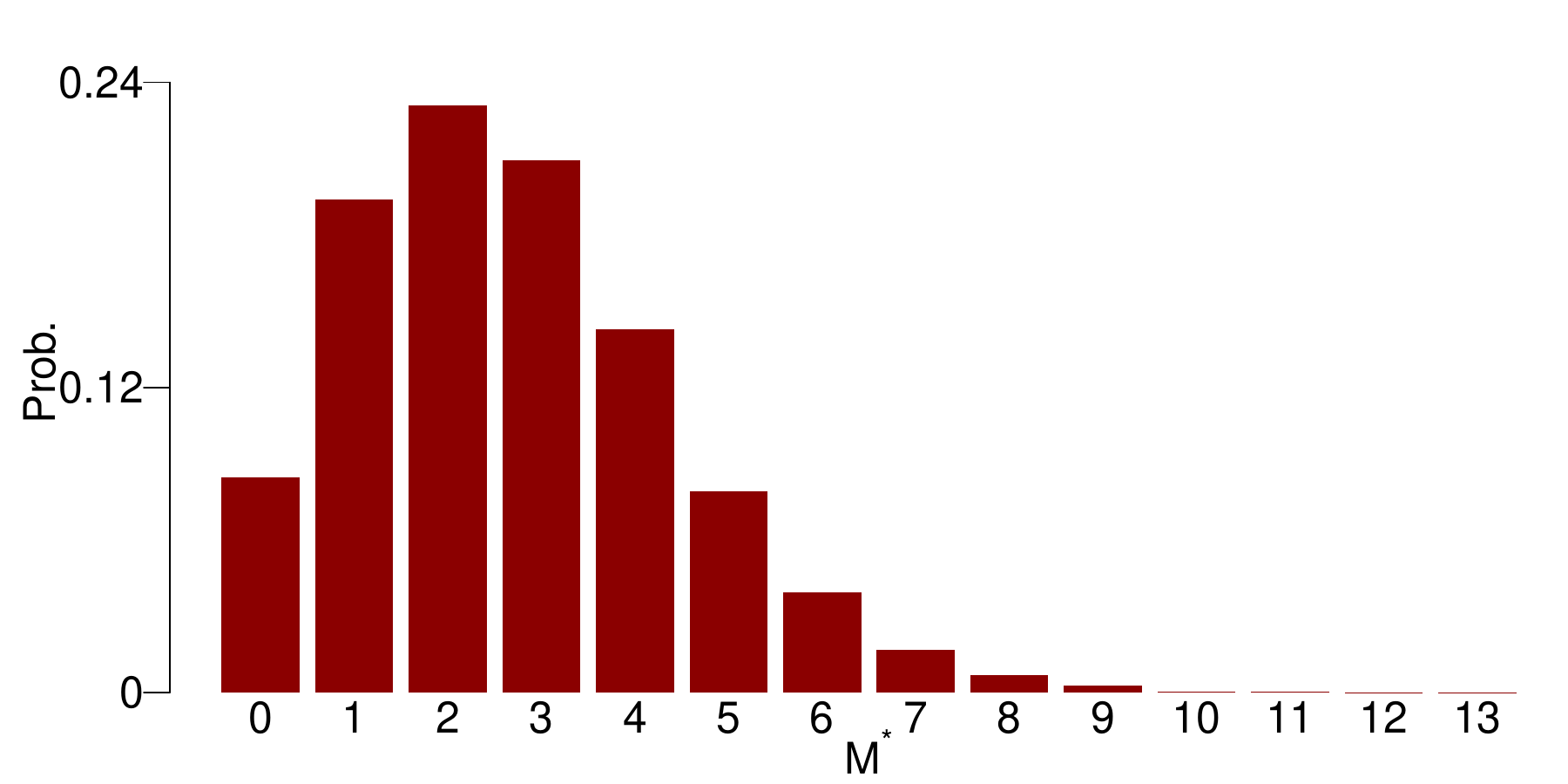}
            \vspace{-7pt}
    \caption{Left: Posterior means and $95\%$ credible bands of {\em locali} probabilities versus observed frequencies~(sorted in decreasing order).
    Right: Posterior distribution of the number of unseen \emph{locali} $M^\star$ under the FDP.}
    \label{fig:Crim4species_paramfit}
        \vspace{-2pt}
\end{figure}

The  above findings provide reassurance on the reliability of the $95\%$ Bayesian credible~upper bounds for $\Mmax$ inferred under three employed priors. These are
\[
( U^{\operatorname{FD}},\,U^{\operatorname{FDP}},\,U^{\pityor})
\,=\,
(0.0401,\,0.0398,\,0.0370)\,.
\]
Such values yield sharper, and hence more informative, one-sided intervals than the one obtained under the worst-case analysis of \cite{Painsky24}, which yields $U^{\operatorname{Freq}}=0.0495$.~In~this~criminology application, $U^{\operatorname{FD}}$, $U^{\operatorname{FDP}}$ and $U^{\pityor}$  can be interpreted as the upper bounds on the population frequency~of~the~largest unobserved \textit{locale}, which would account for at most $\approx 4\%$ of the {\em ‘Ndrangheta} members operating in the Milan area. To place this value into context, the left panel of Figure~\ref{fig:Crim4species_AlInf_SR} displays the $95\%$ confidence and credible one-sided intervals for the population frequency of both observed and unseen \textit{locali}. The unseen \textit{locali} share the same interval by construction and hence are reported only once. For the observed \textit{locali} we leverage instead standard inferential results for Dirichlet-distributed random variables in the Bayesian case, and Binomial confidence intervals within the frequentist setting. According to Figure~\ref{fig:Crim4species_AlInf_SR}, an unseen \textit{locale} with population frequency around $4\%$ is comparable in size with many of the observed \textit{locali}, and hence, not negligible. At the same time, the comparison among the intervals rules out the possibility that the investigation missed the most relevant \textit{locali}, since all the five most prevalent observed ones have population frequencies well above the modal missing probability. These findings provide further evidence of the effectiveness of the investigations underlying \textit{Operazione Infinito}.

We conclude the analysis under the species setting by illustrating the impact of the stopping rule discussed in Section \ref{sec:comp_and_infer} and formalized in \eqref{eqn:Nstop_definition}.  Here, $\varepsilon$ represents the population frequency~of the largest unseen \textit{locale} that law-enforcement considers negligible enough to be safely missed.~The right panel of Figure~\ref{fig:Crim4species_AlInf_SR} displays the optimal stopping sample size $n^*$, as a function of $\varepsilon$, resulting from our proposed upper bounds and the worst-case one of \cite{Painsky24}. The latter systematically requires a larger sample size to fall below $\varepsilon$ than any of our three model-based alternatives. For instance, taking $\varepsilon=0.05$ as a reference value, the stopping rules based on $U^{\operatorname{FD}}$, $U^{\operatorname{FDP}}$, and $U^{\pityor}$ indicate that the investigation could have been stopped after profiling about \smash{$90$ criminals}. By contrast, the stopping rule based on $U^{\operatorname{Freq}}$ would have required monitoring roughly $30$ additional affiliates. This provides a concrete illustration of the practical impact of the sharper bounds we derive for the modal missing probability.

\begin{figure}[t]
\vspace{-15pt}
    \centering
    \includegraphics[width=0.495\linewidth]{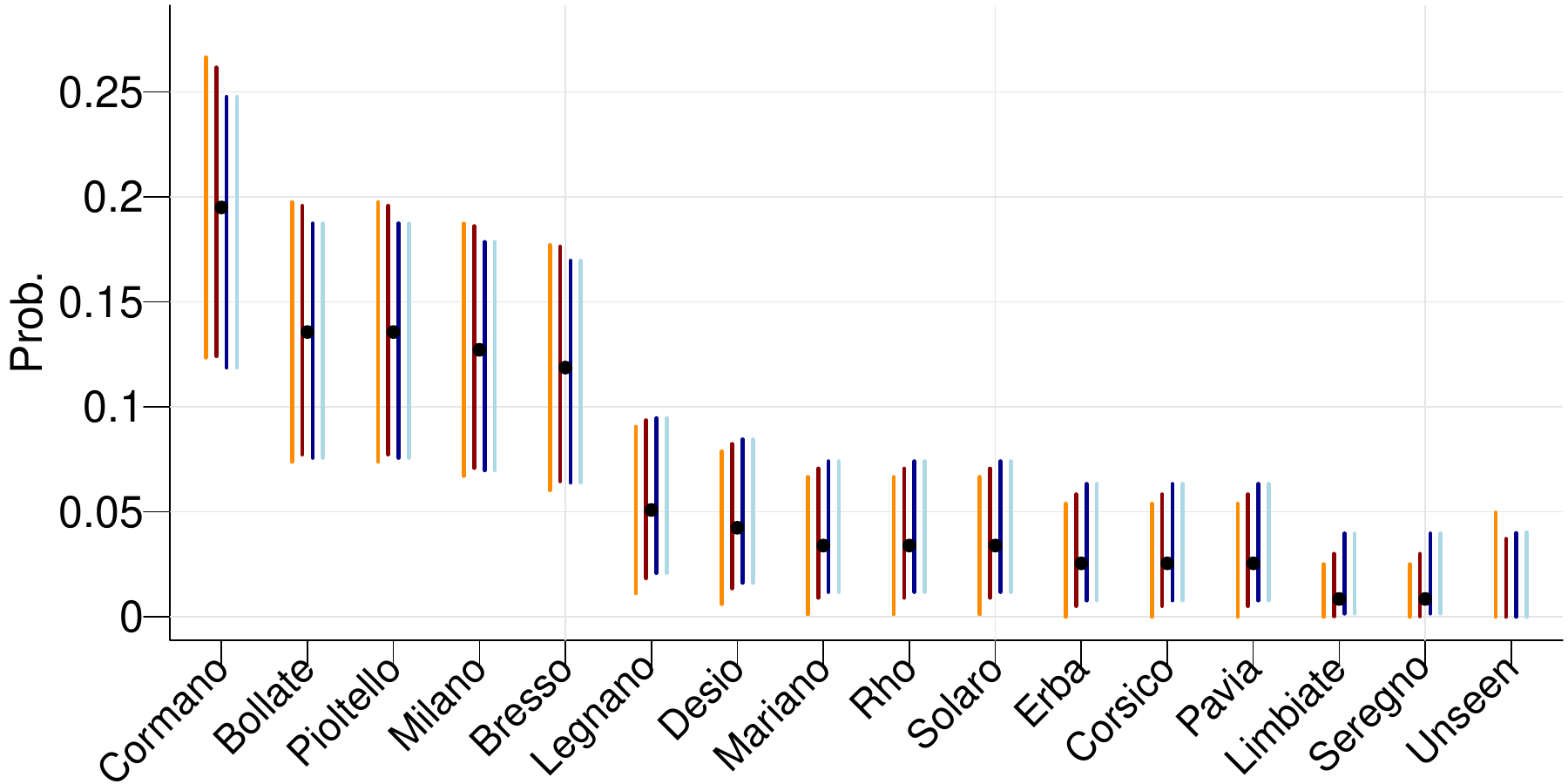}
    \hfill
    \includegraphics[width=0.495\linewidth]{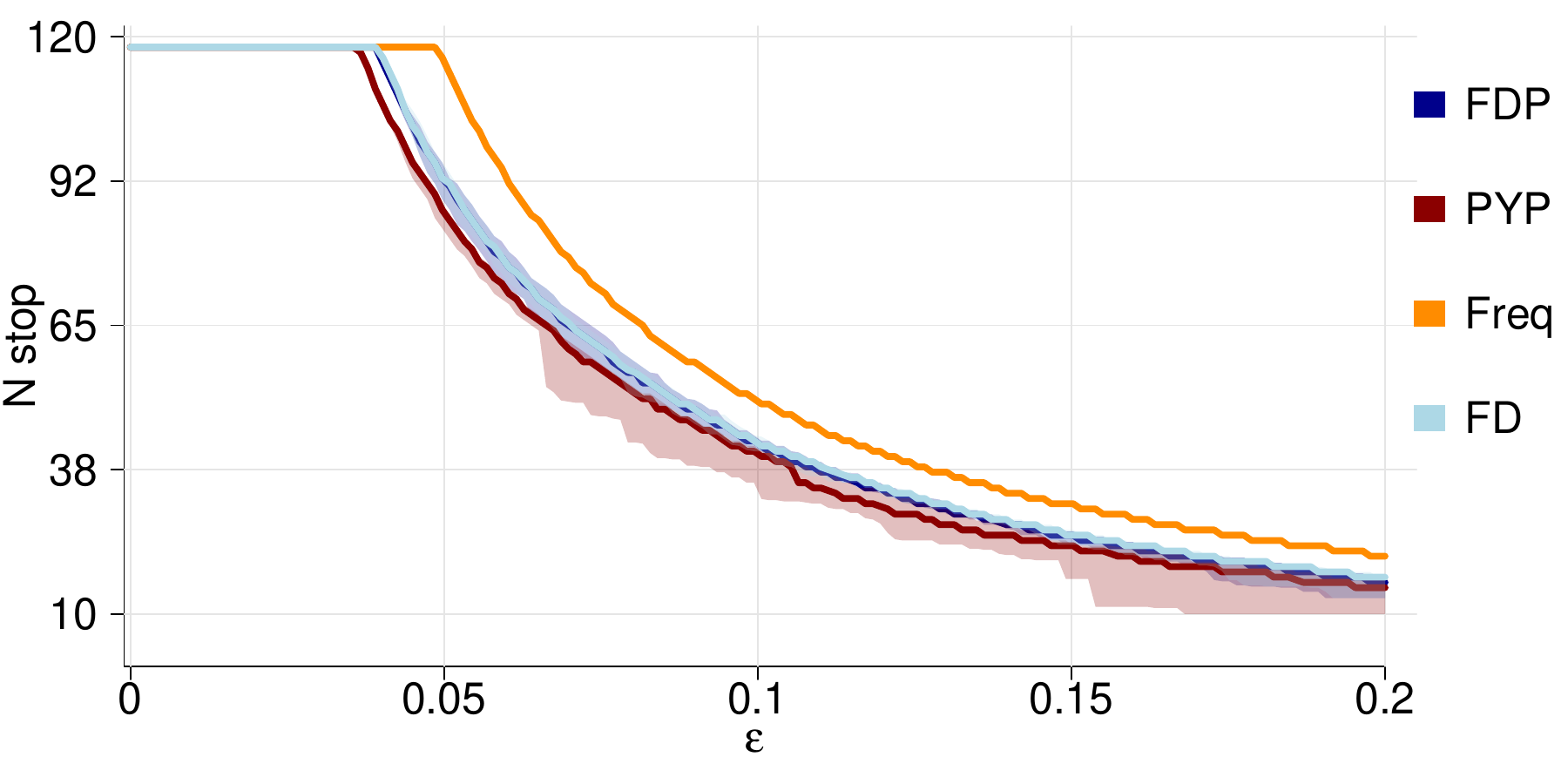}
            \vspace{-15pt}
    \caption{Left: Credible and confidence intervals for the population frequency of observed and unseen {\em locali}. Right: Estimated stopping sample size $n^*$ as a function of the threshold $\varepsilon$.}
    \label{fig:Crim4species_AlInf_SR}
            \vspace{-12pt}
\end{figure}

\vspace{15pt}
\subsection{\large 7.2 Features setting}\label{sec:application_features}
Besides studying unseen modular subgroups, a central challenge in law-enforcement investigations is assessing the presence and relevance of yet-unobserved criminals. Addressing this more refined {\em boundary specification problem} \citep[e.g.,][]{campana2022studying,diviak2022key} is fundamental for understanding the extent to which the criminal organization propagates beyond what is observed during the investigations. Within {\em Operazione Infinito}, this objective can be naturally addressed under a feature sampling framework by leveraging information on the participation of the $K_n=118$ observed criminals to the $n=47$ monitored summits of the organization. Here, the summits constitute the observational units, while criminals correspond to the features, each of which~may or may not be observed at a given summit. Therefore, within this setting, inference~on~the~modal missing probability  allows us to assess whether monitoring additional summits would reveal yet-unobserved active criminals with a high probability of appearing in {\em 'Ndrangheta} meetings.

Consistent with the above objective, the left panel of Figure~\ref{fig:Crim4features_Acc_Nj} reports, for each of the~$K_n$~observed criminals, the number of monitored meetings in which that criminal was found. Notice that nearly $40\%$ of the observed members of the organization appear in only one summits, a pattern suggesting that monitoring additional meetings may reveal previously unobserved criminals. This is consistent with the sample accumulation curve in the right panel of Figure~\ref{fig:Crim4features_Acc_Nj}, which has not yet stabilized and hence indicates that further observed meetings could still reveal~unseen~affiliates.

\begin{figure}[t!]
\vspace{-15pt}
    \centering
    \includegraphics[width=0.475\linewidth]{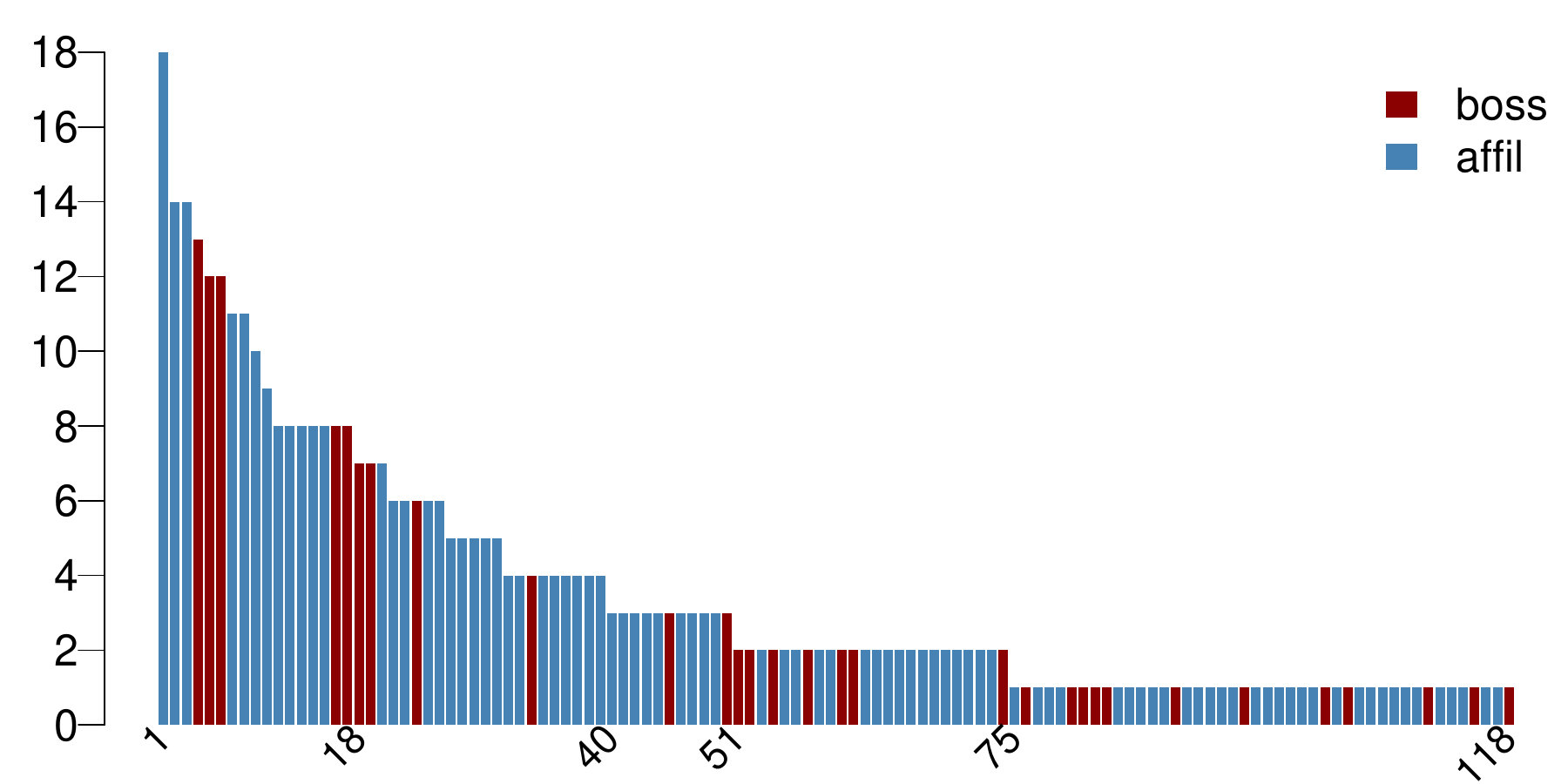}
    \hfill
    \includegraphics[width=0.475\linewidth]{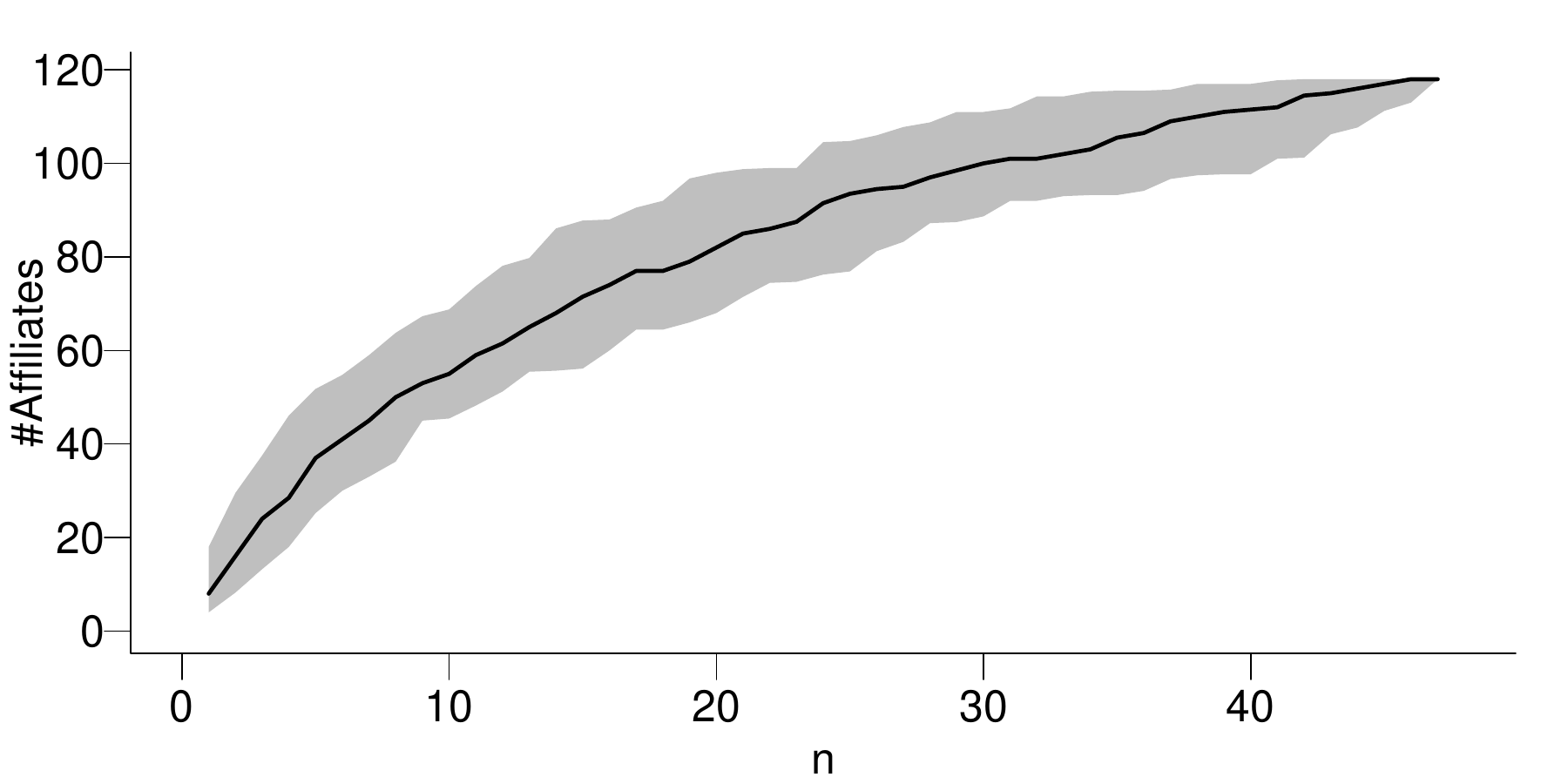}
    \vspace{-7pt}
    \caption{Left: number of summits attended by each observed criminal. Right: sample accumulation curve for the number of distinct observed criminals. The shaded bands represent $95\%$ bootstrap intervals.}
    \label{fig:Crim4features_Acc_Nj}
\end{figure}

The above results motivate inference on unseen criminals under the models we introduced~in Section \ref{sec:features}, with focus on FB, N-MBP and IBP priors  (see also Theorem~\ref{thm:Features_upper_bounds}). As for the species setting in Section~\ref{sec:application_specie}, also in this case the parameters of the different priors are estimated via the empirical Bayes strategy outlined in Section \ref{sec:comp_and_infer}, setting $M=200$ for the total number of affiliates in the FB case. Similarly~to Figure~\ref{fig:Crim4species_paramfit} in the species case, also Figure~\ref{fig:Crim4features_paramfit} confirms~that~these~Bayesian models achieve an accurate characterization of the distribution underlying the observed criminal appearance frequencies across the monitored summits, and therefore, are reliable in terms~of~inference. Notice how, under~the $\operatorname{N-MBP}$, the posterior expected number of unseen criminals is $\E(M^\star \mid X)=27.76$.~This~indicates that investigations might have missed $\approx 30$ affiliates. However, it is unclear whether these unseen criminals cover a peripheral role in the organization or are instead active members with~high~probability of participating in several summits. The answer to this question is encoded in the $95\%$ credible upper bounds for $\Mmax$ derived under the Bayesian models we consider. These are
\[
(
    U^{\operatorname{FB}},
    U^{\operatorname{N-MBP}},
    U^{\operatorname{IBP}}
)
=
(0.1255,\,0.1444,\,0.1226)\,,
\]
yielding sharper, and hence more informative, one-side intervals for $\Mmax$ than~the~frequentist bounds $(U^{\operatorname{Bdd}},U^{\operatorname{Ubd}})=(0.1873,0.1494)$ of \cite{ColombiFreq}. 

\begin{figure}[t]
\vspace{-15pt}
    \centering
    \includegraphics[width=0.495\linewidth]{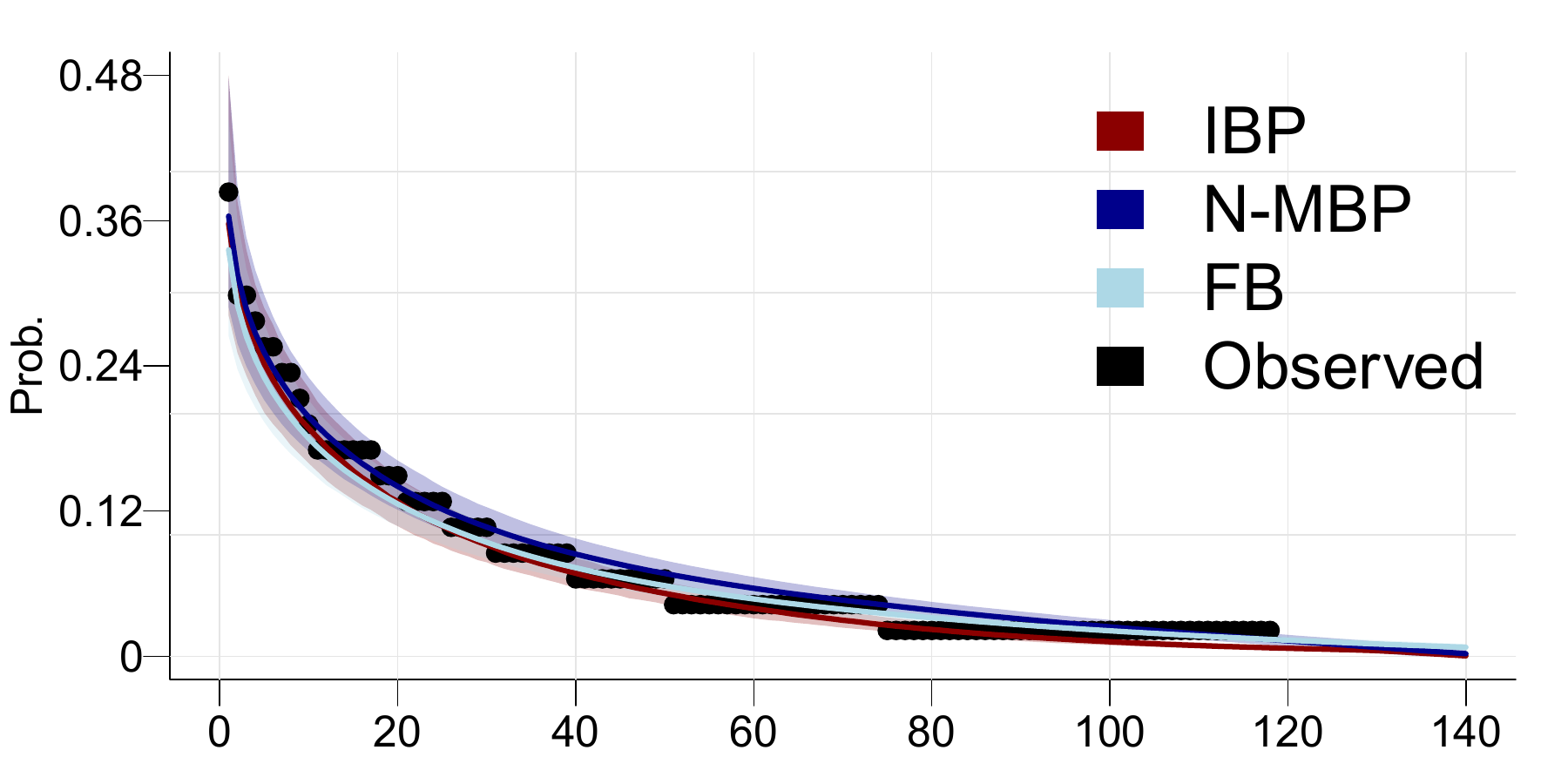}
    \hfill
    \includegraphics[width=0.495\linewidth]{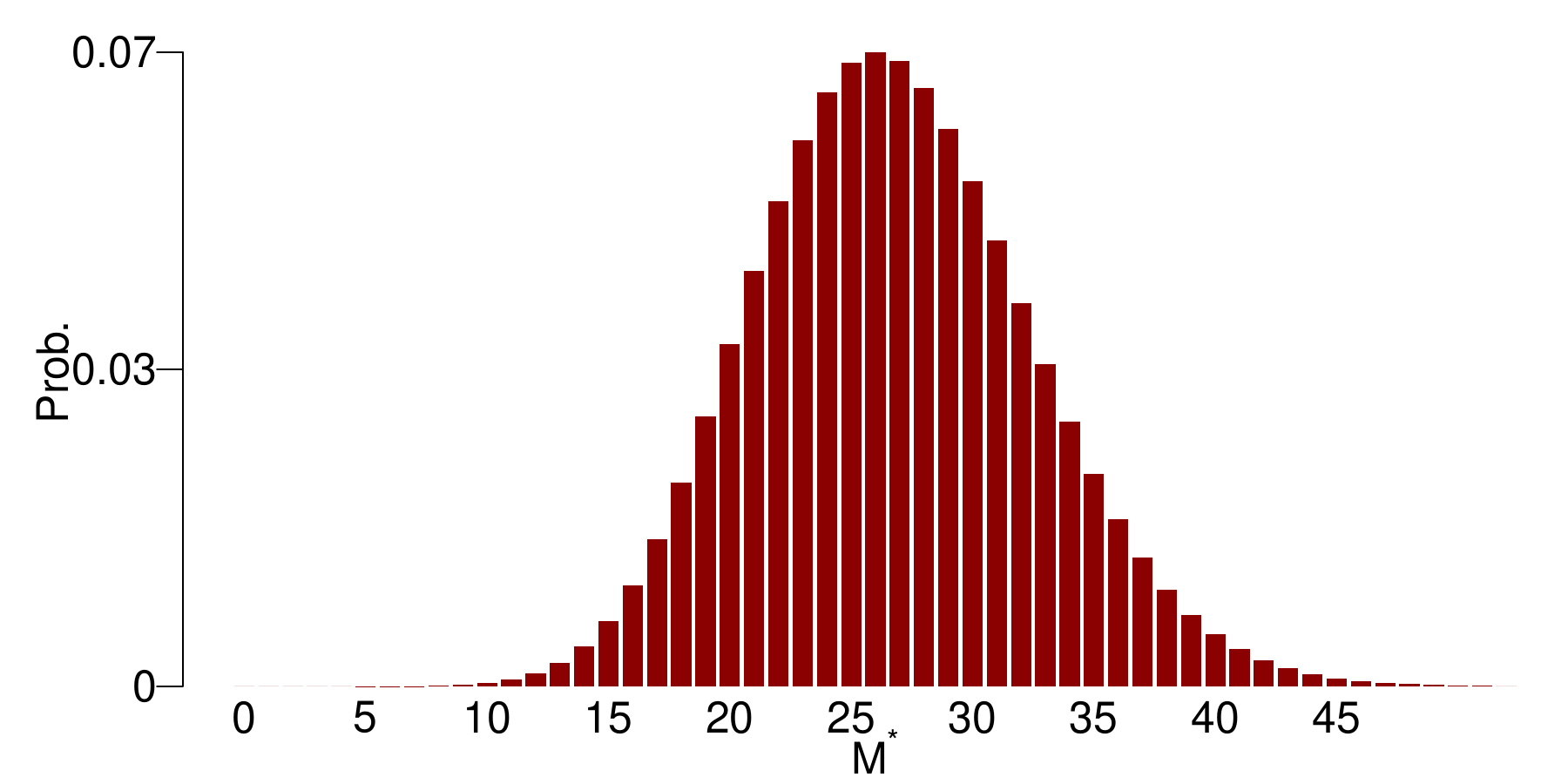}
        \vspace{-15pt}
    \caption{
 Left: Posterior means and $95\%$ credible bands of criminal appearance probabilities versus observed frequencies (sorted in decreasing order).
    Right: Posterior distribution of the number of unseen \emph{criminals} $M^\star$ under the N-MBP.}
    \label{fig:Crim4features_paramfit}
    \vspace{-5pt}
\end{figure}

\begin{figure}[b!]
    \centering
    \includegraphics[width=0.495\linewidth]{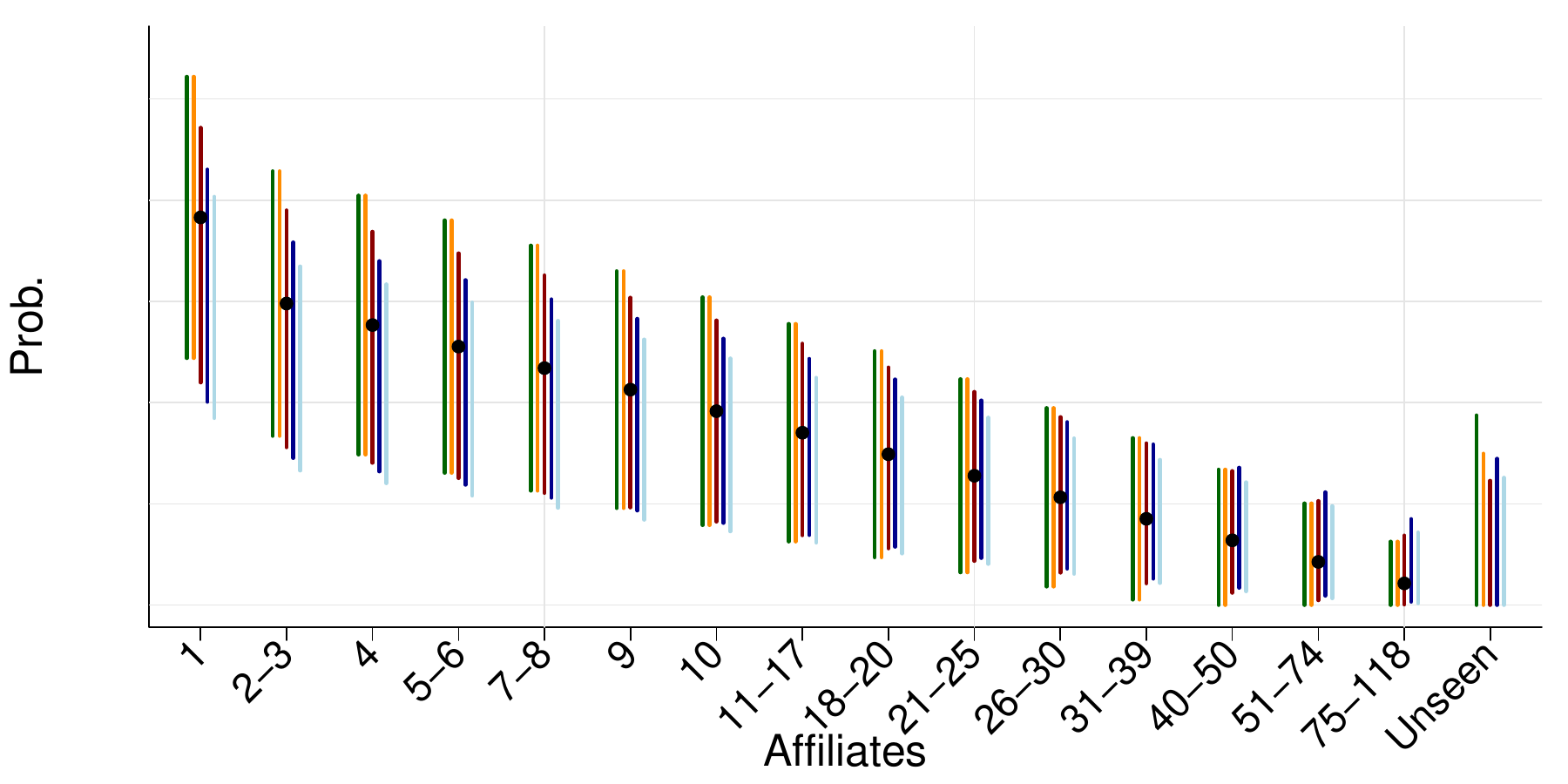}
    \hfill
    \includegraphics[width=0.495\linewidth]{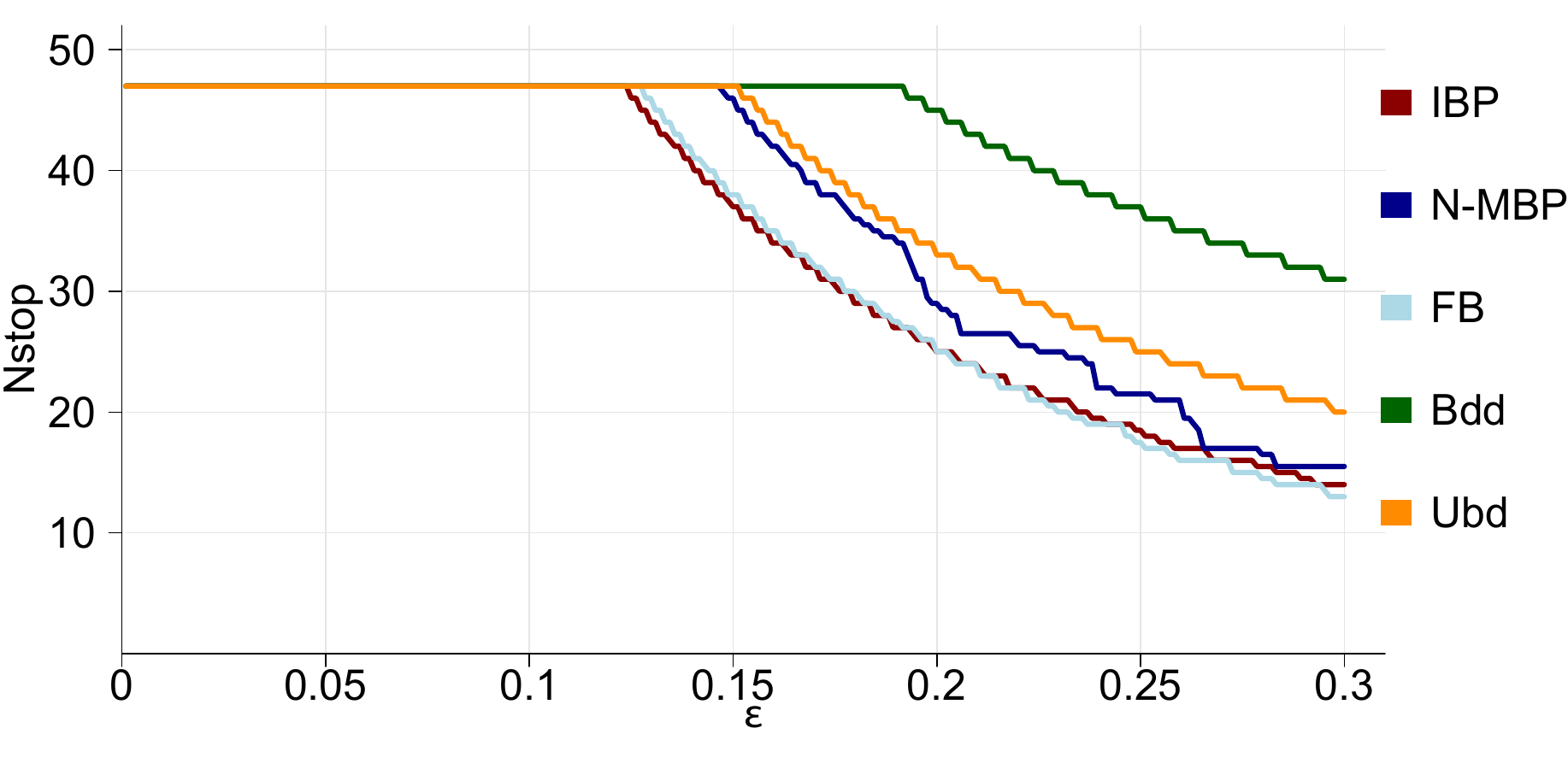}
    \vspace{-18pt}
    \caption{Left: Credible and confidence intervals for the attendance probabilities of both observed~and~unseen criminals (aggregated according to common probabilities). Right: Estimated stopping sample size $n^*$ as a function of the threshold $\varepsilon$.}
    \label{fig:Crim4features_AlInf_SR}
    \vspace{-5pt}
\end{figure}

To contextualize the above bounds on the modal appearance probability of unseen criminals, Figure~\ref{fig:Crim4features_AlInf_SR} compares the resulting one-sided intervals with the corresponding counterparts obtained for the observed criminals. The latter are derived leveraging similar methods and results yielding the intervals in Figure~\ref{fig:Crim4species_AlInf_SR} under the species setting. Given the large number of criminals, in Figure~\ref{fig:Crim4features_AlInf_SR} the members of the organization are grouped according to common appearance probabilities. As such, the last intervals correspond to those induced by the estimated upper bounds for the unseen criminals. In this specific context, such bounds provide an upper high credibility (or confidence) limit on the largest relative proportion of all population summits, both observed and yet-unobserved, attended by any unseen criminal. In  Figure~\ref{fig:Crim4features_AlInf_SR}~the~sharpest~of~these~bounds~is~the one provided by our model-based perspective under the $\operatorname{IBP}$ prior, which also achieves a sensible characterization of the observed attendance frequencies in Figure~\ref{fig:Crim4features_paramfit}. The~resulting~value~indicates that any unseen criminal (including the most active one), is estimated to have attended at most about $12\%$ of all population meetings, where the total number of meetings is likely larger than the $n=47$ observed ones.  According to Figure~\ref{fig:Crim4features_AlInf_SR}, the associated interval is similar~to~the one obtained for the observed criminals who participated in three meetings, a value~close~to~empirical mean of $3.5$ attended summits. This finding suggests that the set of unseen criminals may include individuals with a level of activity \smash{comparable to the one of the average observed criminal}. Consequently, extending the investigations and observing additional summits could have provided~valuable additional information by identifying unseen members playing a non-negligible participation role~in the organization.

As in Section~\ref{sec:application_specie}, we conclude by assessing in Figure~\ref{fig:Crim4features_AlInf_SR} the impact of the above upper bounds in terms of the induced stopping rule on the number of monitored summits. Similarly~to~the~species case, also in this features setting, the improved sharpness achieved by the proposed model-based bounds yields systematic gains also in terms of stopping rules compared to the frequentist counterparts. Focusing on the $\operatorname{IBP}$ bound and taking, for instance, $\varepsilon=0.15$, the induced stopping rule saves up to ten additional meetings from being inspected relative to the frequentist alternatives under the same threshold. As such, the practical gain is again substantial.

\vspace{5pt}
\section{\large 8. Discussion}
\label{section:discussion}
Inference in species and feature sampling problems has traditionally focused~on~the~total~probability mass associated with unobserved labels and its predictive implications.~More~refined~functionals of the distribution over unseen labels, such as the largest unobserved label probability, have received attention only recently, with existing methods largely relying on worst-case distribution-free analyses that yield overly-conservative intervals and substantial information~loss.~In~this~article we have shown both theoretically and empirically that these intervals can be substantially sharpened without sacrificing the robustness and frequentist coverage guarantees of the existing worst-case solutions. We achieve this through a model-based Bayesian nonparametric perspective that flexibly combines prior probabilistic structure with the information contained in the observed sample to localize uncertainty around data-generating mechanisms compatible with the data. This perspective yields  closed-form and sharper  upper bounds for the modal missing probability, leading tractable one-sided credible intervals  that are more informative than state-of-the-art ones while retaining frequentist coverage guarantees and robustness. The latter property stems from the flexibility of the Bayesian nonparametric framework, which can adapt to a broad class of data-generating mechanisms. These gains are supported empirically through simulation studies and a real-data application to organized crime, where our improved credible bounds enable a more refined quantification of the obscure boundaries of the monitored illicit organization and provide a basis for more informed law-enforcement decision-making during investigations.

The above advancements and promising results motivate several \smash{directions for future research}. An important one concerns relaxing our independence assumptions both across sampling~units~and labels. For example, in the organized crime application, summit attendance patterns may exhibit dependence, as the presence of a given criminal at a meeting may increase the probability~of~observing other members connected to that criminal, for instance through family ties. Although~the inclusion of these more complex dependencies lies beyond the scope of the present contribution, this example highlights the potential value of generalizing the proposed inferential framework~to dependent sampling models that incorporate exogenous information. Besides these extensions, another interesting line of research lies in inference for rare but already observed labels. Our~contribution focuses on unseen labels, yet those observed only once or a few times are also subject~to substantial uncertainty, and the associated empirical frequencies may provide unreliable estimates of the corresponding probabilities.  Recent frequentist contributions have already considered simultaneous inference for observed and unseen probabilities in large-alphabet settings \citep[e.g.,][]{Painsky2023}. Therefore, a natural extension of our approach would be to develop Bayesian nonparametric upper bounds that jointly cover unseen and low-frequency labels.

\section{\bf Acknowledgments}
Daniele Durante and Alessandro Colombi
are funded by the European Union (ERC, NEMESIS, project number: 101116718). Views and
opinions expressed are however those of the author(s)~only and do not necessarily reflect those of
the European Union or the European Research Council~Executive Agency. Neither the European
Union nor the granting authority can be held responsible~for them.
\newtheorem{lemma}[theorem]{Lemma}
\fontsize{9.9}{13.5}\selectfont

\vspace{40pt}
{\noindent \bf \Large Supplementary Material}
\vspace{-5pt}

\setcounter{section}{0}
\renewcommand{\thesection}{\Alph{section}}
\numberwithin{table}{section}
\numberwithin{figure}{section}
\numberwithin{equation}{section}
\numberwithin{theorem}{section}

\section{\large A. Proofs}
Below we prove the theoretical results stated in the main article. The proofs of Theorems~\ref{thm:species_post}~and~\ref{thm:features_post} can be found in  \cite{argiento2022annals}, \cite{pitman1996}, \cite{beraha2025bayesian} and \cite{Ghilotti2026}, and hence, are omitted here.

\vspace{5pt}
\subsection{\large A.1 Proof of Theorem \ref{thm:Species_UB_FDP_PD}}
From \eqref{eqn:species_UB_general}, it follows that it is enough to compute the posterior expected value of the $r$-th order missing mass, i.e., $\mean{M_r\mid X^n}$, under both the FDP and the PYP prior.

If $\Pcr = \operatorname{FDP}(q_M,\gamma) \times P_0$,  the posterior distribution in Theorem~\ref{thm:species_post}  can be equivalently~expressed as
\begin{equation}
    \label{eqn:FDP_post_proof}
    (\mu\mid X^n,M^\star) \,\myequal{d} \,
    \sum\nolimits_{j=1}^{K_n}\,p^*_j\delta_{\xi^\star_j} \,+\, 
    \sum\nolimits_{j=1}^{M^\star}\,p^\prime_j\delta_{\xi_j} \,,
\end{equation}
where $\xi_j \iid P_0$ and $(p^*_1,\ldots,p^*_{K_n},p^\prime_1,\ldots,p^\prime_{M^\star}\mid M^\star)\sim\operatorname{FD}_{K_n+M^\star}\left(\gamma+n_1,\ldots,\gamma+n_{K_n},\gamma,\ldots,\gamma\right)$. Consequently, $\mean{M_r\mid X^n}$ is equal to
\begin{equation*}
    \label{eqn:missing_FDP} 
    \E_{M^\star}[ \E(\sum\nolimits_{j=1}^{M^\star}\,(p^\prime_j)^r\mid M^\star)]
    \, = \,
    \E_{M^\star}
   ( M^\star \E[(p^\prime_1)^r\mid M^\star]),
\end{equation*}
where the first expected value is with respect to $M^\star \sim q_{M^\star}$; see Theorem~\ref{thm:species_post}. 
Moreover,~the~marginal posterior of any $p^\prime_j$ is
$(p^\prime_j \mid M^\star,X^n) \sim\operatorname{Beta}\left(\gamma,\, n + \gamma(K_n+M^\star-1)\right)$, for $j=1,\ldots,M^\star$.~In particular, the marginal distribution does not depend on $j$. Therefore, the statement follows~by~recalling that the $r$-th moment of the generic random variable $Z\sim\operatorname{Beta}(a,b)$~is 
\begin{equation}
    \label{eqn:rth_moment_beta}
        \E\left(Z^r\right) \, = (a)_r/(a+b)_r\,.
\end{equation}
To prove the result for $U^{\operatorname{FD}}$ it suffices to replace $M^*$ with $M-K_n$ in the above derivations.~Recall that, under FD, $q_{M*}=\delta_{M-K_n}$.

Moving now to the PYP prior, notice that, if $\Pcr = \pityor(\sigma,\theta) \times P_0$, then from~the~posterior~distribution in Theorem~\ref{thm:species_post} with $\mu^\prime = \sum_{j\ge 1} \tilde{p}^\prime_j\delta_{\xi_j}$, the $r$-th order missing mass is
\begin{equation*}
    \label{eqn:proof_gmassmiss_PD_1}
    \begin{split}
        &\E(\sum\nolimits_{j\geq 1} p_j^r\, \indic(n_j = 0) 
        \mid X^n )  
        \,=\,
        \E(\sum\nolimits_{j\geq 1} \left(p^* \tilde{p}^\prime_j\right)^r \mid X^n )
        \,=\,
        \E_{p^*}[\E(\sum\nolimits_{j\geq 1} (p^* \tilde{p}^\prime_j)^r \mid p^*, X^n )] \, .
    \end{split}
\end{equation*}
We recall that the marginal posterior of $p^*$ is $(p^* \mid X^n) \sim \operatorname{Beta}(\theta + K_n\sigma, \, n - \sigma K_n)$.
Moreover,~let $\breve{\mu}^{\prime} = \sum_{j\geq 1} \breve{p}^{\prime}_j \delta_{\xi_j}$ be the size-biased representation of $\mu^\prime$ \citep[see, e.g.,][]{pitman92}. Using~the size-biased sampling identity \citep[Equation (11)]{Pitman2003PoissonKingman}, we have
\begin{equation*}
    \label{eqn:missing_PD_proof2}
    \begin{split}
        &\E(\sum\nolimits_{j\geq 1} p_j^r\, \indic(n_j = 0) 
        \mid X^n )  
        \, = \,
        \E_{p^*}\left[ \int_0^1 \frac{\left(p^* \breve{p}^{\prime}_1\right)^r}{\breve{p}^{\prime}_1}  \breve{\nu}(\D \breve{p}^{\prime}_1) \right] \, 
    \end{split}
\end{equation*}
where $\breve{\nu}(\D \breve{p}^{\prime}_1)$,  is the distribution of the first weight in the size-biased representation, namely,~$\breve{p}^{\prime}_1 \sim \breve{\nu}(\D \breve{p}^{\prime}_1)$. 
In the  PYP case \citep[see, e.g.,][]{pitman92},  we have
\begin{equation*}
\label{eqn:missing_PD_proof3}
    \begin{split}
        \breve{p}^{ \prime}_1 \sim 
        \operatorname{Beta}(1-\sigma,\, \theta + (K_n+1)\sigma)\,.
    \end{split}
\end{equation*}
Hence, \eqref{eqn:PD_UB_species} follows directly by combining the above results with the independence 
between $p^*$ and $\breve p^{\prime}_1$, using \eqref{eqn:rth_moment_beta} to compute the $(r-1)$-th moment of $\breve p^{\prime}_1$, and the $r$-th moment of $p^*$. 

\vspace{9pt}
\subsection{\large A.2 Proof of Theorem \ref{thm:Features_upper_bounds}}\label{proof:UB_features}
From \eqref{eqn:species_UB_general}, it follows that it is enough to compute the posterior expected value of the $r$-th order missing mass, i.e., $\mean{M_r\mid X^n}$, under both the MBP and the IBP prior.

If $\Fcr=\operatorname{MBP}(q_M,a,b) \times P_0$,  then the posterior distribution in Theorem~\ref{thm:features_post} can be equivalently represented as
\begin{equation*}
    \label{eqn:MBP_post_proof}
    (\mu\mid X^n,M^\star) \,\myequal{d} \,
    \sum\nolimits_{j=1}^{K_n}\,p^*_j\delta_{\xi^\star_j} \,+\, 
    \sum\nolimits_{j=1}^{M^\star}\,p^\prime_j\delta_{\xi_j} \,,
\end{equation*}
where $\xi_j \iid P_0$, while $p^*_j$ and $q_{M^\star}$ are distributed as in \eqref{eqn:MBP_post}, and 
$(p^\prime_j\mid M^\star)\iid\operatorname{Beta}\left(a,b+n\right)$, for each $j=1,\ldots,M^\star$. 
As a result, $\mean{M_r\mid X^n}$ coincides with
\begin{equation*}
    \label{eqn:missing_MBP} 
    \E_{M^\star}[ \E(\sum\nolimits_{j=1}^{M^\star}\,(p^\prime_j)^r\mid M^\star)]
    \, = \,
    \E_{M^\star}
   ( M^\star \E[(p^\prime_1)^r\mid M^\star]),
\end{equation*}
where the first expected value is with respect to $M^\star \sim q_{M^\star}$. The statement follows by computing the $r$-th moment of the above Beta-distributed variable, leveraging \eqref{eqn:rth_moment_beta}. Similarly to the proof of Theorem \ref{thm:Species_UB_FDP_PD}, also in this case, the result for $U^{\operatorname{FB}}$ can be readily obtained by replacing $M^\star$ with $M-K_n$ in the above derivations.~Recall that, under FB, $q_{M^\star}=\delta_{M-K_n}$.

To prove the result for $\Fcr = \operatorname{IBP}(\gamma,\sigma,c) \times P_0$, notice that from the posterior~distribution~in~Theorem~\ref{thm:features_post}, with $\mu^\prime = \sum_{j\ge 1} p^\prime_j\delta_{\xi_j}$, the $r$-th order missing mass is
\begin{equation*}
    \label{eqn:proof_gmassmiss_IBP_1}
    \begin{split}
        &\E(\sum\nolimits_{j\geq 1} p_j^r\, \indic(n_j = 0) 
        \mid X^n )  
        \,=\,
        \E(\sum\nolimits_{j\geq 1} ( p^\prime_j)^r \mid X^n )
        \,=\,
        \E\left[ \int_0^1 s^r \Phi^\prime(\D s) \mid X^n \right]\,,
    \end{split}
\end{equation*}
where $\Phi^\prime$ is a Poisson process on $[0,1]$, whose Lévy intensity is $\rho_n(s)=(1-s)^n\rho(s)=\gamma s^{-\sigma-1}(1-s)^{c+n+\sigma-1}$. By Campbell's theorem \citep{kallenberg2021}, we have that 
\begin{equation*}
    \label{eqn:proof_gmassmiss_IBP_2}
    \begin{split}
        &\E(\sum\nolimits_{j\geq 1} p_j^r\, \indic(n_j = 0) 
        \mid X^n )  
        \,=\,
        \int_0^1 \, s^r \rho_n(s)\D s 
        \, = \, 
        \gamma B(r-\sigma,c+n+\sigma),
    \end{split}
\end{equation*}
which concludes the proof. Notice that the above Beta function is well-defined since $r\ge1$ ensures $r>\sigma$ and $c+n+\sigma > 0$.

\vspace{10pt}

\subsection{\large A.3 Proof of Theorem \ref{thm:freq-validity-species}}

We prove the two statements separately. First consider the finite-support case. Let $\mathcal J_0 := \{ j \ge 1 : p_{0,j} > 0 \}$.
By assumption, $\mathcal J_0$ is finite and has cardinality $m_0$. Since every $p_{0,j}$ is strictly positive on $\mathcal J_0$, the minimum
$p_{\min} := \min_{j \in \mathcal J_0} p_{0,j}$
is well defined and strictly positive.

If all the atoms of $\mu_0$ with positive mass have been observed in the sample, then every unseen atom has probability zero and therefore $M_{\max}(X_n;\mu_0) = 0$.
Consequently,
\[
\mathbb{P}_{\mu_0}( M_{\max}(X_n;\mu_0) > U^{\operatorname{FDP}}(X_n) )
\le
\mathbb{P}_{\mu_0}(K_n < M_0).
\]
The event $\{ K_n < M_0 \}$ means that at least one index $j \in \mathcal J_0$ satisfies $n_j = 0$. By the union bound,
\[
\mathbb{P}_{\mu_0}(K_n < M_0)
\le
\sum\nolimits_{j \in \mathcal J_0} \mathbb{P}_{\mu_0}(n_j = 0) \leq \sum\nolimits_{j \in \mathcal J_0} e^{-np_{0, j}} \le m_0 e^{-n p_{\min}}.
\]
Therefore
\[
\mathbb{P}_{\mu_0} ( M_{\max}(X_n;\mu_0) > U^{\operatorname{FDP}}(X_n) ) \to 0,
\]
which proves the claim for $U_n^{\operatorname{FDP}}$. The same arguments and derivations apply to $U_n^{\pityor}$.

Focusing now on the second statement of Theorem \ref{thm:freq-validity-species}, let us write for simplicity  $U^{\pityor}_n := U^{\pityor}(X_n)$.
Moreover,~for~every real $r \ge 1$, define
\[
u_n(r)
:=
\left[
\frac{1}{\alpha}
\frac{(\theta + \sigma K_n) \Gamma(r-\sigma) \Gamma(n+\theta)}
{\Gamma(1-\sigma) \Gamma(n+\theta+r)}
\right]^{1/r}, \qquad U^{\pityor}_n = \min_{r \ge 1} u_n(r).
\]
This coincides with \eqref{eqn:PD_UB_species} after writing the rising factorials through the Gamma function.
Instead of working directly with $u_n(r)$, it is convenient to define
\[
h_n(r) := \log(n u_n(r)) 
=
\log r - 1 + \frac{A_n - \left( \sigma + 0.5 \right) \log r + O_{\mathbb{P}_{\mu_0}}(1)}{r},
\]
where $A_n = \sigma_0 \log n + \log L(n) + O_{\mathbb{P}_{\mu_0}}(1)$; see the proof of Lemma \ref{lem:r-minimizer}.
Observe that the minimizers of $h_n(r)$ and $u_n(r)$ must coincide.
By Lemma \ref{lem:r-minimizer}, $u_n(r)$ is minimized at
\[
\hat r_n
=
\sigma_0 \log n
-
\left( \sigma + 0.5 \right) \log \log n
+
\log L(n)
+
O_{\mathbb{P}_{\mu_0}}(1).
\]
We next convert this expansion into an asymptotic one for $U^{\pityor}_n$. Evaluating $h_n(r)$ at $\hat r_n$ yields
\[
h_n(\hat r_n)
=
\log \hat r_n - 1 + \frac{A_n - \left( \sigma + 0.5 \right) \log \hat r_n + O_{\mathbb{P}_{\mu_0}}(1)}{\hat r_n},
\]
and substituting $A_n$ we have
\[
h_n(\hat r_n)
=
\log \hat r_n + O_{\mathbb{P}_{\mu_0}} \left( 1/\log n \right).
\]
Since $
n U^{\pityor}_n = n u_n(\hat r_n) = \exp(h_n(\hat r_n))$
we conclude that
\begin{multline}\label{eq:Un-asymp}
   n U^{\pityor}= \hat r_n \exp \left( O_{\mathbb{P}_{\mu_0}} \left( 1/\log n \right) \right) = \hat r_n + O_{\mathbb{P}_{\mu_0}}(1) = \\  \sigma_0 \log n - \left( \sigma + 0.5 \right) \log \log n + \log L(n) + O_{\mathbb{P}_{\mu_0}}(1). 
\end{multline}

We now compare the asymptotic expansion in \eqref{eq:Un-asymp} with an upper bound for the true unseen modal probability.
To this end, observe that
\begin{align*}
   \mathbb{P}_{\mu_0}(M_{\max}(X_n;\mu_0) > x)
\le
\sum\nolimits_{j : p_{0,j} > x} \mathbb{P}_{\mu_0}(n_j = 0) \le \sum\nolimits_{j : p_{0,j} > x} e^{-n p_{0,j}} =: \Psi_n(x),
\end{align*}
and by Lemma \ref{lem:psi-expansion} below, if $x_n \downarrow 0$ with $z_n=n x_n \to \infty$,
\[
    \Psi_n(x_n) = O \left( \frac{e^{-z_n} \bar\nu_0(x_n)}{z_n} \right).
\]
We now choose the specific sequence $x_n$ that will be used in the coverage argument.
Let $c$ be any constant satisfying $\sigma + 0.5 < c < \sigma_0 + 1$,  which is possible because by assumption $\sigma < \sigma_0 + 0.5$.
Define
\begin{equation}
\label{eq:xn-threshold-revised}
x_n
:=
\frac{\sigma_0 \log n - c \log \log n + \log L(n)}{n} \asymp (\log n)/n.
\end{equation}
Since $L$ is slowly varying, $\log L(n) = o(\log n)$, so $x_n \downarrow 0$ and $n x_n \to \infty$.
By assumption, $\bar\nu_0(x_n) \sim c_0 x_n^{-\sigma_0} L(1/x_n)$, so that,
\[
\Psi_n(x_n)
=
O \left(
\frac{e^{-n x_n}}{n x_n}
x_n^{-\sigma_0} L(1/x_n)
\right).
\]
Moreover, from \eqref{eq:xn-threshold-revised}, we compute
\[
e^{-n x_n}
=
\exp \left( -\sigma_0 \log n + c \log \log n - \log L(n) \right)
=
n^{-\sigma_0} (\log n)^c L(n)^{-1}.
\]
Also, from $x_n \asymp (\log n)/n$, it follows that
\[
\frac{x_n^{-\sigma_0}}{n x_n}
=
\frac{1}{n} x_n^{-(\sigma_0+1)}
=
O( n^{\sigma_0} (\log n)^{-(\sigma_0+1)} ).
\]
Combining the last two displays, we get
\begin{equation}
\label{eq:Psi-before-Potter-revised}
\Psi_n(x_n)
=
O (
(\log n)^{c-(\sigma_0+1)}
L(1/x_n)/L(n)
).
\end{equation}
Since $x_n \asymp (\log n)/n$, we have $1/x_n \asymp n/\log n$. 

By Potter's bounds for slowly varying functions \citep[see, e.g., Theorem 1.5.6 in][]{bingham1989regular}, for every $\varepsilon > 0$, we have
\[
L(1/x_n)/L(n) = O \left( (\log n)^\varepsilon \right).
\]
Choose $\varepsilon > 0$ so small that $c + \varepsilon < \sigma_0 + 1$,
then \eqref{eq:Psi-before-Potter-revised} implies $
\Psi_n(x_n) \to 0$.
Combining the above results,~we obtain
\begin{equation}
\label{eq:Mmax-threshold-goes-to-zero-revised}
\mathbb{P}_{\mu_0}(M_{\max}(X_n;\mu_0) > x_n) \to 0.
\end{equation}
The proof concludes by comparing $x_n$ and the credible upper bound $U^{\pityor}_n$. Observe that
$
n U^{\pityor}_n - n x_n
=
\left( c - \sigma - 0.5 \right) \log \log n
+
O_{\mathbb{P}_{\mu_0}}(1).
$
Since $c > \sigma + 0.5$, the deterministic coefficient of $\log \log n$ is positive, and therefore
\begin{equation}
\label{eq:Un-above-threshold-revised}
\mathbb{P}_{\mu_0}(U^{\pityor}_n \ge x_n) \to 1.
\end{equation}
Finally,
\[
\mathbb{P}_{\mu_0}(M_{\max}(X_n;\mu_0) > U^{\pityor}_n)
\le
\mathbb{P}_{\mu_0}(M_{\max}(X_n;\mu_0) > x_n)
+
\mathbb{P}_{\mu_0}(U^{\pityor}_n < x_n)
\to 0.
\]
which concludes the proof.

\vspace{5pt}

\begin{lemma}\label{lem:r-minimizer}
Under the same assumptions of Theorem \ref{thm:freq-validity-species} (item 2.), $U_n^{\pityor}$ is minimized for
\begin{equation}
\label{eq:rhat-second-order-revised}
\hat r_n
=
\sigma_0 \log n
-
( \sigma + 0.5) \log \log n
+
\log L(n)
+
O_{\mathbb{P}_{\mu_0}}(1).
\end{equation}
\end{lemma}
\begin{proof}
Define $h_n(r) := \log(n u_n(r))$.
The function $h_n(r)$ is continuous on $[1,\infty)$. We first show that $h_n(r)$ admits a global minimizer. To this end, set $
C_n := (\theta + \sigma K_n)/(\alpha \Gamma(1-\sigma))$,~so~that~$u_n(r)^r = C_n (\Gamma(r-\sigma) \Gamma(n+\theta))/\Gamma(n+\theta+r)$.
As $r \to \infty$, the Gamma ratio satisfies
\[
(\Gamma(r-\sigma) \Gamma(n+\theta))/\Gamma(n+\theta+r)
=
\Gamma(n+\theta) r^{-(n+\theta+\sigma)} (1+o(1)),
\]
and hence
\[
u_n(r)^r = C_n \Gamma(n+\theta) r^{-(n+\theta+\sigma)} (1+o(1)).
\]
Taking $r$-th roots, $u_n(r) \to 1$ as  $r \to \infty$ so that $h_n(r) \to \log n$ as $r \to \infty$.
Since $h_n(r)$~is~continuous on $[1,\infty)$ and has finite limit at infinity, it attains its minimum~on~$[1,\infty)$.~Let~$\hat r_n$~be~any~minimizer, so that
$U^{\pityor}_n = u_n(\hat r_n)$.

We next derive the asymptotic expansion of $\hat r_n$. By the regular variation assumption~(see~e.g., \cite{gnedin07}), it holds $K_n = c_K n^{\sigma_0} L(n) \left( 1 + o_{\mathbb{P}_{\mu_0}}(1) \right)$.
Setting $A_n := \log C_n$ we obtain~the asymptotic expansion
\begin{equation}
\label{eq:An-expansion-revised}
A_n = \sigma_0 \log n + \log L(n) + O_{\mathbb{P}_{\mu_0}}(1).
\end{equation}
By definition of $h_n(r)$, it follows that
\[
h_n(r)
=
\log n
+
\frac{1}{r}
\left[
A_n + \log \Gamma(r-\sigma) + \log \Gamma(n+\theta) - \log \Gamma(n+\theta+r)
\right].
\]
Differentiating with respect to $r$, and writing $\psi$ for the digamma function, gives
\[
h_n'(r)
=
-
\frac{1}{r^2}
\left[
A_n + \log \Gamma(r-\sigma) + \log \Gamma(n+\theta) - \log \Gamma(n+\theta+r)
\right]
+
\frac{1}{r}
\left[
\psi(r-\sigma) - \psi(n+\theta+r)
\right].
\]
Fix now a constant $C > 1$ and consider the interval $I_n(C) := \left[ C^{-1}\log n, C \log n \right]$. Then, uniformly for $r \in I_n(C)$, Stirling's formula gives $\log \Gamma(r-\sigma)
=
\left( r-\sigma-0.5 \right) \log r - r + O(1)$ and  $\psi(r-\sigma) = \log r + O \left( 1/r\right)$.
Also, since $r = O(\log n)$, we have
\[
\log \Gamma(n+\theta+r) - \log \Gamma(n+\theta)
=
\int_0^r \psi(n+\theta+s) ds
=
r \log n + O ( r^2/n ),
\]
and $\psi(n+\theta+r) = \log n + O \left( r/n \right)$.
Substituting these expansions into the formulas for $h_n(r)$~and $h_n'(r)$ yields, uniformly for $r \in I_n(C)$,
\begin{equation}
\label{eq:hn-expansion-revised}
h_n(r)
=
\log r - 1 + \frac{A_n - \left( \sigma + 0.5 \right) \log r + O_{\mathbb{P}_{\mu_0}}(1)}{r},
\end{equation}
and
\begin{equation}
\label{eq:hnprime-expansion-revised}
h_n'(r)
=
\frac{r - A_n + \left( \sigma + 0.5 \right) \log r + O_{\mathbb{P}_{\mu_0}}(1)}{r^2}.
\end{equation}

We now show that the minimizer $\hat r_n$ lies in $I_n(C)$ for suitable $C$ and is uniquely determined~by the equation $h_n'(\hat r_n)=0$.
Take $C$ so large such that $1/C < \sigma_0 < C$. Then the point $r_n^\star := \sigma_0 \log n$ belongs to $I_n(C)$ for all sufficiently large $n$. Evaluating \eqref{eq:hn-expansion-revised} at $r = r_n^\star$ and using \eqref{eq:An-expansion-revised}, we have $h_n(r_n^\star) = \log \log n + O_{\mathbb{P}_{\mu_0}}(1)$.
On the other hand,
\[
h_n(1) = \log \left( \frac{n(\theta+\sigma K_n)}{\alpha(n+\theta)} \right)
= \sigma_0 \log n + O_{\mathbb{P}_{\mu_0}}(1),
\]
which is much larger than $\log \log n$. Therefore, $\hat r_n > 1$ with probability tending to one. Since $\hat r_n$ is then an interior minimizer of the differentiable function $h_n$, it must satisfy $h_n'(\hat r_n) = 0$.

We next prove that $h_n'(r)$ has a unique zero and that this zero is of order $\log n$. Define~$G_n(r) := r^2 h_n'(r)$.
A direct differentiation gives
\[
G_n'(r)
=
r \left[ \psi_1(r-\sigma) - \psi_1(n+\theta+r) \right],
\]
where $\psi_1$ is the trigamma function. Since $\psi_1$ is strictly positive and strictly decreasing on $(0,\infty)$, and since $r-\sigma < n+\theta+r$,
we have $G_n'(r) > 0$ for every $r \ge 1$.
Thus, $G_n(r)$~is~strictly~increasing~on $[1,\infty)$, and therefore $h_n'(r)$ can vanish at most once.
Fix $\varepsilon \in (0,\sigma_0)$, and define~$r_n^- := (\sigma_0-\varepsilon)\log n$ and $r_n^+ := (\sigma_0+\varepsilon)\log n$.
For $n$ large, both points belong to $I_n(C)$. Using \eqref{eq:hnprime-expansion-revised} and \eqref{eq:An-expansion-revised},~we~get
\[
r_n^- - A_n + \left( \sigma + 0.5 \right) \log r_n^-
=
-\varepsilon \log n + O(\log \log n) + O_{\mathbb{P}_{\mu_0}}(1),
\]
hence $h_n'(r_n^-) < 0$ with probability tending to one. Similarly,
\[
r_n^+ - A_n + \left( \sigma + 0.5 \right) \log r_n^+
=
\varepsilon \log n + O(\log \log n) + O_{\mathbb{P}_{\mu_0}}(1),
\]
so $h_n'(r_n^+) > 0$ with probability tending to one.

Since $h_n'(r)$ has at most one zero and changes sign between $r_n^-$ and $r_n^+$, it has exactly one zero inside $(1,\infty)$, and this zero is the unique global minimizer of $h_n(r)$. Therefore $\hat r_n \in [r_n^-, r_n^+]$
with probability tending to one, and, in particular
\begin{equation}
\label{eq:rhat-first-order-revised}
\hat r_n = \sigma_0 \log n + O_{\mathbb{P}_{\mu_0}}(\log \log n).
\end{equation}
We now use the identity $h_n'(\hat r_n)=0$ together with equation \eqref{eq:hnprime-expansion-revised}. Since $\hat r_n \in I_n(C)$~with~probability tending to one, it follows that
\[
0
=
\frac{\hat r_n - A_n + \left( \sigma + 0.5 \right) \log \hat r_n + O_{\mathbb{P}_{\mu_0}}(1)}{\hat r_n^2},
\]
and multiplying by $\hat r_n^2$, yields
\[
\hat r_n
=
A_n - \left( \sigma + 0.5 \right) \log \hat r_n + O_{\mathbb{P}_{\mu_0}}(1).
\]
Using \eqref{eq:rhat-first-order-revised}, we have
\[
\log \hat r_n = \log \log n + O_{\mathbb{P}_{\mu_0}}(1),
\]
and substituting this together with \eqref{eq:An-expansion-revised} into the previous display concludes the proof.
\end{proof}

\vspace{5pt}
\begin{lemma}\label{lem:psi-expansion}
Let $ \Psi_n(x) = \sum_{j : p_{0,j} > x} e^{-n p_{0,j}}$ and consider a deterministic sequence $x_n \downarrow 0$ such that $z_n=n x_n \to \infty$. Then, under the same assumptions of Theorem \ref{thm:freq-validity-species} (item 2.),
\begin{equation}
\label{eq:Psi-general-bound-revised}
\Psi_n(x_n)
=
O \left( \frac{e^{-z_n} \bar\nu_0(x_n)}{z_n} \right).
\end{equation}
\end{lemma}
\begin{proof}
Notice that, for each $j$ such that $p_{0,j} > x$, it holds $e^{-n p_{0,j}} = e^{-n x} - \int_x^{p_{0,j}} n e^{-n u} du.$~Summing over all such $j$ and using Tonelli's theorem, we obtain
\[
\Psi_n(x)
=
e^{-n x} \bar\nu_0(x) - n \int_x^1 e^{-n u} \bar\nu_0(u) du.
\]
Let us now change variable $u = x + y/n$. Since $\bar\nu_0(t) = 0$ for $t > 1$, the upper limit~can~be~extended to $+\infty$, obtaining
\begin{equation}
\label{eq:Psi-exact-revised}
\Psi_n(x)
=
e^{-n x} \int_0^\infty e^{-y} \left[ \bar\nu_0(x) - \bar\nu_0(x+y/n) \right] dy.
\end{equation}
Fix now a deterministic sequence $x_n \downarrow 0$ such that $z_n := n x_n \to \infty$.
Moreover, split the integral in \eqref{eq:Psi-exact-revised} at $y = z_n$, so that $\Psi_n(x_n) = I_{n,1} + I_{n,2}$,
where
\[
I_{n,1}
:=
e^{-z_n} \int_0^{z_n} e^{-y} \left[ \bar\nu_0(x_n) - \bar\nu_0(x_n+y/n) \right] dy
\]
and
\[
I_{n,2}
:=
e^{-z_n} \int_{z_n}^{\infty} e^{-y} \left[ \bar\nu_0(x_n) - \bar\nu_0(x_n+y/n) \right] dy.
\]

Let us now bound $I_{n,1}$. Notice that, for $0 \le y \le z_n$, we have
$
x_n + y/n
=
x_n ( 1 + y/z_n ),
$
and the multiplier $y/z_n$ belongs to $[0,1]$. For $n$ large enough we also have $x_n < x_0$. Therefore~under~the assumptions in item (2) of Theorem \ref{thm:freq-validity-species} it follows that
\[
0 \le \bar\nu_0(x_n) - \bar\nu_0 \left( x_n + y/n \right)
=
\bar\nu_0(x_n) - \bar\nu_0 \left( x_n \left( 1 + y/z_n \right) \right)
\le
C_I (y/z_n) \bar\nu_0(x_n).
\]
Substituting this bound into the definition of $I_{n,1}$ yields
$
I_{n,1}
\le
C_I e^{-z_n} (\bar\nu_0(x_n)/z_n) \int_0^{z_n} y e^{-y} dy.
$
Since $\int_0^{z_n} y e^{-y} dy \le \int_0^\infty y e^{-y} dy = 1,$
we obtain
\begin{equation}
\label{eq:I1-bound-revised}
I_{n,1}
=
O \left( \frac{e^{-z_n} \bar\nu_0(x_n)}{z_n} \right).
\end{equation}

We next bound $I_{n,2}$. By monotonicity of $\bar\nu_0$, it holds $0 \le \bar\nu_0(x_n) - \bar\nu_0(x_n+y/n) \le \bar\nu_0(x_n),$~and
hence
$
I_{n,2}
\le
e^{-z_n} \bar\nu_0(x_n) \int_{z_n}^\infty e^{-y} dy
=
e^{-2 z_n} \bar\nu_0(x_n).
$
Since $z_n \to \infty$, we have $e^{-z_n} = o(z_n^{-1})$, and therefore
\begin{equation}
\label{eq:I2-bound-revised}
I_{n,2}
=
o \left( \frac{e^{-z_n} \bar\nu_0(x_n)}{z_n} \right).
\end{equation}
Combining \eqref{eq:I1-bound-revised} and \eqref{eq:I2-bound-revised}, proves \eqref{eq:Psi-general-bound-revised}.
\end{proof}

\vspace{10pt}
\subsection{\large A.4 Proof of Theorem \ref{thm:freq-validity-features}}

Before proceeding with the proof, we observe that, 
for fixed $(\sigma,c)\in\Theta$, maximizing the IBP marginal likelihood analytically with respect to $\gamma>0$ gives $\hat\gamma_n(\sigma,c)=K_n/\varphi_n(\sigma,c)$~on~the~event~$K_n>0$, where
\begin{equation}\label{eq:maxprofile-gamma}
    \varphi_n(\sigma,c):=\frac{\Gamma(c+1)}{\Gamma(c+\sigma)}\sum_{i=1}^n \frac{\Gamma(c+\sigma+i-1)}{\Gamma(c+i)}.
\end{equation}
Then, the empirical Bayes estimator is obtained by maximizing the profiled likelihood over $(\sigma,c)\in\Theta$ and setting $\hat\gamma_n=K_n/\varphi_n(\hat\sigma_n,\hat c_n)$.

In the proof of Theorem \ref{thm:freq-validity-features} we also use the following lemma.
\begin{lemma}
\label{lem:ibp-asymptotics-theorem6}
Assume that $\mu_0$ satisfies \eqref{eq:reg_vary} and \eqref{eq:smoothness-rv}. Let $\Theta \subset \{(\sigma,c):0<\sigma<1,\ c>-\sigma\}$ be compact and such that, for every $(\sigma,c)\in\Theta$, it holds $0<\underline\sigma\le \sigma\le \overline\sigma<1$ and $c+\sigma\ge \eta$.

\begin{enumerate}
\item If $x_n\downarrow0$ and $n x_n\to\infty$, then
\[
\Psi_n(x_n):=\sum\nolimits_{j:p_{0,j}>x_n} e^{-n p_{0,j}}=O(e^{-nx_n}\bar\nu_0(x_n)/(nx_n)).
\]

\item Let $(\sigma,\gamma,c)$ be fixed, with $0<\sigma<1$, $\gamma>0$, and $c>-\sigma$. Then
\[
nU_n^{\operatorname{IBP}}=\sigma\log n-\left(\sigma+0.5\right)\log\log n+O(1).
\]
\item Let $(\hat\sigma_n,\hat c_n)$ be selected by maximum marginal likelihood over $\Theta$ after maximizing analytically over $\gamma$. Set $\hat\gamma_n=K_n/\varphi_n(\hat\sigma_n,\hat c_n)$, and let $\widehat U_n^{\operatorname{IBP}}:=\min\nolimits_{r\ge1}\left[(\hat\gamma_n/\alpha)B(r-\hat\sigma_n,\hat c_n+\hat\sigma_n+n)\right]^{1/r}.$
Then
\[
n\widehat U_n^{\operatorname{IBP}}=\sigma_0\log n+\log L(n)-\left(\hat\sigma_n+0.5\right)\log\log n+O_{\mathbb P_{\mu_0}}(1).
\]
\item If $(\hat\sigma_n,\hat c_n)\in\Theta$, then
\[
\log\varphi_n(\hat\sigma_n,\hat c_n)=\hat\sigma_n\log n+O_{\mathbb P_{\mu_0}}(1),
\]
where $\varphi_n$ is as in \eqref{eq:maxprofile-gamma}.
\end{enumerate}
\end{lemma}
\begin{proof}
The proofs of the first, second, and third statements follow closely the arguments~and~the~derivations considered in the proof of Lemmas \ref{lem:r-minimizer} and \ref{lem:psi-expansion}, with \eqref{eqn:UB_IBP} replacing \eqref{eqn:PD_UB_species}. We thus prove only the fourth statement.

By the uniform Gamma-ratio expansion on compact subsets of the admissible parameter space, we have
$
\Gamma(c+\sigma+i-1)/\Gamma(c+i)=i^{\sigma-1}\left(1+O\left(1/i\right)\right),
$
uniformly over $(\sigma,c)\in\Theta$. Therefore,
\[
\sum_{i=1}^n \frac{\Gamma(c+\sigma+i-1)}{\Gamma(c+i)}=\sum_{i=1}^n i^{\sigma-1}+O\left(\sum_{i=1}^n i^{\sigma-2}\right)
\]
uniformly over $\Theta$. Since $\sigma\le \overline\sigma<1$, the second sum is bounded uniformly in $n$ and $(\sigma,c)\in\Theta$. Moreover, since $\sigma\in[\underline\sigma,\overline\sigma]$ with $\underline\sigma>0$, it follows that
$
\sum_{i=1}^n i^{\sigma-1}=(n^\sigma/\sigma)(1+o(1))
$
uniformly~over $\Theta$. Hence
\[
\sum_{i=1}^n \frac{\Gamma(c+\sigma+i-1)}{\Gamma(c+i)}=\frac{n^\sigma}{\sigma}(1+o(1))
\]
uniformly over $\Theta$, and therefore
\[
\varphi_n(\sigma,c)=\frac{\Gamma(c+1)}{\sigma\Gamma(c+\sigma)}n^\sigma(1+o(1))
\]
uniformly over $\Theta$. Since $\Theta$ is compact and bounded away from $\sigma=0$ and $c+\sigma=0$, the function $(\sigma,c)\mapsto \Gamma(c+1)/(\sigma\Gamma(c+\sigma))$ is bounded above and below away from zero on $\Theta$. Taking logarithms gives $\log\varphi_n(\sigma,c)=\sigma\log n+O(1)$
uniformly over $\Theta$,  which entails the fourth statement.
\end{proof}

\vspace{5pt}
Leveraging the above results, we can now proceed with the proof of Theorem \ref{thm:freq-validity-features}. To this end let us first notice that the proof for $U_n^{\operatorname{MBP}}$ and $U_n^{\operatorname{IBP}}$ in the finite support case~for~$\mu_0$,~follows~directly from the same arguments and derivations as those considered for the proof of the first statement of Theorem \ref{thm:freq-validity-species}. Therefore, it suffices to cover the regularly-varying case. Starting from the fixed-hyperparameter claim under the IBP, recall that, by Lemma \ref{lem:ibp-asymptotics-theorem6},
\[
nU^{\operatorname{IBP}}_n=\sigma\log n-\left(\sigma+0.5\right)\log\log n+O(1).
\]
Since the theorem assumes $\sigma>\sigma_0$, choose any constant $a<\sigma_0+1$ and define $x_n:=(\sigma_0\log n+\log L(n)-a\log\log n)/n$. Since $L$ is slowly varying, $\log L(n)=o(\log n)$, so $x_n\downarrow0$ and $n x_n\to\infty$. Arguing as in the proof of Theorem \ref{thm:freq-validity-species}, it follows $\mathbb P_{\mu_0}(M_{\max}(X_n;\mu_0)>x_n)\to0$. Hence, observing that
\[
nU^{\operatorname{IBP}}_n-nx_n=(\sigma-\sigma_0)\log n-\log L(n)+\left(a-\sigma-0.5\right)\log\log n+O(1),
\]
as $\sigma>\sigma_0$ and $\log L(n)=o(\log n)$, we can conclude 
\[
\mathbb P_{\mu_0}(M_{\max}(X_n;\mu_0)>U^{\operatorname{IBP}}_n)\le \mathbb P_{\mu_0}(M_{\max}(X_n;\mu_0)>x_n)+\mathbb P_{\mu_0}(U^{\operatorname{IBP}}_n<x_n)\to0.
\]
Therefore $\mathbb P_{\mu_0}(M_{\max}(X_n;\mu_0)\le U^{\operatorname{IBP}}_n)\to1$, proving the fixed-hyperparameter claim.

We now turn to the empirical Bayes claim in Theorem \ref{thm:freq-validity-features}. By the EFPF in Section~\ref{app:hy_elicitation_features},~for~fixed $(\sigma,c)$ the dependence of the marginal likelihood on $\gamma$ is through the factor $\gamma^{K_n}e^{-\gamma\varphi_n(\sigma,c)}$.~Therefore, maximizing analytically over $\gamma$, yields
$
\hat\gamma_n=K_n/\varphi_n(\hat\sigma_n,\hat c_n),
$
where the regular~variation~assumption implies $K_n=c_K n^{\sigma_0}L(n)(1+o_{\mathbb P_{\mu_0}}(1))$ for some constant $c_K>0$. Combining this result with the fourth statement of Lemma \ref{lem:ibp-asymptotics-theorem6} gives
\[
\log\hat\gamma_n=(\sigma_0-\hat\sigma_n)\log n+\log L(n)+O_{\mathbb P_{\mu_0}}(1).
\]
This is the only difference with the fixed-parameter case, namely, the estimator of $\gamma$ absorbs the slowly varying factor $L(n)$. Recalling now  the third statement of Lemma \ref{lem:ibp-asymptotics-theorem6}, we have that
\[
n\widehat U^{\operatorname{IBP}}_n=\sigma_0\log n+\log L(n)-\left(\hat\sigma_n+0.5\right)\log\log n+O_{\mathbb P_{\mu_0}}(1).
\]
Moreover, by assumption, $\limsup_{n\to\infty}\hat\sigma_n<\sigma_0+1/2$ in $\mathbb P_{\mu_0}$-probability, which implies the existence of a deterministic constant $a$ such that $\mathbb P_{\mu_0}(\hat\sigma_n+0.5<a<\sigma_0+1)\to1$. Fix now such a constant $a$ and define $x_n:=(\sigma_0\log n+\log L(n)-a\log\log n)/n$. As in the fixed-hyperparameter case, $\mathbb P_{\mu_0}(M_{\max}(X_n;\mu_0)>x_n)\to0$. On the other hand,
\[
n\widehat U^{\operatorname{IBP}}_n-nx_n=\left(a-\hat\sigma_n-0.5\right)\log\log n+O_{\mathbb P_{\mu_0}}(1).
\]
By the choice of $a$, we have that $\mathbb P_{\mu_0}(\widehat U^{\operatorname{IBP}}_n\ge x_n)\to1$. Therefore
\[
\mathbb P_{\mu_0}(M_{\max}(X_n;\mu_0)>\widehat U^{\operatorname{IBP}}_n)\le \mathbb P_{\mu_0}(M_{\max}(X_n;\mu_0)>x_n)+\mathbb P_{\mu_0}(\widehat U^{\operatorname{IBP}}_n<x_n)\to0.
\]
Thus $\mathbb P_{\mu_0}(M_{\max}(X_n;\mu_0)\le \widehat U^{\operatorname{IBP}}_n)\to1$, \smash{proving the empirical Bayes claim and completing~the~proof.}

\vspace{20pt}
\subsection{\large A.5 Proof of Proposition \ref{prop:uniquness_minimizer}}

We first focus on the credible bounds induced by the FD, PYP, MBP and IBP priors (see Theorems~\ref{thm:Species_UB_FDP_PD}--\ref{thm:Features_upper_bounds}), whose underlying minimization problem can be treated under a common perspective.  In particular, as clarified in Table~\ref{tab:common_forms}, the functions of $r \geq 1$ minimized in these bounds admit the common form
\begin{equation}
\label{eqn:U_common_form}
    U_r(X_n)=\left[\widetilde C\,\frac{\Gamma(\widetilde a+r)}{\Gamma(\widetilde b+r)}\right]^{1/r},
\end{equation}
for suitable choices of \((\widetilde C,\widetilde a,\widetilde b)\) needed to recover each of the cases of interests. Consequently, it suffices to study $U_r(X_n)$  in \eqref{eqn:U_common_form}, with  $\widetilde C>0$, $\widetilde a>-1$, and $\widetilde b>\widetilde a\,$,  which hold for all priors under analysis (see Table \ref{tab:common_forms}). Consistent with the statement in Proposition \ref{prop:uniquness_minimizer}~and~its~subsequent discussion in the main article, we first study $U_r(X_n)$ with $r\ge1$ being real-valued,~and~then~transfer the results to the integer case. To this end, define
\[
G(r):=\log U_r(X_n)=F(r)/r,\, \qquad
F(r):=\log \widetilde C+\log\Gamma(\widetilde a+r)-\log\Gamma(\widetilde b+r)\,.
\]
Since \(\widetilde a>-1\) and \(r\ge 1\), the arguments \(\widetilde a+r\) and
\(\widetilde b+r\) are strictly positive, so \(F(r)\) and \(G(r)\)~are well defined and continuous. As first step, let us show that \(F(r)\) is strictly convex. Differentiating  \(F(r)\) twice, yields
\[
F''(r)=\psi_1(\widetilde a+r)-\psi_1(\widetilde b+r),
\]
where \(\psi_1\) denotes the trigamma function. Since \(\psi_1\) is strictly decreasing on \((0,\infty)\) and \(\widetilde b>\widetilde a\), it follows that $F''(r)>0$, for all $r\ge1$, and therefore  \(F(r)\) is strictly convex on \([1,\infty)\). 

Next, we study the behavior at infinity. By the Stirling's formula, we have, for $r\to\infty$, that
$
\log \Gamma(r+x)=\left(r+x-0.5\right)\log r-r+O(1),
$
for each fixed \(x\). Therefore
\[
F(r)=\log \widetilde C + (\widetilde a-\widetilde b)\log r + O(1)
=
-(\widetilde b-\widetilde a)\log r + O(1),
\]
and dividing by \(r\), yields
\[
G(r)= -(\widetilde b-\widetilde a)\frac{\log r}{r}+O(r^{-1})\rightarrow 0
\qquad\text{as }r\to\infty.
\]
Since \(\widetilde b-\widetilde a>0\), we in fact have \(G(r)<0\) for all sufficiently large \(r\).

Let us now prove existence of the minimizer.  To this end, choose \(r_0\) large enough that \(G(r_0)<0\).
Since \(G(r)\to 0\), there exists an \(R>r_0\) such that
\[
G(r)>G(r_0)/2>G(r_0)
\qquad\text{for all }r\ge R.
\]
Hence no global minimizer can lie in \([R,\infty)\), and any global minimizer must belong to the compact interval \([1,R]\). Thus, by continuity of \(G(r)\), a global minimizer exists.

\begin{table}[t]
    \centering
    \begin{tabular}{c c c c}
    \hline
    \rule{0pt}{4ex}
    Model & $\widetilde C$ & $\widetilde a$ & $\widetilde b$ \\
    \hline
        \rule{0pt}{8ex}
    FD
    &
    $\dfrac{M-K_n}{\alpha}\,
    \dfrac{\Gamma\!\bigl(n+\gamma(M-K_n)\bigr)}{\Gamma(\gamma)}$
    &
    $\gamma$
    &
    $n+\gamma(M-K_n)$
    \\
    [2em]
    PYP
    &
    $\dfrac{\theta+K_n\sigma}{\alpha}\,
    \dfrac{\Gamma(n+\theta)}{\Gamma(1-\sigma)}$
    &
    $-\sigma$
    &
    $n+\theta$
    \\
    [2em]
    MBP
    &
    $\dfrac{\mathbb E[M^\star]}{\alpha}\,
    \dfrac{\Gamma(a+b+n)}{\Gamma(a)}$
    &
    $a$
    &
    $a+b+n$
    \\
    [2em]
    IBP
    &
    $\dfrac{\gamma}{\alpha}\,\Gamma(c+n+\sigma)$
    &
    $-\sigma$
    &
    $c+n $\\
    [2em]
    \hline
    \end{tabular}
    \vspace{5pt}
    \caption{Parameter choices yielding the common representation in \eqref{eqn:U_common_form}.}
    \label{tab:common_forms}
    \vspace{5pt}
\end{table}

It remains to prove uniqueness. Assume by contradiction that \(x<y\) are two distinct global minimizers of \(G(r)\), and let
\[
m:=G(x)=G(y)=\min\nolimits_{r\ge 1}G(r).
\]
Then
\[
F(x)=mx,\qquad F(y)=my.
\]
Fix \(\lambda\in(0,1)\) and set \(z=\lambda x+(1-\lambda)y\). By the strict convexity of \(F(r)\), we have
\[
F(z)<\lambda F(x)+(1-\lambda)F(y)
= \lambda mx + (1-\lambda)my
= mz.
\]
which implies that
$
G(z)=F(z)/z<m,
$
thereby contradicting the definition of $m$~as~the~minimum~value of the function. Hence the global minimizer is unique.

\begin{remark}[On the conditions \(\widetilde a\) and \(\widetilde b\)]
The condition \(\widetilde a>-1\) ensures that \(\Gamma(\widetilde a+r)\) is well defined for every \(r\ge 1\), while \(\widetilde b>\widetilde a\) is exactly the sign condition making \(F(r)\) strictly convex~and~underlies the uniqueness of the minimizer. If \(\widetilde b=\widetilde a\), then
\[
U_r(X^n)=\widetilde C^{1/r},
\]
and the problem becomes degenerate. If \(\widetilde b<\widetilde a\), then \(F(r)\) is strictly concave and the previous proof no longer yields a unique minimum. Finally, if $\widetilde C = 0$, then $U_r(X^n)=0$ for every $r \ge 1$. Hence every $r \ge 1$ is a minimizer. The latter case may arise either when $M=K_n$ or $\mean{M^\star}=0$.
\end{remark}

Let us now conclude the proof by covering the FDP case, which cannot be framed and studied under the common form in \eqref{eqn:U_common_form}. Writing explicitly the expectation in~\eqref{eqn:UB_FDP}, we have 
\[
U_r^{\mathrm{FDP}}(X_n)
=
\left[
\frac{1}{\alpha}\,
(\gamma)_r
\sum_{m=1}^\infty \frac{m\,q_{M^\star}(m)}{(n+\gamma(K_n+m))_r}
\right]^{1/r}\,,
\]
for $r\in[1,\infty)$. Define, for brevity,
$
A_m:=n+\gamma(K_n+m)$, $
S(r):=\sum_{m=1}^\infty m\,q_{M^\star}(m)/(A_m)_r,$ and
$U_r^{\mathrm{FDP}}(X_n)=[(1/\alpha)(\gamma)_r\,S(r)]^{1/r},
$
and notice that $A_m>\gamma$ for each $m\ge 1$. Moreover, let
\[
H(r):=\bigl(U_r^{\mathrm{FDP}}(X_n)\bigr)^r=\frac{1}{\alpha}(\gamma)_r\,S(r)
=
\sum_{m=1}^\infty c_m e^{f_m(r)}\,, \quad 
\mbox{and} \quad
F(r):=\log H(r)= r \log U_r^{\mathrm{FDP}}(X_n)
\]
where
$
c_m:=m\,q_{M^\star}(m)\Gamma(A_m)/(\alpha\,\Gamma(\gamma)),
$
and
$
f_m(r):=\log\Gamma(\gamma+r)-\log\Gamma(A_m+r).
$
The latter, $f_m(r)$, is a strictly convex function on $[1,\infty)$ since
\[
f_m''(r)=\psi_1(\gamma+r)-\psi_1(A_m+r)>0\,.
\]
Leveraging the previous result, we next show the strict convexity on \([1,\infty)\) of $F(r)$. To this end, 
let \(z=\lambda x+(1-\lambda)y\) be a convex combination, for \(x\neq y\), \(\lambda\in(0,1)\), 
it follows that
\[
c_m e^{f_m(z)}
<
\bigl(c_m e^{f_m(x)}\bigr)^\lambda
\bigl(c_m e^{f_m(y)}\bigr)^{1-\lambda}.
\]
Summing over \(m\) and applying Hölder's inequality with exponents $1/\lambda$ and $1/(1-\lambda)$, we get
\[
H(z)
<
\sum\nolimits_{m=1}^\infty
\bigl(c_m e^{f_m(x)}\bigr)^\lambda
\bigl(c_m e^{f_m(y)}\bigr)^{1-\lambda}
\le
H(x)^\lambda H(y)^{1-\lambda}.
\]
Taking the logarithms yields $F(z)<\lambda F(x)+(1-\lambda)F(y)$, and hence, \(F(r)\) is strictly convex.

By defining
$$
G(r):=\log U_r^{\mathrm{FDP}}(X_n) = F(r)/r\,,
$$
the proof is completed using the same arguments as those employed for the previous bounds. In particular, we first prove the existence of a minimizer by studying the asymptotic behaviour of $G(r)$ and then leverage the strict convexity of $F(r)$ to prove its uniqueness. More precisly, we are only left to show that $G(r)$ is strictly negative for sufficiently large values of $r$. To this end, let $ m_0:=\min\{m\ge 1{:}\ q_{M^\star}(m)>0\}$ and $A_{m_0}=n+\gamma(K_n+m_0)$.~This~is~well~defined~since~\(q\)~is~a~probability mass function on the positive integers and its support is nonempty. Without any loss of generality, we assume that $q_{M^\star}(1)>0$ so that $m_0 = 1$. 
We first show that the summation $S(r)$~is well approximated by its first term, when $r\to\infty$. Namely,
\[
S(r)\sim \frac{q_{M^\star}(1)}{(A_1)_r},
\qquad\text{as }r\to\infty.
\]
Indeed,
\[
\frac{S(r)}{q_{M^\star}(1)/(A_1)_r}
=
1+
\sum_{m>1}
\frac{m\,q_{M^\star}(m)}{ q_{M^\star}(1)}
\rho_m(r)\,,
\]
where, for \(m>1\), we defined $\rho_m(r):=(A_1)_r/(A_m)_r$. The latter is a strictly decreasing function~on \([1,\infty)\) since
\[
\frac{d}{dr}\log \rho_m(r)=\psi(A_1+r)-\psi(A_m+r)<0,
\]
where the  inequality holds because \(A_m>A_1\). Therefore
\[
0\le \rho_m(r)\le \rho_m(1)=A_1/A_m.
\]
Also, by the standard Gamma-ratio asymptotics, we have that
\[
\rho_m(r)\sim C_m\,r^{A_1-A_m}\to 0
\qquad\text{as }r\to\infty,
\]
for some constant \(C_m>0\). Hence, for \(m>1\),
\[
0\le
\frac{m\,q_{M^\star}(m)}{q_{M^\star}(1)}\rho_m(r)
\le
\frac{m\,q_{M^\star}(m)}{q_{M^\star}(1)}\frac{A_1}{A_m}
\le
\frac{A_1}{\gamma\,q_{M^\star}(1)}\,q_{M^\star}(m),
\]
where the final inequality follows by noticing that $m/A_m < 1/\gamma$.
Then, the dominating sequence on the right hand side is summable because \(\sum_m q_{M^\star}(m)=1\). By dominated convergence,
\[
\sum_{m>1}
\frac{m\,q_{M^\star}(m)}{q_{M^\star}(1)}\rho_m(r)\to 0,
\]
which proves the claimed asymptotic equivalence. Consequently,
\[
H(r)
=
\frac{1}{\alpha}(\gamma)_r S(r)
\sim
\frac{q_{M^\star}(1)}{\alpha}\,
\frac{(\gamma)_r}{(A_1)_r}.
\]
Taking the logarithms, we have that
\[
F(r)=\log H(r)
=
\log\!\left(q_{M^\star}(1)/\alpha\right)
+\log(\gamma)_r
-\log(A_1)_r
+o(1),
\]
and using again the Gamma-ratio asymptotics,
\[
F(r)=-(A_1-\gamma)\log r+O(1)\,,
\]
with $A_1-\gamma = n+\gamma K_n>0$. Hence, we obtain
\[
\log U_r^{\mathrm{FDP}}(X^n)
=
\frac{F(r)}{r}
=
-(A_1-\gamma)\frac{\log r}{r}+O(r^{-1}).
\]
The proof is completed as we showed that $\log U_r^{\mathrm{FDP}}(X_n)$ approaches zero from below. Existence and uniqueness can now be proved as for the previous bounds.

\begin{remark}[Restriction to integer \(r\)]
The previous proof concerns the case in which $r$ is continuous and takes values in $[1,\infty)$. In contrast,  the credibile bounds within Theorems~\ref{thm:Species_UB_FDP_PD}--\ref{thm:Features_upper_bounds},~are~defined~over positive integer orders $r$. Under this restriction, the minimum is still attained in all cases since, as shown above, \(U_r(X^n)\to 1\) and \(U_r(X^n)<1\) for all sufficiently large \(r\). Hence only finitely many integer values of $r \geq 1$ need to be considered in the search for the minimum. Uniqueness, however, requires additional care, since two consecutive integers may attain the same minimum value even when the continuous minimizer is unique. Nonetheless, this aspect poses no difficulty~within~our setting, as the credible bounds we derive depend only on the minimum value of $U_r(X^n)$ and not~on the particular minimizing integer $r^*$.
\end{remark}

\vspace{10pt}
\section{\large B. Limiting Behavior of the Credible Upper Bounds }
\label{app:UB_limits}
In this Section, we fix $r>1$ and compute the limits of the Bayesian credible upper bounds, as reported in Section \ref{section:species_results}.

\vspace{5pt}
\subsection{\large B.1 Species setting}
\label{app:UB_limits_species}

Let us first focus on $U^{\operatorname{FD}}$ in Theorem~\ref{thm:Species_UB_FDP_PD} and compute the limits when $\gamma\to 0$ and $\gamma\to\infty$.~For~the $\gamma\to 0$ case, the result follows directly from the continuity of the rising factorial. Specifically
\begin{equation*}
    \begin{split}
    \lim_{\gamma \to 0}  \left[ \,
    \frac{1}{\alpha}\,
    \frac{(M-K_n) \, (\gamma)_r}{(n + \gamma M)_r}
    \right]^{1/r} \, = \, 
    \lim_{\gamma \to 0}  \left[ \,
    \frac{1}{\alpha}\,
    \frac{(M-K_n) \, (\gamma)_r}{(n)_r}
    \right]^{1/r} \, = \, 0.
    \end{split}
\end{equation*}
For $\gamma\to\infty$ we exploit, instead, the following approximation for any fixed $r$: $(x)_r=\prod_{j=0}^{r-1}(x+j)\sim x^r$. This yields $(\gamma)_r \sim \gamma^r$ and $(n+\gamma M)_r \sim (\gamma M)^r$, which implies 
\begin{equation*}
    \begin{split}
    \lim_{\gamma \to \infty}  \left[ \,
    \frac{1}{\alpha}\,
    \frac{(M-K_n) \, (\gamma)_r}{(n + \gamma M)_r}
    \right]^{1/r} \, = \, 
    \lim_{\gamma \to \infty}  
    \left[ \,
    \frac{1}{\alpha}\,
    \frac{(M-K_n) \, \gamma^r}{(\gamma M)^r}
    \right]^{1/r} \, = \, \frac{1}{M}\left[\frac{M-K_n}{\alpha}\right]^{1/r}.
    \end{split}
\end{equation*}

We now turn our attention to $U^{\pityor}$ in \eqref{eqn:PD_UB_species}. Once again, the limiting behavior when $\sigma\to 1$ follows from the continuity of the rising factorial operator. Instead, fixing $\sigma\in(0,1)$ and letting $\theta\to0$, we have that, by continuity, $\theta+K_n\sigma \to K_n\sigma$, and also
$ (n+\theta)_r = \prod_{j=0}^{r-1}(n+\theta+j) \to \prod_{j=0}^{r-1}(n+j) = (n)_r$.
Therefore
\begin{equation*}
    \label{eqn:species_UB_PD_lim0}
    \lim_{\theta \to 0} 
    \left[ \frac{1}{\alpha}\, \frac{ (\theta + K_n\sigma)\, (1-\sigma)_{r-1}} {(n + \theta)_{r}} \right]^{1/r}
    \, = \,
    \left[\frac{K_n}{\alpha (n)_r} \, \sigma \, (1-\sigma)_{r-1}\right]^{1/r}\,.
\end{equation*}

On the other hand, for $\theta\to\infty$, we use the same approximation as the one considered for the FD case when $\gamma \to \infty$  to get
\begin{equation*}
    \label{eqn:species_UB_PD_liminf}
    \lim_{\theta \to \infty} 
    \left[ \frac{1}{\alpha}\, \frac{ (\theta + K_n\sigma)\, (1-\sigma)_{r-1}} {(n + \theta)_{r}} \right]^{1/r}
    \, = \,
    \lim_{\theta \to \infty} \left[
    \left(\frac{(1-\sigma)_{r-1}}{\alpha}\right)^{1/r} \, \theta^{-(1-\frac{1}{r})}
    \right] = 0
    \,,
\end{equation*}
thereby concluding the derivations for the species setting.
\vspace{5pt}

\subsection{\large B.2 Features setting}
\label{app:UB_limits_fea}
Consider  first $U^{\operatorname{MBP}}$ in \eqref{eqn:UB_GeneralMBP} parametrized through the expected value of the 
Beta~prior,~as~clarified in  Section \ref{section:species_results}. For any order $r>1$, 
its limiting behavior for $\omega\to0$ and $\omega\to1$, keeping fixed~$\eta>1$, follows from the continuity of the rising factorial operator. This implies
\begin{equation*}
    \begin{split}
        &\lim_{\omega \to 0} \left[\frac{\mean{M^\star}}{\alpha}\,\frac{\left(\omega(\eta-1)\right)_r}{\left(\eta-1+n\right)_r}\right]^{1/r} = 0\,,\\
        &\lim_{\omega \to 1} \left[\frac{\mean{M^\star}}{\alpha}\,\frac{\left(\omega(\eta-1)\right)_r}{\left(\eta-1+n\right)_r}\right]^{1/r} = 
        \left[\frac{\mean{M^\star}}{\alpha}\,\frac{\left(\eta-1\right)_r}{\left(\eta-1+n\right)_r}\right]^{1/r}
        \,.
    \end{split}
\end{equation*}
Similarly, its monotonicity with respect to $\omega$ is due to the one of the rising factorial.
More~interestingly, we show that it is  also strictly increasing with respect to $\eta$ for at least any fixed value~of $r$ smaller than $1+n\,\omega/(1-\omega)$.

Since the \(1/r\)-power and the positive constant \(\mean{M^\star}/\alpha\) preserve monotonicity, it is enough to study
$h(x)=(\omega x)_r/(x+n)_r$, for $x=\eta-1>0$.
Using the product representation of the rising factorial and applying the $\log$ transformation, we obtain
$$
\log h(x) =\sum\nolimits_{j=0}^{r-1}\log(\omega x+j) - \log(x+n+j)\,,
\quad
\frac{d}{dx}\log h(x) = \sum\nolimits_{j=0}^{r-1}\left( \frac{\omega}{\omega x+j}  - \frac{1}{x+n+j} \right).
$$
For \(j=0\), the first term is $1/x-1/(x+n)$, which is always positive. 
Instead, for \(j\ge 1\),
\[
\frac{\omega}{\omega x+j} \ge \frac{1}{x+n+j} \quad 
\iff 
\omega n\ge (1-\omega)j.
\]
Hence every summand is nonnegative whenever
\[
r\le 1+\omega n/(1-\omega).
\]
Under this condition, $\log h(x)$ is strictly increasing, which proves our result.
Its behavior remains unclear for larger $r$, but such a regime is not expected to occur in practice as $r$ is usually much smaller than $n$, typically even smaller than $\log(n)$.

Turning our attention to $U^{\operatorname{IBP}}$ in \eqref{eqn:UB_IBP}, we prove that it is monotonically decreasing~with~respect to $c_n=c+n$ and increasing with respect to $\sigma$ for a vast range of $r$ values. 
As before, since the \(1/r\)-power preserves monotonicity, it is enough to study
$B(r-\sigma,c_n+\sigma)$. Then, taking the derivative of $\log B(r-\sigma,c_n+\sigma)$ with respect to $c_n$ we have that
\[
\frac{\partial}{\partial c_n}
\log B(r-\sigma,c_n+\sigma)
=
\psi(c_n+\sigma)-\psi(r+c_n) < 0\,,
\]
where \(\psi(\cdot)\) is the digamma function, which is increasing on \((0,\infty)\), yielding to a negative derivative since
\(r>\sigma\). Thus, the upper bound is decreasing in \(c_n\). By taking~the~partial~derivative~with~respect to $\sigma$, we instead get that 
\[
\frac{\partial}{\partial \sigma}
\log B(r-\sigma,c_n+\sigma)
=
-\psi(r-\sigma)+\psi(c_n+\sigma),
\]
which is positive whenever
$
c_n+\sigma>r-\sigma
$, or, alternatively, $\sigma>(r-c_n)/2$.

\vspace{10pt}

\section{\large C. Further Details on Hyperparameter Elicitation}
\label{app:hy_elicitation}
Recalling Section~\ref{sec:comp_and_infer}, the hyperparameters $\bm{\phi}$ are estimated via an empirical Bayes strategy solving the optimization problem in \eqref{eqn:species_EB}. This is equivalent to computing the MAP estimator under~the Bayesian model having likelihood  $\Pi(n_1,\ldots,n_{K_n}\mid \bm{\phi})$ and prior $\Lcr(\bm{\phi};v)$. Sections \ref{app:hy_elicitation_species} and \ref{app:hy_elicitation_features} provide details on both quantities in the species and features settings, respectively.

\vspace{5pt}
\subsection{\large C.1 The species setting}
\label{app:hy_elicitation_species}
Focusing first on the likelihood, we note that the random probability measure $\mu$ can be integrated out from the hierarchical model in \eqref{eqn:Definetti_model}, yielding the marginal distribution of the observations. By \cite{pitman1996}, such a distribution factorizes into two independent components. One involving~the base probability measure $P_0$ evaluated at the observed species labels, and the other referring to the observed species frequencies. As clarified within Section~\ref{sec:species}, the former is not of interest~here since the precise species labels are irrelevant  in our analysis, serving merely as identifiers.  Consequently, the main object of interest is the distribution of the species abundances, commonly referred to as the Exchangeable Partition Probability Function (EPPF). Such a function must be symmetric in its arguments \citep{pitman1996}, and is commonly denoted by $\Pi(n_1,\ldots,n_{K_n}\mid \bm{\phi})$. Below we report the specific forms of the EPPFs for the prior choices of interest.

Under the finite Dirichlet prior, the marginal distribution for the complete vector~of~species prevalences is the one of a Dirichlet--Multinomial. Since the alphabet size $M$ is known in this~setting, the distribution  is evaluated on the full frequency vector $(n_1, \ldots, n_M)$, including zero counts, unlike a proper EPPF, which depends only on the counts of the observed species. Specifically
\begin{equation*}
    \label{eqn:species_Dir_eppf}
    \Pi^{(\operatorname{FD})}(n_1,\ldots,n_M\mid \bm{\phi}) \, = \, 
    \frac{n!}{\prod_{j=1}^M n_j!}\,\frac{1}{(\gamma M)_n}\,\prod_{j=1}^M \, (\gamma)_{n_j}\,,    
\end{equation*}
where $\sum_{j=1}^M n_j=n$, whereas $(x)_n=x(x+1)\ldots(x+n-1)$ is the rising factorial of order $n$.

As shown in \citet{deblasi2015} and \citet{argiento2022annals}, under the FDP prior the corresponding EPPF is instead
\begin{equation*}
    \label{eqn:species_FDP_eppf}
    \Pi^{(\operatorname{FDP})}(n_1,\ldots,n_{K_n}\mid \bm{\phi}) \, = \, 
    V_{n}^{K_n}\,\prod_{j=1}^{K_n} \, (\gamma)_{n_j}\,,   
    \quad\text{where}\quad
    V_{n}^{K_n} = \sum_{m=K_n}^\infty \frac{m!}{(m-K_n)!}\,q_M(m)\,\frac{1}{(\gamma m)_n}.
\end{equation*}

Finally, by the classical results of \citet{PitmanYor97Annals}, the EPPF induced by the Pitman--Yor (PYP) process is
\begin{equation*}
    \label{eqn:species_PD_eppf}
    \Pi^{(\pityor)}(n_1,\ldots,n_{K_n}\mid \bm{\phi}) \,=\,
    \begin{cases}
        \displaystyle
        \frac{\sigma^{K_n-1}}{(\theta+1)_{n-1}}\,
        \frac{\Gamma(\theta/\sigma + K_n)}{\Gamma(\theta/\sigma + 1)}\,
        \prod_{j=1}^{K_n} (1-\sigma)_{n_j-1}\,,
        & \text{if } \sigma > 0, \\[1em]
        \displaystyle
        \frac{\theta^{K_n}}{(\theta)_{n}}\,
        \prod_{j=1}^{K_n} (n_j-1)!\,,
        & \text{if } \sigma = 0\,,
    \end{cases}
\end{equation*}
For $\sigma=0$, the $\pityor$ degenerates to the Dirichlet process \citep{ferguson73} and the corresponding EPPF coincides with the Ewens sampling formula \citep{Ewens1972}.

We now turn our attention to the specific choices of $\Lcr(\bm{\phi};v)$, which depend on the prior regime under consideration. Under the finite Dirichlet prior, we have that $\phi=\gamma$, which motivates the standard hyperprior
\[
\gamma \sim \operatorname{Gamma}(1/v,\,1/v).
\]
This choice implies $E[\gamma]=1$ and $\operatorname{Var}(\gamma)=v$. Consequently, larger values of $v$ correspond~to~more diffuse priors and therefore to a weaker penalization in the MAP optimization.

Under the FDP prior, where $q_M$ denotes the probability mass function of a one-shifted Poisson with rate $\Lambda$, the hyperparameters are $\bm{\phi}=(\Lambda,\gamma)$. For $\gamma$ we consider the same  hyperprior as~the~one employed for the finite Dirichlet case, whereas for $\Lambda$ we leverage default prior elicitations for the rate of a Poisson and assume 
\[
\Lambda \sim \operatorname{Gamma}(K_n^2/v,\,K_n/v).
\]
The above choice centers $\Lambda$ at $K_n$ and let $v$ regulate its variability since $\operatorname{Var}(\Lambda)=v$. 

Finally, as for the FDP, also under the $\pityor$ prior $\bm{\phi}$ is a bivariate random vector $\bm{\phi}=(\sigma,\theta)$,~for which we again assume independent components. Specifically, leveraging classical hyperpriors for the $\pityor$ parameters, we consider
\[
\sigma \sim \operatorname{Beta}(a_{\sigma},b_{\sigma}),
\qquad
\theta \sim \operatorname{Gamma}(a_v,b_v),
\]
where $a_{\sigma}$ and $b_{\sigma}$ are set equal to one, yielding a uniform distribution on $(0,1)$, while $a_v$~and~$b_v$~are chosen so that $\operatorname{Var}(\theta)=v$ and the prior mean of $K_n$, evaluated at $\mathbb E[\sigma]$~and~$\mathbb E[\theta]$,~matches~the~observed number of distinct species.

\vspace{10pt}
\subsection{\large C.2 The feature setting}
\label{app:hy_elicitation_features}
As discussed in Section~\ref{sec:comp_and_infer}, unlike the species setting, we did not find it necessary to introduce any penalization term in the features case. Therefore, in this context, hyperparameter estimation~proceeds through the maximization of the log-marginal likelihood only. As shown by \citet{Broderick13}, such an optimization closely parallels the one we described in the previous section for the species setting. Specifically, the marginal likelihood factorizes again into the product of a term depending only on the labels of the alphabet symbols --- which we discard here as well --- and~a function of the observed feature frequencies, also known as the Exchangeable~Feature~Probability Function (EFPF). The EFPF is the analogue, for feature allocation models, of the Exchangeable Partition Probability Function (EPPF) from the species setting, and it is symmetric~in~its~arguments. The prior constructions we consider belong to the class of product feature allocation models introduced by \citet{Ghilotti2026}. Therefore, the corresponding EFPFs can be derived inheriting the results from this general framework.

Under the finite Beta prior the alphabet size $M$ is known. Therefore, as for the FD case in the species setting, also in this context  the predictive distribution for the complete vector~of~features frequencies is evaluated on the entire configuration $(n_1, \ldots, n_M)$, including zero observed  counts. This departs from a proper EFPF, which is a function of the counts of observed features only,~and yields
\begin{equation*}
    \label{eqn:fea_FB_efpf}
    \Pi^{(\operatorname{FB})}(n_1,\ldots,n_M) \, = \, 
    \frac{1}{B(a,b)^M}
    \,\prod_{j=1}^M \, \binom{n}{n_j}\,B(n_j+a,\,n-n_j+b)\,,    
\end{equation*}
where $B(a,b)$ is the Beta function.

Under a mixed Beta processes with a Negative Binomial NegBin($m,q$)  prior for the number~of labels, the corresponding EFPF is instead
\begin{equation*}
    \label{eqn:fea_MBP_efpf}
    \begin{split}
        &\Pi^{(\operatorname{N-MBP})}(n_1,\ldots,n_{K_n}) \\
        &\quad = \, 
        \frac{\Gamma(K_n+m)}{\Gamma(K_n+1)\Gamma(m)}\,
        \frac{(1-q)^{m}}{\left( 1 - q\kappa_n\right)^{K_n+m}}
        \left[ \frac{q}{(a+b)_n\Gamma(a)\Gamma(b)}\right]^{K_n}\,
        \prod_{j=1}^{K_n}\left\{ \Gamma(n_j+a)\,\Gamma(n-n_j+b) \right\}\,,
    \end{split}
\end{equation*}
where we recall that $\kappa_n = (b)_n/(a+b)_n$. We reduce the optimization space by fixing $q = 1 - 1/v$ for some prespecified positive value of $v$, usually $v = 500$ as in \cite{Ghilotti2026}.

Finally, standard results for the three-parameter Indian Buffet Process \citep[see, e.g.,][]{Ghilotti2026} yield the EFPF
\begin{equation*}
    \label{eqn:fea_IBP_efpf}
    \begin{split}
        \Pi^{(\operatorname{IBP})}(n_1,\ldots,n_{K_n}) \, = \,
        \frac{\gamma^{K_n}\,e^{ - \gamma \varphi_n}}{K_n!}\,
        \prod_{j=1}^{K_n}\frac{B(n_j-\sigma, \, n-n_j+c+\sigma)}{B(1-\sigma,c-\sigma)}\,,
    \end{split}
\end{equation*}
where 
$
\varphi_n = (\Gamma(c+1)/\Gamma(c+\sigma))\sum_{i=1}^n (\Gamma(c+\sigma+i-1)/\Gamma(c+i)).
$

\bibliographystyle{abbrvnat}
\bibliography{references.bib}

\end{document}